\documentclass{acmart}

\usepackage{etex}
\usepackage[utf8]{inputenc}
\usepackage[T1]{fontenc}
\usepackage{microtype}
\DeclareMathAlphabet{\mathcal}{OMS}{cmsy}{m}{n}
\usepackage{caption}
\usepackage{tikz}
\usepackage{forest}
\usepackage{color}
\usepackage{standalone}
\usepackage{subcaption}
\usetikzlibrary{positioning,chains,fit,shapes,calc,arrows.meta,patterns}
\colorlet{darkgreen}{green!50!black}
\renewcommand{\thesection}{\Roman{section}}

\newlength{\limage}
\newlength{\rimage}
\newlength{\rrimage}

\usepackage{amsmath,  graphicx}
\usepackage{url}
\usepackage{xspace}
\usepackage{paralist}
\usepackage{graphicx}
\usepackage{tabularx}
\usepackage{booktabs}
\usepackage{amsthm}
\usepackage{bm}

\makeatletter
\newcommand*{\defeq}{\mathrel{\rlap{\raisebox{0.3ex}{$\m@th\cdot$}}\raisebox{-0.3ex}{$\m@th\cdot$}}=}
\makeatother

\newcommand{\kw}[1]{{\textsf{#1}}\xspace}
\newcommand{\danger}{\kw{Dng}}

\newcommand{\nondanger}{\kw{NDng}}
\newcommand{\positions}{\kw{Pos}}
\newcommand{\pos}{\positions}

\newcommand{\full}{\kw{full}}

\newcommand{\appelem}[2]{\kw{Wants}(#1, #2)}
\newcommand{\appelemb}[3]{\kw{Wants}_{#1}(#2, #3)}
\newcommand{\chase}[2]{\kw{Chase}(#1, #2)}
\newcommand{\cov} {\kw{sim}}

\newcommand{\fragcup}{\wedge}
\newcommand{\formcup}{\cup}

\newcommand{\bsim}{\leq}
\newcommand{\bbsim}{\simeq}

\newcommand{\cc}{\mathrm{c}}

\newcommand{\f}{\mathrm{f}}
\newcommand{\ii}{\mathrm{i}}
\renewcommand{\ll}{\mathrm{l}}
\newcommand{\m}{\mathrm{m}}

\newcommand{\p}{\mathrm{p}}
\renewcommand{\t}{\mathrm{t}}

\newcommand{\eqfun}{\leftrightarrow_{\mathrm{FUN}}}

\newcommand{\eqids}{\sim_{\mathrm{\incd}}}
\newcommand{\eqidsclass}[1]{[#1]_{\mathrm{\incd}}}

\newcommand{\eqrev}{\eqids}

\newcommand{\ui}[2]{#1 \subseteq #2}

\newcommand{\ovl}{\kw{OVL}}
\newcommand{\overlap}{\ovl}

\newcommand{\myeat}[1]{}

\newcommand{\calA}{\mathcal{A}}

\newcommand{\calE}{\mathcal{E}}
\newcommand{\calF}{\mathcal{F}}

\newcommand{\calH}{\mathcal{H}}

\newcommand{\calN}{\mathcal{N}}

\newcommand{\calY}{\mathcal{Y}}
\newcommand{\calZ}{\mathcal{Z}}

\newcommand{\NN}{\mathbb{N}}
\newcommand{\ZZ}{\mathbb{Z}}

\newcommand{\fd}{\kw{FD}}
\newcommand{\uid}{\kw{UID}}

\newcommand{\ufd}{\kw{UFD}}
\newcommand{\incd}{\kw{ID}}
\newcommand{\cq}{\kw{CQ}}

\newcommand{\acq}{\kw{ACQ}}

\DeclareMathOperator{\id}{id}

\DeclareMathOperator{\pr}{\mathsf{pr}}
\DeclareMathOperator{\parts}{\mathfrak{P}}

\DeclareMathOperator{\dom}{dom}

\newcommand{\entfin}{\models_\mathrm{fin}}
\newcommand{\entunr}{\models_\mathrm{unr}}

\newcommand\np{\textup{NP}\xspace}

\newcommand{\fds}{\Sigma_{\mathrm{FD}}}
\newcommand{\ufds}{\Sigma_{\mathrm{UFD}}}
\newcommand{\ids}{\Sigma_{\mathrm{UID}}}
\newcommand{\idsb}{\ids}
\newcommand{\idsr}{\ids^{\mathrm{rev}}}
\newcommand{\con}{\Sigma}
\newcommand{\ucon}{\Sigma_{\mathrm{U}}}
\newcommand{\uconb}{\Sigma_{\mathrm{U}}}
\newcommand{\uconr}{\Sigma^{\mathrm{rev}}_{\mathrm{U}}}

\newcommand{\idprec}{\rightarrowtail}
	
\newcommand\restr[2]{{
  \kern-\nulldelimiterspace 
  #1 
  _{|#2} 
  }}

\newcommand{\fdrestr}[2]{#1^{#2}}
\newcommand{\relrestr}[2]{\restr{#1}{#2}}

\newcommand{\quot}[2]{#1/{#2}}

\newcommand{\card}[1]{\left|#1\right|}

\newcommand{\mapb}{\lambda}

\newcommand{\deft}[1]{\textbf{#1}}

\newcommand{\defo}[1]{\emph{#1}}

\newcommand{\defp}[1]{\emph{#1}}

\newcommand{\arity}[1]{\card{#1}}

\newcommand{\invar}[1]{{\selectfont\textsf{#1}}}

\newcommand{\neqfunc}{n}
\newcommand{\neqidsc}{n}

\newcommand{\rfcl}{\afactcl}
\newcommand{\afactcl}{\kw{AFactCl}}

\newcommand{\lab}[1]{\Lambda(#1)}

\newcommand{\sprod}{\otimes}
\newcommand{\mprod}{\otimes}

\newcommand{\pire}{PI}

\newcommand{\mybf}[1]{\text{\textit{\textbf{#1}}}}

\newcommand{\is}[1]{\makebox[0pt]{\(\scriptstyle#1\)}}

\newcommand{\dense}{\mathrm{dense}}
\newcommand{\pre}{\mathrm{pre}}
\newcommand{\ach}{\mathrm{ach}}

\newcommand{\closure}[1]{#1^*}
\newcommand{\closuref}[1]{#1^{\mathrm{f}*}}

\newcommand{\cl}{\kw{CL}}

\newcommand{\NP}{$\mathrm{NP}$\xspace}
\newcommand{\ACZ}{$\mathrm{AC}^0$\xspace}

\newcommand{\factive}{\mathrm{active}}
\newcommand{\fnew}{\mathrm{new}}
\newcommand{\w}{\mathrm{witness}}
\newcommand{\rr}{\mathrm{reuse}}

\newcommand{\apxref}[1]{Appendix~\ref{#1}}
\newcommand{\secref}[1]{Section~\ref{#1}}
\newcommand{\secrefs}[2]{Sections~\ref{#1} and~\ref{#2}}
\newcommand{\secname}{section\xspace}
\newcommand{\subsecname}{subsection\xspace}
\newcommand{\secnames}{sections\xspace}

\newcommand{\secshortname}{Section}
\newcommand{\thepaper}{the paper\xspace}
\newcommand{\thispaper}{this paper\xspace}
\newcommand{\Thispaper}{This paper\xspace}

\newcommand{\myparagraph}[1]{\paragraph{\textbf{#1}}}

\newenvironment{maintheorem}{\begin{theorem}}{%
    \end{theorem}\ignorespacesafterend
}
\newenvironment{mainproposition}{\begin{proposition}}{%
    \end{proposition}\ignorespacesafterend
}
\newenvironment{mainexample}{\begin{example}}{%
    \end{example}\ignorespacesafterend
}
\newcommand{\proofof}{Proof of\xspace}

\newcommand{\mydefthm}[4]{
  \begin{restatable}[#3]{#2}{#1}
    \label{thm:#1}
    #4
  \end{restatable}
}

\newcommand{\qedef}{}
\newcommand{\descindent}{\null \hspace{13pt} \null}

\title{Finite Open-World Query Answering with Number Restrictions}

\author{Antoine Amarilli}
\affiliation{
\institution{LTCI, Télécom Paris, Institut Polytechnique de Paris}
\city{Paris}
\country{France}
}
\author{Michael Benedikt}
\affiliation{%
  \institution{The University of Oxford}
  \city{Oxford}
  \country{United Kingdom}
}

\setcopyright{acmcopyright}
\acmJournal{TOCL}
\acmYear{2019} \acmVolume{1} \acmNumber{1} \acmArticle{1} \acmMonth{1} \acmPrice{15.00}\acmDOI{10.1145/3365834}
\begin{abstract}
Open-world query answering  is the problem of deciding, given a set of facts,
conjunction of constraints, and query, whether the 
facts and constraints imply the  query.
This amounts to 
reasoning over all instances that include the facts and satisfy the constraints.
We study \emph{finite open-world query answering} (FQA), which assumes that the underlying world
is finite and thus only considers the \emph{finite}
completions of the instance.
The major known decidable cases of FQA
  derive from the following: the guarded fragment of first-order logic, which can express referential constraints (data in one
place points to data in another) but cannot express number restrictions such
as functional dependencies; and the guarded fragment with number restrictions but
on a signature of arity only two.
  In this paper, we give the first decidability results for FQA that
 combine both referential constraints and number
restrictions
 for arbitrary
signatures: we show that, for unary inclusion dependencies and functional
dependencies, the finiteness assumption of FQA can be lifted up to taking the finite implication closure of
the dependencies.
Our result relies on new techniques to construct finite universal models of such
constraints, for any bound on the maximal query size.
\end{abstract}

\mathchardef\breakingcomma\mathcode`\,
{\catcode`,=\active
    \gdef,{\breakingcomma\discretionary{}{}{}}
}

\usepackage{thm-restate}
\usepackage{thm-restate-fix}

\begin{document}

\maketitle

\section{Introduction}
\label{sec:intro}
A longstanding goal in computational logic is to design logical languages that are both decidable and expressive.
One approach is to distinguish integrity constraints and queries, and have
separate languages for them.
We would then seek  
decidability of the \emph{query answering with constraints} problem:
given a query $q$,  a conjunction of constraints $\con$, and a finite
instance~$I_0$,
determine which answers
to~$q$ are certain to hold over any instance~$I$ that extends $I_0$ and
satisfies $\con$.
This problem is often called \emph{open-world query answering}. It is 
fundamental for deciding
query
containment under constraints, querying in the presence of ontologies, or
reformulating queries with
constraints. Thus
 it has been the subject
of intense study within several communities for decades
(e.g. \cite{johnson1984testing,cali2003decidability,barany2014querying,pratt2009data,ibanezgarcia2014finite}).

In many cases (e.g., in databases)
 the instances $I$ of interest  are the finite ones, and hence we can define
 \emph{finite open-world query answering} (denoted here as FQA),
which restricts the quantification
to \emph{finite} extensions  $I$ of~$I_0$. In contrast, by \emph{unrestricted open-world query answering} (UQA) we refer to
the problem where $I$ can be either finite or infinite.
Generally the class of queries is taken to be the conjunctive queries ($\cq$s) --- queries built up
from relational atoms via existential quantification and conjunction. We will restrict to $\cq$s here, and
thus omit explicit mention of the query language, focusing on the constraint language.

A first constraint class known to have tractable open-world query answering problems
are \emph{inclusion dependencies} ($\incd$s) --- constraints of the form, e.g.,
$\forall x y z ~ R(x, y, z) \rightarrow \exists v w ~ S(z, v, w, y)$.
The fundamental results of Johnson and Klug \cite{johnson1984testing} and  Rosati \cite{rosatifinitecontrol}
show that both FQA and UQA are decidable for~$\incd$ and that, in fact, they coincide. 
When FQA and UQA coincide,  the constraints
are said to be \emph{finitely controllable}.
These results have been generalized by \cite{barany2014querying}
to a much richer class of constraints,  the guarded fragment of first-order logic.

However, those results do not cover a second 
important kind of constraints, namely \emph{number restrictions}, which express, e.g.,
uniqueness. We represent them by the class of
\emph{functional dependencies} ($\fd$s) --- 
of the form
$\forall \mybf{x} \mybf{y} ~ (R(x_1, \ldots, x_n) \wedge R(y_1, \ldots, y_n) \wedge
\bigwedge_{i\in L} x_i = y_i) \rightarrow x_r = y_r$ for some set $L$ of indices
and some index~$r$.
The implication problem (does one $\fd$ follow from a set of $\fd$s) is
decidable for $\fd$s, and coincides with implication
restricted to finite instances \cite[Theorem~9.2.3]{abiteboul1995foundations}. Trivially, the FQA and UQA problems
are also decidable for $\fd$s alone, and coincide.

$\incd$s require the model to contain some elements, while $\fd$s restrict the ability to add elements.
The interaction is severe enough that trying to combine $\incd$s and $\fd$s makes both UQA and FQA undecidable in
general \cite{cali2003decidability}. Some progress has been made
on obtaining decidable cases for UQA.
UQA 
is known to be decidable when the $\fd$s and the $\incd$s are
\emph{non-conflicting}
\cite{johnson1984testing,cali2003decidability}. 
Intuitively, this condition guarantees that the $\fd$s can be ignored, as long as they
hold on the initial instance~$I_0$, and one can then solve the query answering
problem by considering the $\incd$s alone.
But the non-conflicting condition 
only 
applies to UQA and not to FQA.
In fact it is known that even for very simple classes of~$\incd$s and $\fd$s,  including
non-conflicting classes, 
FQA and UQA do not coincide.
Rosati \cite{rosatifinitecontrol} showed that FQA is  undecidable 
 for non-conflicting $\incd$s and $\fd$s: this is already the case for $\incd$s
 and keys, i.e., for the special case of $\fd$s that specify that some relations
 are determined by a subset of their attributes.

Thus a  broad question is to what extent these classes,
$\fd$s and $\incd$s,  can be combined while retaining decidable FQA.
The only  decidable cases impose very
severe requirements. For example,
for the specific case of key dependencies (KDs) and foreign keys (FKs),
the constraint class of ``single KDs and FKs'' introduced
in \cite{rosatifinitecontrol} has decidable FQA, but
such constraints cannot model, e.g., $\fd$s which are not keys.
Further, in contrast with the general case of $\fd$s and $\incd$s,
single KDs and FKs are always finitely controllable, which limits their
expressiveness. 
Indeed, \emph{we know of no tools to deal with FQA for non-finitely-controllable constraints  on relations of arbitrary arity.}

A second decidable case is  when all relations
and all subformulae of the constraints have arity at most two. In this context,
results of Pratt-Hartmann \cite{pratt2009data} imply the
decidability of both FQA and UQA
for a very rich non-finitely-controllable sublogic of first-order logic.
For some fragments of this arity-two logic, the complexity of FQA has recently
been isolated by Ib{\'a\~n}ez-Garc{\'i}a et al.
 \cite{ibanezgarcia2014finite}, 
 and some extensions have been proposed to description logics with transitive
 roles \cite{gogacz2018finite,gogacz2019finite}. Yet these results do not apply to arbitrary
 arity signatures.

 \emph{The contribution of this paper is to provide the first result
about
finite query
answering for non-finitely-controllable $\incd$s and $\fd$s over relations of
arbitrary arity.}
As the problem is undecidable in general, we must naturally make some restriction. Our choice is to limit to 
\emph{Unary $\incd$s} ($\uid$s), which export
only one variable: for instance, $\forall x y z ~ R(x, y, z) \rightarrow \exists w ~ S(w, x)$.
$\uid$s and $\fd$s are an interesting class to study because they are not
finitely controllable, and allow the
modeling, e.g., of single-attribute foreign keys, a common use case in database
systems. In contrast, Johnson and Klug \cite{johnson1984testing} showed that $\uid$s in isolation
are finitely controllable.
The decidability of UQA for $\uid$s and $\fd$s is known because they
are always non-conflicting. In this paper, we show that
finite query answering is decidable for $\uid$s and $\fd$s, and obtain tight bounds on
its complexity.

The idea is to \emph{reduce the finite
case to the unrestricted case}, but in a more complex way than by finite
controllability.
We make use of a technique
originating in Cosmadakis et al. \cite{cosm} to study finite implication on $\uid$s and $\fd$s:
the 
\emph{finite closure} operation
which takes
 a conjunction of~$\uid$s and $\fd$s and determines exactly which additional
 $\uid$s and $\fd$s are
implied
over finite instances. 
Rosati \cite{rosatieswc} and Ib\'a\~nez-Garc\'ia et al. \cite{ibanezgarcia2014finite}  make use of the closure
operation in their study of constraint classes over schemas of arity two. They
show that
finite query answering for a query $q$, instance~$I_0$, and constraints $\con$ reduces to 
unrestricted query answering for~$I_0$, $q$, and the finite
closure~$\closuref{\con}$ of
$\con$. In other words, the closure construction which is sound for implication is also sound for query answering.

We show that the same approach applies to arbitrary arity signatures,
with
constraints being $\uid$s and $\fd$s. Our main result thus reduces finite query answering to unrestricted query answering,
for $\uid$s and $\fd$s in arbitrary arity:

\begin{theorem}
  \label{thm:main2}
  For any finite instance~$I_0$,
  CQ~$q$,  
  and constraints $\con$ consisting of~$\uid$s and $\fd$s,
 the finite open-world query answering problem for $I_0, q$ under $\con$ has the same answer
as the unrestricted open-world query answering problem for $I_0, q$ under the
finite closure of~$\con$.
\end{theorem}

Using the known results about the complexity of UQA for $\uid$s,
we isolate the precise complexity
of finite query answering  with respect to $\uid$s and $\fd$s, showing that it matches
that of UQA:

\begin{corollary} \label{cor:comp} The combined complexity of the finite open-world query answering problem for
  $\uid$s and $\fd$s is $\np$-complete, and it is PTIME in data complexity (that
  is, when the constraints and query are fixed).
\end{corollary}

Our proof of Theorem~\ref{thm:main2} is quite involved, since building finite models
that satisfy number restrictions and inclusion dependencies in a signature with
arbitrary arity 
introduces a multitude of new difficulties
that do not arise in
the arity-two case or in the case of $\incd$s in isolation.

We borrow and adapt a variety of techniques from prior work: \begin{itemize}
\item using $k$-bounded
simulations to preserve small acyclic $\cq$s \cite{ibanezgarcia2014finite}, 
\item partitioning
$\uid$s into components that have limited interaction, and satisfying the $\uid$s component-by-component   \cite{cosm,ibanezgarcia2014finite},
\item performing a chase that reuses sufficiently similar
elements \cite{rosatifinitecontrol}, 
\item  taking the product with groups of large
girth to blow up cycles \cite{otto2002modal}.
\end{itemize}
However, we must also develop some  new infrastructure to deal with number restrictions
in an arbitrary arity setting: distinguishing between so-called \defo{dangerous}
and \defo{non-dangerous} positions when creating a new element to satisfy some
$\incd$s, constructing realizations for
relations in a \emph{piecewise} manner following the $\fd$s, reusing elements
in a \emph{combinatorial} way that shuffles them to avoid violating the
higher-arity $\fd$s,
and a new notion of \defo{mixed product} to blow cycles
up while preserving fact overlaps to avoid violating the higher-arity $\fd$s.

\myparagraph{Paper structure}
The overall structure of the proof, presented in \secref{sec:overall}, is to extend
a given instance $I_0$ to a finite model of~$\uid$ and $\fd$ constraints such that for every conjunctive query of size at most $k$,
the model satisfies it only when it is implied by $I_0$ and the constraints.
We call these \emph{$k$-universal superinstances}. It is easy to show that if a
$k$-universal superinstance exists for an instance
and set of constraints, then finite implication and implication of CQs coincide.

We start with only \emph{unary} $\fd$s
($\ufd$s) and \emph{acyclic} $\cq$s ($\acq$s), and by assuming that the $\uid$s
and $\ufd$s are \emph{reversible}, a condition inspired by the 
finite closure construction.

As a warm-up, \secrefs{sec:binary}{sec:weakarb}
approximate even further by replacing $k$-universality by a weaker notion,
proving the corresponding result starting with binary signatures and generalizing to arbitrary arity.
We extend the result to $k$-universality in \secref{sec:ksound}, 
maintaining a $k$-bounded simulation to the chase, and performing \emph{thrifty} chase steps that reuse
sufficiently similar elements without violating $\ufd$s. We also
rely on a structural observation about the chase under $\uid$s
(Theorem~\ref{thm:similarityb}).
\secref{sec:manyscc} eliminates the assumption that dependencies are
reversible, by
partitioning the $\uid$s into classes that are either reversible or trivial,
and satisfying successively each class following a certain ordering.

We then generalize our result to higher-arity (non-unary) $\fd$s in
\secref{sec:hfds}.
This requires us to define a new notion of thrifty chase steps
that apply to instances with many ways to reuse elements; the
existence of these instances relies on a combinatorial construction of models
of $\fd$s with a high number of facts but a small domain (Theorem~\ref{thm:combinatorial}).
Last, in \secref{sec:cycles}, we apply a cycle blowup process to the result
of the previous constructions, to go from acyclic to arbitrary $\cq$s through a
product with acyclic groups.
The technique is inspired by Otto \cite{otto2002modal} but
must be adapted to respect $\fd$s. 

The paper is an extended version of the conference paper \cite{amarilli2015finite}.

\section{Background} 
\label{sec:prelim}
\subsection{Instances and Constraints} \label{subsec:instanceconstraint}

\myparagraph{Instances}
We assume an infinite countable set of \defp{elements} (or \defp{values}) $a,
b, c, \ldots$ and \defp{variable names} $x, y, z, \ldots$. 
We will often write tuples of elements with boldface as $\mybf{a}$ and denote
the $i$-th element of the tuple by~$a_i$, and likewise for tuples of variables.
A \defp{schema}~$\sigma$ consists of \defp{relation names} (e.g., $R$) with an 
\defp{arity} (e.g., $\arity{R}$)
which we assume is $\geq 1$: we write $\arity{\sigma} \defeq \max_{R \in \sigma}
\arity{R}$. Each relation~$R$ defines a set of $\arity{R}$ \defp{positions} that we write
$\positions(R) \defeq \{R^1, \ldots, R^{\arity{R}}\}$. 
For convenience, given a set $L \subseteq \{1, \ldots, \arity{R}\}$, we will
write $R^L$ to mean $\{R^l \mid l \in L\}$.
We also define
$\positions(\sigma) \defeq \bigsqcup_{R \in \sigma} \pos(R)$, where $\sqcup$
denotes disjoint union.
We will identify $R^i$ and $i$ when no confusion can result.

A relational \defp{instance}
$I$ of~$\sigma$ 
is a set of \defp{facts} of the form $R(\mybf{a})$ where $R$ is a
relation name and $\mybf{a}$ an $\arity{R}$-tuple of values.
The \defp{size} $\card{I}$ of a finite instance $I$ is its number of facts.
The
\defp{active domain} $\dom(I)$ of~$I$ is
the set of the elements which appear in some fact of~$I$. For any position $R^i \in
\positions(\sigma)$, we define the \defp{projection} $\pi_{R^i}(I)$ of~$I$ to
$R^i$ as the
set of the elements of~$\dom(I)$ that occur at position $R^i$ in some 
fact of~$I$.
For $L \subseteq \{1, \ldots, \arity{R}\}$,
the \defp{projection}~$\pi_{R^L}(I)$ is a set of~$\card{L}$-tuples defined analogously;
we will often index those tuples by the positions in~$L$ rather than by $\{1, \ldots,
\card{L}\}$.
A \defp{superinstance} 
 of~$I$ is a (not necessarily finite) instance $I'$ such
that $I \subseteq I'$. 

A \defp{homomorphism} from an instance $I$ to an instance $I'$ is a mapping
$h : \dom(I) \rightarrow \dom(I')$ such that, for every fact $F = R(\mybf{a})$
of $I$, the fact $h(F) \defeq R(h(a_1), \ldots, h(a_{\card{R}}))$ is in~$I'$.

\myparagraph{Constraints}
We consider \defp{integrity constraints} (or \defp{dependencies}) which are special sentences of first-order
logic without function symbols or constants.
We write $I \models \con$
when instance $I$ satisfies constraints~$\con$, and we then call $I$ a \defp{model}
of~$\con$.

An \defp{inclusion dependency} $\incd$ is a sentence of the form
$\tau: \forall \mybf{x} \, (R(x_1, \ldots, x_n) \rightarrow \exists
\mybf{y} \, S(z_1, \ldots, z_m))$,
where $\mybf{z} \subseteq \mybf{x} \sqcup \mybf{y}$ and no variable occurs
at two different positions of the same fact.
The left-hand side of the implication is called the \defp{body} and the
right-hand side is called the \defp{head}.
The \defp{exported variables} are the variables of $\mybf{x}$ that occur in~$\mybf{z}$.
This work only studies 
\defp{unary inclusion dependencies} ($\uid$s) which are the $\incd$s
with exactly one exported variable.
We write a $\uid$ $\tau$ as $\ui{R^p}{S^q}$, where $R^p$ and $S^q$ are the
positions of~$R(\mybf{x})$ and $S(\mybf{z})$ where the exported variable occurs.
For instance, the $\uid$ $\forall x y \, R(x, y) \rightarrow \exists z \, S(y, z)$ is
written $\ui{R^2}{S^1}$.
We assume without loss of generality that there are no \defp{trivial} $\uid$s of the form $\ui{R^p}{R^p}$.

 A \defp{functional dependency} $\fd$
is a  sentence $\phi$ of the form 
$\forall \mybf{x} \mybf{y} \, (R(x_1, \ldots, x_n) \wedge
R(y_1, \ldots, y_n) \wedge \bigwedge_{R^l \in L} x_l = y_l) \rightarrow x_r = y_r$,
where $L \subseteq \{1, \ldots, \arity{R}\}$ and $R^r \in \positions(R)$.
Since such a sentence is determined by the subset $L$ and the position
$r$, for brevity, we abbreviate such a $\phi$ as $R^L \rightarrow R^r$.
We call $\phi$ a \defp{unary functional dependency} $\ufd$ if
$\card{L} = 1$; otherwise it is \defp{higher-arity}.
For instance, $\forall x x' y y' \, R(x, x') \wedge R(y, y') \wedge
x' = y' \rightarrow x = y$ is a $\ufd$, and we write it $R^2 \rightarrow R^1$.
We assume that
$\card{L} > 0$, i.e., we do not allow nonstandard or degenerate
$\fd$s. We also disallow \defp{trivial} FDs, i.e., those for which we have $R^r \in R^L$.
Two facts $R(\mybf{a})$ and $R(\mybf{b})$ \defp{violate}~$\phi$ if
$\pi_{L}(\mybf{a}) = \pi_{L}(\mybf{b})$
but $a_r \neq b_r$.

For $L, L' \subseteq \{1, \ldots, \arity{R}\}$, we write $R^L \rightarrow R^{L'}$ the conjunction
of $\fd$s $R^L \rightarrow R^l$ for $R^l \in L'$. In particular, conjunctions of
the form $\kappa: R^L \rightarrow R$ (i.e., $L' = \{1, \ldots, \arity{R}\}$) are called \defp{key
dependencies}. The key $\kappa$ is \defp{unary} if $\card{L} = 1$. If $\kappa$
holds on a relation~$R$, we call $L$ a \defp{key} (or \defp{unary key}) of~$R$.

\subsection{Implication and Finite Implication}
\label{sec:implication}

We say that a conjunction of constraints $\con$ in a class \textsf{CL}
\defp{finitely implies} a constraint~$\phi$
if any \emph{finite} instance that satisfies $\con$ also satisfies $\phi$. We say that $\con$
\defp{implies} $\phi$ if any instance (finite or infinite) that satisfies $\con$
also satisfies~$\phi$. The
\defp{closure} $\closure{\con}$ of $\con$ is the set of constraints of
\textsf{CL} which are implied by $\con$, and the \defp{finite closure}
$\closuref{\con}$ is the set of those which are finitely implied.

A \defp{deduction rule} for \textsf{CL} is a rule which, given dependencies in~\textsc{CL}, deduces new dependencies in~\textsf{CL}.
An \defp{axiomatization of
implication} for \textsf{CL} is a set of deduction rules such that the following
holds 
for any conjunction $\con$ of dependencies in \textsf{CL}:
letting $\con'$ be the result of defining $\con' \defeq \con$ and applying
iteratively the deduction rules while possible to inflate $\con'$, then the
resulting $\con'$ is
\emph{exactly} $\closure{\con}$.
An \defp{axiomatization of finite implication} is defined
similarly but for~$\closuref{\con}$.

\myparagraph{Implication for $\incd$s}
Given a set $\con$ of $\incd$s, it is known \cite{casanova} that an $\incd$ $\tau$ is implied by
$\con$ iff it is finitely implied. Further, when $\con$ are $\uid$s, we can
easily compute in PTIME the set of implied $\uid$s (from which we exclude the trivial ones),
by closing $\con$ under the
\defp{$\uid$ transitivity rule} \cite{casanova}: if $\ui{R^p}{S^q}$ and $\ui{S^q}{T^r}$
are in $\con$, then so is $\ui{R^p}{T^r}$ unless it is trivial. We call $\con$
\defp{transitively closed} if it is thus closed.

\myparagraph{Implication for $\fd$s}
Again, a set $\fds$ of $\fd$s implies an $\fd$ $\phi$ iff it finitely implies
it: see, e.g., \cite{cosm}.
The standard axiomatization of $\fd$ implication is given in
\cite{armstrong1974dependency}, and includes the \defp{$\ufd$ transitivity
rule}: for any $R \in \sigma$ and $L, L', L'' \subseteq \{1, \ldots, \arity{R}\}$, if $R^L
\rightarrow R^{L'}$ and $R^{L'} \rightarrow R^{L''}$ are in $\fds$, then so
does $R^L \rightarrow R^{L''}$.

\myparagraph{Implication for $\uid$s and $\fd$s}
It was shown in \cite{cosm} that given constraints formed of a conjunction $\ids$ of $\uid$s
and of a conjunction 
$\fds$ of $\fd$s, the implication problem for these constraints can be axiomatized by the above $\uid$ and $\fd$ rules
in isolation. However, for \emph{finite} implication, we must add a
\defp{cycle rule}, which we now define.

Let $\con$ be a conjunction of dependencies formed of $\uid$s $\ids$ and $\fd$s $\fds$.
Define the \defp{reverse} of an $\ufd$ $\phi : R^p \rightarrow R^q$ as
$\phi^{-1} \defeq R^q \rightarrow R^p$, and the \defp{reverse} of a $\uid$ $\tau
: \ui{R^p}{S^q}$ as $\tau^{-1} \defeq \ui{S^q}{R^p}$. 
A \defp{cycle} in $\con$ is a sequence of $\uid$s and $\ufd$s of $\ids$ and
$\fds$ of the following form: $\ui{R_1^{p_1}}{R_2^{q_2}}$, $R_2^{p_2}
\rightarrow R_2^{q_2}$, $\ui{R_2^{p_2}}{R_3^{q_3}}$, $R_3^{p_3} \rightarrow
R_3^{q_3}$, $\ldots$, $\ui{R_{n-1}^{p_{n-1}}}{R_n^{q_n}}$, $R_n^{p_n}
\rightarrow R_n^{q_n}$, $\ui{R_{n}^{p_{n}}}{R_1^{q_1}}$, $R_1^{p_1} \rightarrow
R_1^{q_1}$. The \defp{cycle rule}, out of such a cycle, deduces the
reverse of each $\uid$ and of each $\ufd$ in the cycle. We then have:

\begin{theorem}[\cite{cosm}, Theorem~4.1]
  \label{thm:maincosm}
  The $\uid$ and $\fd$ deduction rules and the cycle rule are an axiomatization
  of finite implication for $\uid$s and $\fd$s.
\end{theorem}

In terms of complexity, this implies:

\begin{corollary}[\cite{cosm}, Corollary~4.4]
  \label{cor:cosmcplx}
  Given $\uid$s $\ids$ and $\fd$s $\fds$, and a $\uid$ or $\fd$ $\tau$, we can
  check in PTIME whether $\tau$ is finitely implied by $\ids$ and $\fds$.
\end{corollary}

\subsection{Queries and QA}

\myparagraph{Queries}
An \defp{atom} $A = R(\mybf{t})$ consists of a relation
name $R$ and an $\arity{R}$-tuple $\mybf{t}$ of variables or constants.
This work studies the \defp{conjunctive queries} $\cq$, which are
existentially quantified conjunctions of atoms,
such that each variable in the quantification occurs in some atom.
The \defp{size} $\card{q}$ of a $\cq$ $q$ is its number of atoms.
A $\cq$ is \defp{Boolean} if it has no free variables.

A Boolean $\cq$ $q$ \defp{holds}
in an instance $I$, written $I \models q$, exactly when there is a homomorphism~$h$ from 
the atoms of~$q$ to
$I$ such that $h$ is the identity on the constants of~$q$
(we call this a \defp{homomorphism from~$q$ to~$I$}).
We call such an~$h$ a \defp{match} of~$q$ in $I$, and by a slight abuse of terminology we
also call the image of $h$ a \defp{match} of~$q$ in~$I$. For any atom $A =
R(\mybf{t})$
of~$q$, we denote by $h(A)$ the fact $R(h(t_1), \ldots, h(t_{\card{R}}))$ of~$I$ to which $h$ maps~$A$.

\myparagraph{QA problems}
We define the \defp{unrestricted open-world query answering} problem (UQA) as
follows: given a finite instance~$I$, a conjunction of constraints~$\con$, and
a Boolean $\cq$~$q$,
decide whether there is a superinstance of~$I$ that satisfies $\con$ and
violates~$q$. If there is none, we say that $I$ and $\con$ \defp{entail} $q$
and write $(I, \con) \entunr q$.
In other words, UQA asks whether the first-order formula $I \wedge \con \wedge
\neg q$ has some (possibly infinite) model.

This work focuses on the
\defp{finite query answering problem} (FQA), which is the variant of open-world query answering where we require the
counterexample superinstance to be \emph{finite}; if no such counterexample exists, we write $(I,
\con) \entfin q$. Of course $(I, \con) \entunr q$ implies $(I, \con)
\entfin q$.

The \defp{combined complexity} of the UQA and FQA problems, for a fixed class
$\cl$ of
constraints, is the complexity of
deciding one of these problems when all of~$I$, $\con$ (in~$\cl$) and $q$ are given
as input. The \defp{data complexity} is defined by assuming 
that $\con$ and~$q$ are fixed, and only $I$ is given as input.

\myparagraph{Assumptions on queries}
Throughout this work, we will make three assumptions about $\cq$s, without loss
of generality for UQA and FQA. First, \emph{we assume that $\cq$s are constant-free}. Indeed,
for each constant $c \in \dom(I_0)$, we could otherwise do the following: add a fresh relation $P_c$
to the signature, add a fact $P_c(c)$ to $I_0$, replace $c$ in $q$
by an existentially quantified variable~$x_c$, and add the atom $P_c(x_c)$
to~$q$. It is then clear that UQA with the rewritten instance and query is
equivalent to UQA with the original instance and query under any constraints
(remember that our constraints are constant-free); the same is true for FQA.

Second, we \emph{assume all $\cq$s to be Boolean}, unless otherwise specified. Indeed,
to perform UQA for non-Boolean queries (where the domain of the free variables
is that of the base instance $I_0$), we can always enumerate all 
possible assignments, and solve our problem by solving polynomially many
instances of the UQA problem with Boolean queries. Again, the same is true of
FQA.

Third, \emph{we assume all $\cq$s to be connected}. 
A $\cq$ $q$ is \emph{disconnected} if there is a
partition of its atoms in two non-empty sets $\calA$ and $\calA'$, such that
no variable occurs both in an atom of $\calA$ and in one atom of $\calA'$.
In
this case, the query $q : \exists \mybf{x} \mybf{y} \, \calA(\mybf{x}) \wedge
\calA(\mybf{y})$ can be rewritten to $q_2 \wedge q_2'$, for two $\cq$s $q_2$
and $q_2'$ of strictly smaller size. In \thispaper, we will show that, on
finitely closed dependencies, FQA and UQA coincide for \emph{connected} queries. 
This clearly implies the same for disconnected queries, by considering all their
connected subqueries. Hence, we can assume that queries are connected.

\myparagraph{Chase}
We say that a superinstance $I'$ of an instance $I$ is \defp{universal} for
constraints~$\con$ if $I' \models \con$ and if for any Boolean $\cq$ $q$, $I'
\models q$ iff $(I, \con) \entunr q$. We now recall the definition of the
\defp{chase} (see \cite{onet2013chase} or \cite[Section~8.4]{abiteboul1995foundations}), a standard
construction of (generally infinite) universal superinstances. We assume that we
have fixed an infinite set~$\calN$ of \defp{nulls} which 
is disjoint from $\dom(I)$.
We only define the chase for transitively closed $\uid$s, which
we call the \defp{\uid chase}.

We say that a fact $F_{\factive} = R(\mybf{a})$ of an instance $I$ is an \defp{active fact}
for a $\uid$ $\tau: \ui{R^p}{S^q}$ if, intuitively, $F_{\factive}$ matches the
body of~$\tau$ but there is no fact matching the head of~$\tau$.
Formally, writing $\tau: \forall\mybf{x} \,
R(\mybf{x}) \rightarrow \exists\mybf{y} \, S(\mybf{z})$,
we call $F_{\factive}$ an \defp{active fact} for~$\tau$ if there is a homomorphism from~$R(\mybf{x})$ to~$F_{\factive}$ but no such
homomorphism can be extended to a homomorphism from
$\{R(\mybf{x}), S(\mybf{z})\}$ to~$I$.
In this case we say that
we \defp{want} to apply the $\uid$~$\tau$ to $a_p$,
written $a_p \in \appelem{I}{\tau}$. Note that
$\appelem{I}{\tau} = \pi_{R^p}(I)
\backslash \pi_{S^q}(I)$.
For a conjunction $\ids$ of $\uid$s, we may also write $a \in
\appelemb{\ids}{I}{S^q}$ if there is $\tau \in \ids$ 
of the form $\tau: \ui{U^v}{S^q}$  such that 
$a \in \appelem{I}{\tau}$; we drop the subscript when there is no
ambiguity.

The result of a \defp{chase step} on the active fact $F_{\factive} = R(\mybf{a})$
for~$\tau: \ui{R^p}{S^q}$
in~$I$ (we call this \defp{applying} $\tau$ to $F_{\factive}$) is the superinstance~$I'$
of~$I$ obtained by adding a new fact $F_{\fnew} = S(\mybf{b})$ defined as follows:
we set $b_q \defeq a_p$, which we call the \defp{exported element} (and $S^q$
the \defp{exported position} of~$F_{\fnew}$), and use fresh nulls from~$\calN$ to instantiate the
existentially quantified variables of~$\tau$ and complete $F_{\fnew}$, using a
different null at each position; we say
the corresponding elements are \defp{introduced} at $F_{\fnew}$.
This ensures that $F_{\factive}$ is no longer an active fact in~$I'$ for~$\tau$.

A \defp{chase round} of a conjunction $\ids$ of~$\uid$s on $I$
is the result of applying simultaneous chase steps on
all active facts for all $\uid$s of~$\ids$, using distinct fresh
nulls.
The \defp{$\uid$ chase} $\chase{I}{\ids}$ of~$I$ by $\ids$ is the (generally infinite)
fixpoint of applying chase rounds. It is a universal superinstance
for~$\ids$ \cite{fagin2003data}. We will sometimes use the natural forest
structure on the facts of $\chase{I}{\ids}$, where the roots are the facts
of~$I$, and every fact $F$ of $\chase{I}{\ids} \setminus I$ has a \defp{parent}
which is some arbitrary choice $F'$ of an active fact used to create~$F$. The
\defp{children} of a fact $F'$ of $\chase{I}{\ids}$ are all facts $F$ such that
$F'$ is the parent of~$F$.

As we are chasing by transitively closed $\uid$s,  
if we perform the \defp{core chase} \cite{deutsch2008chase,onet2013chase} rather than the
$\uid$ chase that we just defined,
we can ensure the following
\defp{Unique Witness Property}: for any element $a \in \dom(\chase{I}{\ids})$ and position $R^p$ of~$\sigma$,
if two different facts of~$\chase{I}{\ids}$ contain~$a$ at position~$R^p$,
then they are both facts of~$I$.
In our context, however, the core chase matches the $\uid$ chase defined above, except at the
first round. Thus, modulo the first round, by $\chase{I}{\ids}$ we refer
to the $\uid$ chase, which has the Unique Witness Property.
See \apxref{apx:chase} for details.

\myparagraph{Finite controllability}
We say a conjunction of constraints $\con$ is \defp{finitely controllable} for
$\cq$ if FQA and UQA   coincide:
for every finite instance~$I$ and every Boolean
$\cq$~$q$,
 $(I, \con) \entunr q$ iff $(I, \con)
\entfin q$.

It was shown in \cite{rosatifinitecontrolpods,rosatifinitecontrol} that, while conjunctions of~$\incd$s are
finitely controllable, even conjunctions of~$\uid$s and $\fd$s may not be.
It was later shown in \cite{rosatieswc} that
the finite closure process could be used to reduce FQA to UQA for some constraints
on relations of arity at most two. Following the same idea,
we say that a conjunction of constraints $\con$ 
is \defp{finitely controllable up to finite closure} if for every finite
instance $I$, and Boolean $\cq$~$q$,
 $(I, \con) \entfin q$
 iff $(I, \closuref{\con}) \entunr q$,
 where $\closuref{\con}$ is the finite closure defined by
 Theorem~\ref{thm:maincosm}.
If $\con$ is finitely controllable up to finite closure, then we can reduce FQA to UQA,
even if finite controllability does not hold, by computing the finite closure
$\closuref{\con}$ of~$\con$ and solving UQA on~$\closuref{\con}$.

\section{Main Result and Overall Approach}
\label{sec:overall}
We study open-world query answering for $\fd$s and $\uid$s.
For UQA, the following is already known:

\begin{mainproposition}
  \label{prp:complexity}
  UQA for $\fd$s and $\uid$s has \ACZ data complexity and \NP-complete 
  combined complexity.
\end{mainproposition}

\begin{proof}
UQA for $\uid$s in isolation is \NP-complete in
combined complexity. The lower bound is immediate because query evaluation for
  conjunctive queries is NP-complete already without $\uid$s
  \cite[Theorem~6.4.2]{abiteboul1995foundations},
and \cite{johnson1984testing} showed an \NP upper bound for
$\incd$s with any fixed bound on the number of exported variables (which they call ``width'': in their
terminology, $\uid$s are $\incd$s of width~$1$).
For data complexity, the upper bound is from the first-order rewritability of
certain answers for arbitrary $\incd$s~\cite{cali2003query}.

For $\uid$s and $\fd$s,
clearly the lower bound on combined complexity also applies.
The upper bounds are proved by observing that 
$\uid$s and $\fd$s are
 \emph{separable}, namely,
for any $\fd$s $\fds$ and $\uid$s
$\ids$, for any instance~$I_0$ and $\cq$~$q$, if $I_0 \models \fds$
then we have $(I_0, \fds \wedge \ids) \entunr q$ iff $(I_0, \ids) \entunr q$.
Assuming separability, to decide UQA for $\ids$ and $\fds$, we first check whether $I_0
\models \fds$, in PTIME combined complexity, and \ACZ data complexity as
  $\fds$ is expressible in first-order logic which can be evaluated in \ACZ \cite[Theorem~17.1.2]{abiteboul1995foundations}.
If $I_0 \not\models \fds$, then we
  vacuously have $(I_0, \ids \cup \fds) \entunr q$ so UQA is trivial.
Otherwise, we then determine whether $(I_0, \ids) \entunr q$,
using the upper bound for UQA for $\uid$s. By separability, the answer to UQA
under $\ids$ is the same as the answer to UQA under $\fds \wedge \ids$.

Hence, all that remains to show is that $\uid$s and $\fd$s are always separable.
This follows from the \emph{non-conflicting condition}
of \cite{cali2003decidability,cali2012towards} but we give a simpler self-contained
argument.
Assume that $I_0$ satisfies  $\fds$.
It is obvious that $(I_0, \ids) \entunr q$ implies $(I_0, \fds \wedge \ids) \entunr q$, so
let us prove the converse implication. We do it by noticing that 
the chase $\chase{I_0}{\ids}$ satisfies $\fds$. Indeed, 
assuming to the contrary the existence of~$F$ and $F'$ in
$\chase{I_0}{\ids}$ violating an $\fd$ of~$\fds$, there must exist a position
$R^p \in \pos(\sigma)$ such that $\pi_{R^p}(F) = \pi_{R^p}(F')$. Yet, by the 
Unique Witness Property, this implies that $F$ and $F'$ are facts of~$I_0$,
but we assumed that $I_0 \models \fds$, a contradiction.

Hence, $\chase{I_0}{\ids}$ satisfies $\fds$, so it is a superinstance
of~$I_0$ that satisfies $\fds \wedge \ids$. Hence, $(I_0, \fds \wedge \ids) \entunr q$ implies that we must
have $\chase{I_0}{\ids} \models q$. By universality of the chase, this implies
$(I_0, \ids) \entunr q$. Hence, the converse implication is proven, so the two UQA
problems are equivalent, which implies that $\ids$ and $\fds$ are separable.
\end{proof}

In the \emph{finite case}, however, even the decidability of FQA for $\fd$s and
$\uid$s was not known. \Thispaper shows that it is decidable, and that the
complexity matches that of UQA:

\begin{maintheorem}
  \label{thm:complexityfqa}
 FQA for $\fd$s and $\uid$s has \ACZ data complexity and \NP-complete
  combined complexity.
\end{maintheorem}

This result follows from our Main Theorem, which is proven in the rest of
\thispaper:

\begin{maintheorem}[Main theorem]
  \label{thm:mymaintheorem}
Conjunctions of~$\fd$s and $\uid$s are finitely controllable up to finite closure.
\end{maintheorem}

From the Main Theorem, we can prove
Theorem~\ref{thm:complexityfqa}, using the closure process of \cite{cosm}:

\begin{proof}[\proofof Theorem~\ref{thm:complexityfqa}]
  Again, the \NP-hardness lower bound is immediate from query
  evaluation \cite[Theorem~6.4.2]{abiteboul1995foundations}, so we only show the upper bounds.
  Consider an input to the FQA problem for $\fd$s and $\uid$s, consisting of an instance
  $I_0$, a conjunction $\con$ of~$\incd$s $\ids$ and $\fd$s $\fds$, and a $\cq$
  $q$. Let $\fds^*$ be the $\fd$s 
  and $\ids^*$ the $\uid$s of the finite closure $\closuref{\con}$. By our Main
  Theorem, we have $(I_0, \con) \entfin q$ iff $(I_0, \closuref{\con}) \entunr q$.
  As the
  computation of $\closuref{\con}$ from $\con$ is data-independent, the data complexity
  upper bounds follow from Proposition~\ref{prp:complexity}, so we
  need only show the combined complexity upper bound.

  Materializing $\closuref{\con}$ from the input may take exponential time, which we
  cannot afford, so we need a
  more clever approach. Remember from the proof of Proposition~\ref{prp:complexity} that, as
  $\closuref{\con}$ consists of $\uid$s and $\fd$s, it is separable. Hence, to solve UQA
  for $I_0$, $\closuref{\con}$ and $q$, as $\closuref{\con}$ is separable, we need to perform two steps: \begin{inparaenum}
  \item check whether $I_0 \models \fds^*$
  \item if yes, solve UQA for $I_0$, $\ids^*$ and $q$.
  \end{inparaenum}

  To perform step~1, compute in PTIME the set $\ufds^*$ of the $\ufd$s of
  $\closuref{\con}$,
  using Corollary~\ref{cor:cosmcplx}. By \cite{cosm} (remark above
  Corollary~4.4), all non-unary $\fd$s in $\fds^*$ are implied by $\ufds^*
  \fragcup
  \fds$ under the axiomatization of $\fd$ implication; hence, to check whether
  $I_0
  \models \fds^*$, it suffices to check whether $I_0 \models \ufds^*$ and $I_0 \models
  \fds$, which we do in PTIME.

  To perform step~2, compute $\ids^*$ in PTIME by considering each possible
  $\uid$ (there are polynomially many) and determining in PTIME from $\con$
  whether it is in $\closuref{\con}$, using Corollary~\ref{cor:cosmcplx}. Then, solve UQA
  in \NP combined complexity
  by Proposition~\ref{prp:complexity}. The entire process takes \NP combined
  complexity, and the answer matches that of FQA by our Main
  Theorem, which proves the \NP upper bound.
\end{proof}

In this \secname, we first explain
how we can prove Theorem~\ref{thm:mymaintheorem} from a different statement, namely:
we can construct \defo{finite universal superinstances} for finitely closed $\uid$s and
$\fd$s, which we will equivalently call \defo{finite universal models}. We conclude this \secname with the outline of the proof of this result
(Theorem~\ref{thm:univexists}) which will be developed in the rest of \thispaper.

\subsection{Finite Universal Superinstances}

Our Main Theorem claims that a certain class of constraints, namely finitely
closed $\uid$s and $\fd$s, are finitely controllable for the class of
conjunctive queries ($\cq$). To prove this, it will be easier to work
with a notion of \defo{$k$-sound} and \defo{$k$-universal instances}.

\begin{definition}
  \label{def:finuniv}
  For $k \in \mathbb{N}$, we say that a superinstance $I$ of an instance~$I_0$ is
  \deft{$\bm{k}$-sound} for constraints $\con$, for~$I_0$, and for $\cq$s if,
  for every 
  $\cq$~$q$ of size $\leq k$ such that $I \models q$,
  we have $(I_0,
  \con) \entunr q$. We say it is \deft{$\bm{k}$-universal} if the converse also
  holds: $I \models
  q$ whenever $(I_0, \con) \entunr q$.
  For a subclass $\textsf{Q}$ of $\cq$s,
  we call $I$ \deft{$\bm{k}$-sound} or \deft{$\bm{k}$-universal}
  for~$\con$, for~$I_0$,  and for~$\textsf{Q}$ if the same holds for
  all queries~$q$ of size $\leq k$
  that are in~$\textsf{Q}$.

  We say that a class $\textsf{CL}$ of constraints
  \emph{has finite universal superinstances} (or, for brevity, \emph{finite universal
  models})
  for a class $\textsf{Q}$ of $\cq$s,
  if for any constraints $\con$ of
  $\textsf{CL}$, for any $k \in \NN$, for any instance $I_0$, if $I_0$ has some
  superinstance that satisfies $\con$, then it has a 
  \emph{finite} 
  superinstance 
  that satisfies $\con$ 
  and is $k$-sound for~$\con$ and~$\textsf{Q}$
  (and hence is also $k$-universal for~$\con$ and~$\textsf{Q}$).
\end{definition}

We will thus show that the class of finitely closed $\uid$s and $\fd$s have finite universal
superinstances for~$\cq$s. We explain why this implies our Main Theorem:

\begin{proposition}
  \label{prp:univsup2fc}
  If constraint class $\textsf{CL}$ has
  finite universal superinstances
  for query class $\textsf{Q}$, then $\textsf{CL}$ is finitely controllable
  for $\textsf{Q}$.
\end{proposition}

\begin{proof}
  Let $\con$ be constraints in $\textsf{CL}$, $I_0$ be a finite instance and
  $q$ be a query in~$\textsf{Q}$. We show that $(I_0, \con) \entunr q$ iff $(I_0, \con)
  \entfin q$. The forward implication is immediate: if all superinstances of
  $I_0$
  that satisfy $\con$ must satisfy $q$, then so do the finite ones.

  For the converse implication, assume that $(I_0, \con) \not\entunr q$. In
  particular, this implies that $I_0$ has some superinstance that satisfies
  $\con$, as otherwise the entailment would be vacuously true. As $\textsf{CL}$
  has finite universal superinstances for $\textsf{Q}$, let $I$ be a finite
  $k$-sound superinstance of $I_0$ that satisfies $\con$, where $k \defeq \card{q}$. As $I$ is
  $k$-sound, we have $I \not\models q$, and as $I \models \con$, $I$ witnesses
  that $(I_0, \con) \not\entfin q$. This proves the converse direction, so we
  have established finite controllability.
\end{proof}

So, in \thispaper, we will actually show the following restatement of the Main
Theorem:

\mydefthm{univexists}{theorem}{Universal models}{
  The class of finitely closed $\uid$s and $\fd$s has finite
  universal models for~$\cq$: for every conjunction $\con$ of~$\fd$s $\fds$ and $\uid$s $\ids$ closed under
finite implication,
for any $k \in \mathbb{N}$,
for every finite instance $I_0$ that satisfies $\fds$,
there exists a finite $k$-sound superinstance $I$ of~$I_0$ that satisfies $\con$.
}

Indeed, once we have shown this, we can easily deduce the Main Theorem, namely,
that any 
conjunction $\con$ of~$\fd$s and~$\uid$s is finitely controllable up to finite
closure. Indeed, for any such $\con$, for any instance $I_0$ and $\cq$ $q$, we have $(I_0, \con) \entfin q$ iff
$(I_0, \closuref{\con}) \entfin q$: the forward statement is because any finite
model of~$\con$ is a model of~$\closuref{\con}$, and the backward statement is
tautological. Now, from the Universal Models Theorem and
Proposition~\ref{prp:univsup2fc}, we know that $\closuref{\con}$ is finitely
controllable, so that $(I_0, \closuref{\con}) \entfin
q$ iff $(I_0, \closuref{\con}) \entunr q$. We have thus shown that $(I_0, \con) \entfin
q$ iff $(I_0, \closuref{\con}) \entunr q$, which concludes the proof of the Main
Theorem.

Hence, we will show the Universal Models Theorem in the rest of \thispaper. We
proceed in incremental steps, following the plan that we outline next.

\subsection{Proof Structure}

We first make a simplifying assumption on the signature, without loss of
generality, to remove \emph{useless} relations.
Given an instance $I_0$, $\uid$s $\ids$ and $\fd$s $\fds$, it may the be case
that the signature $\sigma$ contains a relation $R$ that does not occur in
$\chase{I_0}{\ids}$, namely, it does not occur in $I_0$ and the existence of an
$R$-fact is not implied by~$\ids$. In this case, relation $R$ is 
\emph{useless}: a
$\cq$ $q$ involving $R$ will never be entailed under~$\con$, neither on
unrestricted nor on finite models, unless $I_0$ has no completion at all
satisfying the constraints. In any case, the query $q$ can be replaced by the
trivial $\cq$ $\text{False}$, which is only (vacuously) entailed if there are no
completions.

Hence, we can always remove useless relations from the signature, up to
rewriting the query to the false query. Thus, without loss of generality,
\emph{we always assume that the signature contains no useless relations}
in this sense: all relations of the signature occur in the chase.

\bigskip

We now present several assumptions that we use to prove weakenings of the
Universal Models Theorem. The first is on queries, which we require
to be \emph{acyclic}. The second is on $\fd$s, which we require to be
\emph{unary}, i.e., $\ufd$s. The
third is to replace $k$-soundness by the simpler notion of
\emph{weak-soundness}. Then we present two additional assumptions: the first one,
the reversibility assumption, is on the constraints, and requires that they have a certain
special form; the second one, the arity-two assumption, is on the constraints and signature,
which we require to be binary. In the next \secname, we show the Universal
Models Theorem under all these assumptions, and then we lift the assumptions one
by one, in each \secname. See
Table~\ref{tab:roadmap} for a synopsis.

Hence, let us present the assumptions that we will make (and later lift).

\begin{table}
  \caption{Roadmap of intermediate results.}
  \label{tab:roadmap}
  \begin{minipage}{\linewidth}
    \begin{tabularx}{\linewidth}{XXXrl@{\qquad\qquad}c}
    \toprule
     & {\bf Signature}
    & {\bf Universality} & 
    \multicolumn{2}{c}{\bf Constraints} & {\bf Query}  
    \\
    \midrule
  \bf \secshortname~\ref{sec:binary}: & binary & weakly-sound &
    reversible & $\uid$s, $\ufd$s & $\acq$ 
    \\ 
  \bf \secshortname~\ref{sec:weakarb}: &
    \bf arbitrary
    & weakly-sound &
    reversible & $\uid$s, $\ufd$s & $\acq$
    \\ 
  \bf \secshortname~\ref{sec:ksound}: &
    arbitrary
    & \bf $\mybf{k}$-sound &
    reversible & $\uid$s, $\ufd$s & $\acq$
    \\ 
  \bf \secshortname~\ref{sec:manyscc}: &
    arbitrary
    & $k$-sound &
    \textbf{finitely closed} & $\uid$s, $\ufd$s & $\acq$
    \\ 
  \bf \secshortname~\ref{sec:hfds}: &  
    arbitrary
    & $k$-sound & 
    finitely closed & $\uid$s, {\bf $\fd$s} & $\acq$
    \\ 
  \bf \secshortname~\ref{sec:cycles}: &
    arbitrary
    & $k$-sound &
    finitely closed & $\uid$s, $\fd$s & \bf $\cq$ 
    \\ 
    \bottomrule
  \end{tabularx}
\end{minipage}
\end{table}

\myparagraph{Acyclic queries}
It will be helpful to focus first on the subset of \emph{acyclic} $\cq$s,
denoted $\acq$, which are the queries that contain no \defo{Berge cycle}.
Formally:
\begin{definition}
  \label{def:acq}
  A \defp{Berge cycle} in a $\cq$~$q$ is a sequence 
$A_1, \allowbreak x_1, \allowbreak A_2, \allowbreak x_2, \allowbreak \ldots,
\allowbreak A_n, \allowbreak x_n$ with $n \geq 2$,
where the $A_i$ are pairwise distinct atoms of~$q$, the
$x_i$ are pairwise distinct variables of $q$,
the variable  $x_i$ occurs in $A_i$ and $A_{i+1}$
for $1 \leq i < n$, and 
$x_n$ occurs in~$A_n$ and~$A_1$.
A query~$q$ is in $\acq$ if $q$ has no Berge cycle and if no variable of~$q$
occurs more than once in the same atom.

Equivalently, consider the \defp{incidence multigraph} of $q$, namely, the bipartite
undirected multigraph on variables and atoms obtained by putting one edge
between variable $x$ and atom $A$ for every time where $x$ occurs in $A$
(possibly multiple times). Then $q$ is in~$\acq$ iff its incidence
multigraph is acyclic in the standard sense.
\end{definition}

\begin{example}
  The queries $\exists x ~ R(x, x)$, $\exists x y ~ R(x, y) \wedge S(x, y)$, and
  $\exists x y z ~ R(x, y) \wedge R(y, z) \wedge R(z, x)$ are not in~$\acq$: the
  first has an atom with two occurrences of the same variable, the other two
  have a Berge cycle. The following query is
  in~$\acq$: $\exists x y z w ~ R(x, y, z) \wedge S(x) \wedge T(y, w) \wedge U(w)$.
\end{example}

Intuitively, in the chase, all query matches are acyclic unless they involve
some cycle in the initial instance $I_0$. Hence, only acyclic $\cq$s have
matches, except those that match on $I_0$ or those whose cycles have
self-homomorphic matches, so, in a $k$-sound model, the $\cq$s of
size $\leq k$ which hold are usually acyclic. For this reason, we focus only
on $\acq$ queries first. We will ensure in \secref{sec:cycles} that cyclic
queries of size $\leq k$ have no matches.

\myparagraph{Unary $\fd$s}
We will first show our result for \emph{unary} $\fd$s ($\ufd$s); recall
from \secref{sec:prelim} that they are the $\fd$s with exactly one
determining attribute.
We do this because the finite closure construction of \cite{cosm} is not concerned
with higher-arity $\fd$s, except for the $\ufd$s that they imply. Hence, while
the $\ufd$s of the finite closure have a special structure that we can rely on, the
higher-arity $\fd$s are essentially arbitrary. This is why we deal with them
only in \secref{sec:hfds}, using a different approach.

\myparagraph{$\mybf{k}$-soundness and weak-soundness}
Rather than proving that $\uid$s and $\ufd$s have finite universal
models for $\acq$, it will be easier to prove first that they have \emph{finite
weakly-universal models}. This is defined relative to a notion of \emph{weak
soundness}, a weakening of $k$-soundness that we use as an intermediate step in
the proof:

\begin{definition}
  \label{def:wssinstance}
  A superinstance $I$ of an instance $I_0$ is
  \deft{weakly-sound} for a set of $\uid$s $\ids$ and for~$I_0$ if the following holds:
  \begin{itemize}
  \item Elements of~$I_0$ only appear in new facts at positions where they want to
    appear. Formally, for any $a \in \dom(I_0)$ and $R^p \in \pos(\sigma)$, if $a \in
    \pi_{R^p}(I)$, then either $a \in \pi_{R^p}(I_0)$ or $a \in
    \appelem{I_0}{R^p}$;
  \item Each new element only occurs at positions that are related by UIDs. Formally, for any $a \in \dom(I) \backslash \dom(I_0)$ and $R^p, S^q \in \pos(\sigma)$,
    if $a \in \pi_{R^p}(I)$ and $a \in \pi_{S^q}(I)$ then either we have $R^p = S^q$ or $\ui{R^p}{S^q}$
    and $\ui{S^q}{R^p}$
    are in~$\ids$. \qedef
  \end{itemize}
\end{definition}

Thus, we first show that $\ufd$s and $\uid$s have
\emph{finite weakly-universal
superinstances} for $\acq$,
defined analogously to Definition~\ref{def:finuniv}:
for any constraints $\ucon$ of~$\ufd$s $\ufds$ and $\uid$s $\ids$, for any query $q$ in $\acq$, for any instance $I_0$,
if $I_0$ has a superinstance that satisfies~$\ucon$, then it has a \emph{finite}
superinstance that satisfies~$\ucon$ and is weakly-sound for~$\ids$ and~$I_0$.
We will then generalize from weak soundness to the general case of $k$-soundness
in \secref{sec:ksound}.

\myparagraph{Reversibility assumption}
We will initially make a simplifying assumption on the structure of the $\uid$s
and $\ufd$s, which we call the \emph{reversibility assumption}. This assumption is motivated by
the finite closure rules of Theorem~\ref{thm:maincosm}; intuitively, it amounts
to assuming that a certain \emph{constraint graph} defined from the dependencies has a single
connected component:

\begin{definition}
  \label{def:reversible}
  Let $\idsr$ be a set of $\uid$s and $\ufds$ be a set of $\ufd$s.
  We call $\idsr$ and $\ufds$ \deft{reversible} if:
    \begin{itemize}
      \item The set $\idsr$ is closed under implication, and so is $\ufds$;
    \item All $\uid$s in $\idsr$ are reversible, i.e., their reverses
  are also in~$\idsr$;
\item For any $\ufd$ $\phi: R^p \rightarrow R^q$ in $\ufds$,
  if $R^p$ occurs in some $\uid$ of~$\idsr$ and $R^q$ also occurs in some $\uid$
  of~$\idsr$, then $\phi$ is reversible,
  i.e., $\phi^{-1}$ is also in~$\ufds$.\qedef
  \end{itemize}
\end{definition}

We can now state the assumption that we make:

\medskip

\begin{compactitem}
\item \emph{Reversibility assumption:} The $\uid$s $\ids$ and $\ufd$s $\ufds$ are reversible.
\end{compactitem}

\medskip

When making the reversibility assumption, we will write the $\uid$s $\idsr$
rather than $\ids$, as in the definition above.
Observe that $\idsr$ and $\ufds$ are then
finitely closed: they are closed under $\uid$ and
$\ufd$ implication, and the $\uid$s and $\ufd$s of any cycle must be reversible.
To lift the reversibility assumption and generalize to the general case, we will follow an SCC
decomposition of the constraint graph to manage each SCC separately. See
\secref{sec:manyscc} for details.

\myparagraph{Arity-two assumption}
We will start our proof in \secref{sec:binary} by introducing important notions in the much simpler
case of a binary signature. For this, we will initially make the following
\emph{arity-two assumption} on the signature and on~$\con$:

\medskip
\begin{compactitem}
\item \emph{Arity-two assumption:} 
  Each relation $R$ has arity~$2$ and the $\ufd$s $R^1
  \rightarrow R^2$ and $R^2 \rightarrow R^1$ are in~$\con$.
\end{compactitem}
\medskip

We will lift this assumption in \secref{sec:weakarb}.

\myparagraph{Roadmap}
Each of the next \secnames will prove that a certain constraint class
has finite universal models for a certain query class in a certain sense, under certain
assumptions. Table~\ref{tab:roadmap} summarizes the results that are proved in
each \secname.

The rest of \thepaper follows this roadmap: each \secname starts by stating
the result that it proves.

\section{Weak Soundness on Binary Signatures}
\label{sec:binary}
\begin{maintheorem}
  \label{thm:binary}
  Reversible $\uid$s and $\ufd$s have finite weakly-universal
  superinstances
  for $\acq$s under the arity-two assumption.
\end{maintheorem}

We prove this result in this \secname. Fix an instance $I_0$ and reversible
constraints $\uconr$ formed of $\uid$s $\idsr$ and $\ufd$s $\ufds$. Assume
that $I_0 \models \ufds$ as the question is vacuous otherwise, and make
the arity-two assumption.

Our goal is to construct a weakly-sound superinstance $I$ of~$I_0$ that satisfies
$\uconr$. We do so by a \emph{completion process}
that adds new (binary) facts  to connect elements together.
Remember that the arity-two assumption implies that all 
possible $\ufd$s hold, so if we extend $I_0$ to $I$ by adding 
a new fact $R(a_1, a_2)$, we must have $a_i \notin \pi_{R^i}(I_0)$ for $i \in
\{1, 2\}$. In order to achieve weak soundness we require in particular
that if $a_i \in \dom(I_0)$  we must have $a_i \in
\appelem{I_0}{R^i}$. Our task in this section is thus to complete $I_0$ to $I$ by adding
$R$-facts, for each relation $R$, that connect together elements of
$\appelem{I_0}{R^1}$ and $\appelem{I_0}{R^2}$.

\subsection{Completing Balanced Instances}

One situation where completion is easy is when the instance $I_0$ is \defo{balanced}: for
every relation~$R$, we can construct a bijection between the elements that want
to be in~$R^1$ and those that want to be in~$R^2$:

\begin{definition}
  \label{def:balanced}
  We call $I_0$ \deft{balanced} (for $\uid$s $\idsr$) if, for every two
  positions $R^p$ and $R^q$ such that $R^p \rightarrow R^q$ and $R^q \rightarrow
  R^p$ are in~$\ufds$, we have $\card{\appelemb{\idsr}{I_0}{R^p}} =
  \card{\appelemb{\idsr}{I_0}{R^q}}$.
\end{definition}

Note that, as we make the arity-two assumption in this section, the notion of
balancedness always applies to the two positions $R^1$ and $R^2$ of each relation $R$.
However, the definition is phrased in a more general
way because we will reuse it in later sections without making the arity-two
assumption.

If $I_0$ is balanced, we can show Theorem~\ref{thm:binary} by constructing $I$
with $\dom(I) = \dom(I_0)$, adding new facts that pair together the existing
elements:

\begin{proposition}
  \label{prp:balwsnd}
  Under the arity-two and reversibility assumptions,
  any balanced finite instance $I_0$ satisfying $\ufds$ has a finite weakly-sound
  superinstance $I$ that satisfies $\uconr$, with $\dom(I) = \dom(I_0)$.
\end{proposition}

We first exemplify this process:

\begin{figure}
  {
  \noindent\null\hfill\begin{minipage}[b]{.62\linewidth}
    \centering
    \begin{tikzpicture}[
  text height=1ex,text depth=0ex,xscale=1.5,
  mispos/.append style={red, font=\scriptsize},
  newfact/.append style={red, dashed, -{>[scale=1.2]}}
]
  \node (a) at (0.5, 0) {$a$};
  \node (b) at (1.5, 0) {$b$};
  \node (c) at (2.5, 0) {$c$};
  \node (d) at (3.5, 0) {$d$};
  \node (e) at (4.5, 0) {$e$};

  \node (f) at (0, -2) {$f$};
  \node (g) at (1, -2) {$g$};
  \node (h) at (2, -2) {$h$};
  
  \node (a2) [below = -.5ex of a,mispos] {$T^2$};
  \node (b2) [below = -.5ex of b,mispos] {$S^1$};
  \node (c2) [below = -.5ex of c,mispos] {$R^2$};
  \node (e2) [below = -.5ex of e,mispos] {$R^1$};
  
  \node (f2) [below = -.5ex of f,mispos] {$T^1$};
  \node (h2) [below = -.5ex of h,mispos] {$S^2$};

  \draw[->] (a) -- (b) node[above,midway] {$R$};
  \draw[->] (b) -- (c) node[above,midway] {$U$};
  \draw[->] (c) -- (d) node[above,midway] {$S$};
  \draw[->] (d) -- (e) node[above,midway] {$T$};
  
  \draw[->] (g) -- (f) node[above,midway] {$S$};
  \path[->] (g) edge [loop above,in=110,out=70,min distance=7mm] node {$R$} (g);
  \draw[->] (h) -- (g) node[above,midway] {$T$};

  \path (e) edge[newfact,bend left] node[below,midway,yshift=-.5ex] {$R$} (c);
  \path (b) edge[newfact,bend left]
  node[above,midway,xshift=1.5ex,yshift=-1ex] {$S$} (h);
  \path (f) edge[newfact,bend left,shorten >=2pt]
  node[above,midway,xshift=-1.5ex,yshift=-1ex] {$T$} (a);
\end{tikzpicture}
    \caption{Connecting balanced instances (see Example~\ref{exa:balance})}
    \label{fig:complete_balanced}
  \end{minipage}\hfill\begin{minipage}[b]{.32\linewidth}
    \centering
    \begin{tikzpicture}[
  text height=1ex,text depth=0ex,xscale=1.5,yscale=.8,
  mispos/.append style={red, font=\scriptsize},
  newfact/.append style={red, dashed, -{>[scale=1.2]}}
]
  \node (a) at (0, 0) {$a$};
  \node (b) at (1, 0) {$b$};

  \node[red] (h) at (.5, -2) {$h$};

  \node (a2) [below = -.5ex of a,mispos] {$T^2$};
  \node (b2) [below = -.5ex of b,mispos] {$S^1$};
  
  \node (h2) [below = -.5ex of h,mispos,align=center] {$S^2, T^1$};

  \draw[->] (a) -- (b) node[above,midway] {$R$};
  
  \path (b) edge[newfact,bend left=60] node[below,midway,xshift=1ex] {$S$} (h);
  \path (h) edge[newfact,bend left=60,shorten >=2pt] node[below,midway,xshift=-1ex] {$T$} (a);
\end{tikzpicture}
    \caption{Using helper elements to balance (see Example~\ref{exa:helper})}
    \label{fig:complete_helper}
  \end{minipage}\hfill\null
}
\end{figure}

\begin{example}
  \label{exa:balance}
  Consider four binary relations $R$, $S$, $T$, and $U$, with the $\uid$s
  $\ui{R^2}{S^1}$, $\ui{S^2}{T^1}$, $\ui{T^2}{R^1}$
  and their reverses, and the
  $\fd$s prescribed by the arity-two assumption. Consider $I_0 \defeq \{R(a,
  b), U(b, c), S(c, d), T(d, e), \allowbreak S(g, f), \allowbreak R(g, g), \allowbreak T(h, g)\}$, as depicted by
  the black elements and solid black arrows in Figure~\ref{fig:complete_balanced}.

  We compute, for each element, the set of positions where it wants to be,
  and write it in red under each element in Figure~\ref{fig:complete_balanced}
  (in this example, it is a set of size at most one for each element).
  For instance, we have $\appelem{I_0}{T^1} = \{f\}$. We observe that the
  instance is balanced: we have $\card{\appelem{I_0}{R^1}} =
  \card{\appelem{I_0}{R^2}}$, and likewise for~$S$, $T$, and~$U$.

  We can construct a weakly-sound superinstance $I$ of $I_0$ as $I \defeq I_0
  \sqcup \{R(e, c), \allowbreak S(b, h), \allowbreak T(f, a)\}$: the additional facts are represented as
  dashed red arrows in Figure~\ref{fig:complete_balanced}. Intuitively, we just
  create new facts that connect together elements which want to occur at the right
  positions.
\end{example}

We now give the formal proof of the result:

\begin{proof}[\proofof Proposition~\ref{prp:balwsnd}]
Define a bijection $f_R$ from
$\appelem{I_0}{R^1}$ to $\appelem{I_0}{R^2}$
for every relation~$R$ of~$\sigma$;
this is possible because $I_0$ is
balanced.

Consider the superinstance $I$ of~$I_0$, with $\dom(I) = \dom(I_0)$, obtained by
adding, for every $R$ of~$\sigma$, the fact $R(a, f_R(a))$ for every $a \in
\appelem{I_0}{R^1}$. $I$ is clearly a finite weakly-sound superinstance
of~$I_0$,
because for every $a \in \dom(I)$, if $a$ occurs at some position $R^p$ in some
fact $F$ of~$I$, then either $F$ is a fact of~$I_0$ and $a \in \pi_{R^p}(I_0)$, or
$F$ is a new fact in $I \backslash I_0$ and by definition $a \in \appelem{I_0}{R^p}$.

Let us show that $I \models \ufds$. Assume to the contrary that two
facts $F = R(a_1, a_2)$ and $F' = R(a_1', a_2')$ in $I$ witness a violation of a $\ufd$ $\phi: R^1
\rightarrow R^2$ of~$\ufds$. As $I_0 \models \ufds$, one of~$F$ and $F'$, say $F$,
must be a new fact.
By definition of the new facts, we have $a_1 \in \appelem{I_0}{R^1}$, so that $a_1
\notin \pi_{R^1}(I_0)$. Now, 
as $\{F, F'\}$ is a violation, we must have $\pi_{R^1}(F) = \pi_{R^1}(F')$, so
as $a_1 \notin \pi_{R^1}(I_0)$, $F'$ must also be a new fact. Hence, by definition
of the new facts, we have $a_2 =
a'_2 = f_R(a_1)$, so $F = F'$, which contradicts the fact
that $F$ and $F'$ violate $\phi$. For $\ufd$s $\phi$ of the form $R^2
\rightarrow R^1$, the proof is similar, but we have $a_1 = a'_1 =
f_R^{-1}(a_2)$.

Let us now show that $I \models \idsr$. Assume to the contrary that there is
an active fact $F = R(a_1, a_2)$ that witnesses the violation of a $\uid$ $\tau : \ui{R^p}{S^q}$. If $F$ is a fact of
$I_0$, we had $a_p \in \appelem{I_0}{S^q}$, so $F$
cannot be an active fact in~$I$ as this violation was solved in~$I$.
So we must have $F \in I \backslash I_0$. Hence,
by definition of the new facts,
we had $a_p \in \appelem{I_0}{R^p}$; so there must be $\tau' : \ui{T^r}{R^p}$ in
$\idsr$ such that
$a_p \in \pi_{T^r}(I_0)$. 
Hence, because $\idsr$ is transitively closed, either $T^r = S^q$ or
the $\uid$
$\ui{T^r}{S^q}$ is in $\idsr$. In the first case, as $a_p \in \pi_{T^r}(I_0)$, $F$
cannot be an active fact for~$\tau$, a contradiction.
In the second case, we had $a_p
\in \appelem{I_0}{S^q}$, so $a_p \in \pi_{S^q}(I)$ by definition of~$I$, so
again $F$ cannot be an active fact for~$\tau$.

Hence, $I$ is a finite weakly-sound superinstance of~$I_0$ that satisfies
$\uconr$
and with $\dom(I) = \dom(I_0)$, the desired claim.
\end{proof}

\subsection{Adding Helper Elements}

If our instance $I_0$ is not balanced, we cannot use the construction that we
just presented. The idea
is then to \emph{make $I_0$ balanced},
which we do by adding ``helper'' elements that we assign to positions.
The following example illustrates this:

\begin{example}
  \label{exa:helper}
  We use the same signature and dependencies as in Example~\ref{exa:balance}.
  Consider $I_0 \defeq \{R(a, b)\}$, as depicted in
  Figure~\ref{fig:complete_helper}.
  We have $a \in \appelem{I_0}{T^2}$ and $b \in \appelem{I_0}{S^1}$;
  however $\appelem{I_0}{S^2} = \appelem{I_0}{T^1} = \emptyset$, so $I_0$ is not balanced.

  Still, we can construct the weakly-sound superinstance $I \defeq I_0 \sqcup \{S(b, h), \allowbreak
  T(h, a)\}$ that satisfies the constraints. Intuitively, we have added a
  ``helper'' element~$h$ and ``assigned'' it to the positions $\{S^2, T^1\}$,
  so we could connect~$b$ to~$h$ with~$S$ and $h$ to~$a$ with~$T$.
\end{example}

We will formalize this idea of 
augmenting the domain with \defo{helper elements}, as a
\defo{partially-specified superinstance}, namely, an instance that is augmented with helpers assigned to positions.
However, we first need to
understand at which positions the helpers can appear, without violating
weak-soundness:

\begin{definition}
  \label{def:eqids}
  For any two positions $R^p$ and $S^q$, we write $R^p \eqids S^q$ when
  $R^p = S^q$ or when
  $\ui{R^p}{S^q}$ is in ~$\idsr$ 
(and hence $\ui{S^q}{R^p}$ is in ~$\idsr$ by the reversibility assumption).
  We write $\eqidsclass{R^p}$ for the $\eqids$-class of~$R^p$.
\end{definition}

As $\idsr$ is transitively closed, $\eqids$ is indeed an equivalence relation.
Our choice of where to assign the helper elements will be represented as a mapping
to $\eqids$-classes. We call the result a \defo{partially-specified superinstance}, or
\defo{pssinstance}:

\begin{definition}
  \label{def:completion}
  A \deft{pssinstance}
  of an instance $I$ is a triple $P = (I, \calH,
  \mapb)$ where $\calH$ is a finite set of \deft{helpers} and $\mapb$ maps
  each $h \in \calH$ to an
  $\eqids$-class $\mapb(h)$.

  We define $\appelem{P}{R^p} \defeq \appelem{I}{R^p}
  \sqcup \{h \in \calH \mid R^p \in \mapb(h)\}$.
\end{definition}

In other words, in the pssinstance, elements of~$I$ want to appear at the same
positions as before, and helper elements want to occur at their $\eqids$-class
according to~$\mapb$.
A \defo{realization} of a pssinstance~$P$ is then a superinstance of its underlying
instance~$I$ which adds the helper elements, and whose additional facts respect $\appelem{P}{R^p}$:

\begin{definition}
  \label{def:realization}
  A \deft{realization} of~$P = (I, \calH, \mapb)$ is
  a superinstance $I'$ of~$I$ such that $\dom(I') =
  \dom(I) \sqcup \calH$, and, for any fact $R(\mybf{a})$ of~$I' \backslash I$ and
  $R^p \in \pos(R)$, we have $a_p \in \appelem{P}{R^p}$.
\end{definition}

\begin{example}
  \label{exa:helper2}
  In Example~\ref{exa:helper}, a pssinstance of $I_0$ is $P \defeq (I_0, \{h\}, \mapb)$
  where $\mapb(h) \defeq \{S^2, T^1\}$. Further, it is balanced. For instance,
  $\appelem{P}{S^1} = \{b\}$ and $\appelem{P}{S^2} = \{h\}$. The instance $I$ in
  Example~\ref{exa:helper} is a realization of~$P$.
\end{example}

It is easy to see that realizations of pssinstances are
weakly-sound:

\begin{lemma}[Binary realizations are completions]
  \label{lem:wsrwss}
  If $I'$ is a realization of a pssinstance of~$I_0$ then it is a
  weakly-sound superinstance of~$I_0$.
\end{lemma}

\begin{proof}
Consider $a \in \dom(I')$ and $R^p \in \pos(\sigma)$ such that $a \in
\pi_{R^p}(I')$. As $I'$ is a realization, we know that either $a
\in \pi_{R^p}(I)$ or $a \in \appelem{P}{R^p}$. By definition of
$\appelem{P}{R^p}$, and because $\calH = \dom(I') \backslash \dom(I)$, this means
that either $a \in \dom(I)$ and $a \in \pi_{R^p}(I) \sqcup \appelem{I}{R^p}$, or
$a \in \dom(I') \backslash \dom(I)$ and $R^p \in \mapb(a)$.
Hence, let us check from the definition that $I'$ is weakly-sound:
\begin{itemize}
  \item For any $a \in \dom(I)$ and $R^p \in \pos(\sigma)$, we have shown
that $a \in \pi_{R^p}(I')$ implied that either $a \in \pi_{R^p}(I)$ or $a \in
\appelem{I}{R^p}$.
  \item For any $a \in \dom(I') \backslash \dom(I)$ and for any $R^p, S^q \in
    \pos(\sigma)$, we have shown that $a \in \pi_{R^p}(I')$ and $a \in
    \pi_{S^q}(I')$ implies that $R^p, S^q \in \mapb(a)$, so that $R^p
\eqids S^q$, hence $R^p = S^q$ or $\ui{R^p}{S^q}$ is in~$\idsr$.\qedhere
\end{itemize}
\end{proof}

In the next subsection, we will show that we can construct pssinstances that
are \emph{balanced} in a sense that we will define, and show that we can
construct realizations for these pssinstances.

\subsection{Putting it Together}

What remains to show to conclude the proof of Theorem~\ref{thm:binary} is that
we can construct a \emph{balanced} pssinstance of~$I_0$,
even when $I_0$ itself is not balanced. By a \defo{balanced}
pssinstance, we mean the exact analogue of Definition~\ref{def:balanced} for
pssinstances.

\begin{definition}
  \label{def:balancedpss}
  A pssinstance $P = (I, \calH, \mapb)$ is \deft{balanced} if for every two
  positions $R^p$ and $R^q$ such that $R^p \rightarrow R^q$ and $R^q \rightarrow
  R^p$ are in~$\ufds$, we have $\card{\appelem{P}{R^p}} = \card{\appelem{P}{R^q}}$.
\end{definition}

Again, the definition is phrased in a general way so as to be
usable later without making the arity-two assumption.
If $I_0$ is balanced, the empty pssinstance $(I, \emptyset, \mapb)$, with
$\mapb$ the empty function, is a balanced pssinstance of~$I_0$, and we could
just complete $I_0$ as we presented before. We now show that, even if $I_0$ is not
balanced, we can always construct a balanced pssinstance, thanks to
the helpers:

\begin{lemma}[Balancing]
  \label{lem:hascompletion}
  Any finite instance~$I$ satisfying~$\ufds$ has a
  balanced pssinstance.
\end{lemma}

In fact, this lemma does not use the arity-two assumption. We will 
reuse it in the next \secname.

\begin{proof}
  Let $I$ be a finite instance.
  For any position $R^p$,
  define $o(R^p) \defeq \appelem{I}{R^p} \sqcup
\pi_{R^p}(I)$, i.e., the elements that either appear at $R^p$
or want to appear there. We show that $o(R^p) = o(S^q)$ whenever $R^p \eqids
S^q$, which is obvious if $R^p = S^q$, so assume $R^p \neq S^q$.
First, we have $\pi_{R^p}(I) \subseteq o(S^q)$: elements in $\pi_{R^p}(I)$
want to appear at $S^q$ unless they already do, and in both cases they are in
$o(S^q)$. Second, elements of
$\appelem{I}{R^p}$ either occur at $S^q$, or at some other position $T^r$
such that $\ui{T^r}{R^p}$ is a $\uid$ of~$\idsr$, so that by transitivity $T^r =
S^q$ or
$\ui{T^r}{S^q}$
also holds, and so they want to be at $S^q$ or they already are. Hence
$o(R^p) \subseteq o(S^q)$; and symmetrically $o(S^q) \subseteq o(R^p)$.
Thus, the set $o(R^p)$ only depends on the $\eqids$-class of~$R^p$.

Let $N \defeq \max_{R^p \in \pos(\sigma)} \card{o(R^p)}$, which is finite.
We define for each $\eqids$-class
$\eqidsclass{R^p}$ a set $p(\eqidsclass{R^p})$ of~$N - \card{o(R^p)}$ fresh
helpers. We let $\calH$ be the disjoint union of the $p(\eqidsclass{R^p})$ for all
classes $\eqidsclass{R^p}$, and set $\mapb$ to map the elements of
$p(\eqidsclass{R^p})$ to $\eqidsclass{R^p}$. We have thus defined
a pssinstance $P = (I, \calH, \mapb)$.

Let us now show that $P$ is balanced. Consider now two positions $R^p$ and
$R^q$ such that $\phi: R^p \rightarrow R^q$ and $\phi^{-1}: R^q \rightarrow R^p$
are in
$\ufds$, and show that $\card{\appelem{P}{R^p}} = \card{\appelem{P}{R^q}}$. We
have $\card{\appelem{P}{R^p}} = \card{\appelem{I}{R^p}} +
\card{p(\eqidsclass{R^p})} = \card{o(R^p)} -
\card{\pi_{R^p}(I)} + N - \card{o(R^p)}$, which simplifies to
$N - \card{\pi_{R^p}(I)}$.
Similarly $\card{\appelem{P}{R^q}} = N - \card{\pi_{R^q}(I)}$.
Since  $I \models \ufds$ and $\phi$ and $\phi^{-1}$ are in $\ufds$
we know that
$\card{\pi_{R^p}(I)} = \card{\pi_{R^q}(I)}$.
Hence, $P$ is balanced, as we claimed.
\end{proof}

We had seen in Proposition~\ref{prp:balwsnd} that we could construct a
weakly-sound superinstance of a balanced $I_0$ by pairing together elements. We
now generalize this claim to the balanced pssinstances that we constructed,
showing that we can build realizations of balanced pssinstances that satisfy
$\uconr$:

\begin{lemma}[Binary realizations]
  \label{lem:balwsndcb}
  For any balanced pssinstance $P$ of an instance~$I$ which satisfies
  $\ufds$, we can construct a realization of
  $P$ that satisfies $\uconr$.
\end{lemma}

\begin{proof}
As in Proposition~\ref{prp:balwsnd},
for every relation $R$,
construct a bijection $f_R$ between $\appelem{P}{R^1}$ and $\appelem{P}{R^2}$:
this is possible, as $P$ is balanced.
We then construct our realization $I'$ as in Proposition~\ref{prp:balwsnd}:
 we add to $I$ the fact $R(a, f_R(a))$ for every $R$
of~$\sigma$ and every $a \in \appelem{P}{R^1}$.

We prove that $I'$ is a realization as in Proposition~\ref{prp:balwsnd} by observing that whenever we create a
fact $R(a, f_R(a))$, then we have $a \in \appelem{P}{R^1}$ and $f_R(a) \in
\appelem{P}{R^2}$.
Similarly, we show that $I' \models \ufds$ as in Proposition~\ref{prp:balwsnd}.

We now show that $I'$ satisfies $\idsr$. Assume to the contrary that there is an
active fact $F = R(a_1, a_2)$ that witnesses the violation of a $\uid$ $\tau: \ui{R^p}{S^q}$, so that $a_p
\in \appelem{I'}{S^q}$. If $a_p \in
\dom(I)$, then the proof is exactly as for Proposition~\ref{prp:balwsnd}. Otherwise, if $a_p \in \calH$,
clearly by construction of~$f_R$ and $I'$ we have $a_p \in \pi_{T^r}(I')$ iff
$T^r \in \mapb(a_p)$. Hence, as $a_p \in \pi_{R^p}(I')$ and as $\tau$ witnesses
by the reversibility assumption that $R^p \eqids S^q$ , we
have $a_p \in \pi_{S^q}(I')$, contradicting the fact that $a_p \in
\appelem{I'}{S^q}$.
\end{proof}

We now conclude the proof of Theorem~\ref{thm:binary}. Given the instance $I_0$,
construct a balanced pssinstance $P$
with the Balancing Lemma (Lemma~\ref{lem:hascompletion}),
construct a realization $I'$ of~$P$ that satisfies~$\uconr$ 
with the Binary Realizations Lemma (Lemma~\ref{lem:balwsndcb}),
and conclude 
by the ``Binary Realizations are Completions'' Lemma (Lemma~\ref{lem:wsrwss})
that $I'$ is a weakly-sound superinstance of~$I_0$.

\section{Weak Soundness on Arbitrary Arity Signatures}
\label{sec:weakarb}
We now lift the arity-two assumption and extend the results to arbitrary arity
signatures:
\begin{maintheorem}
  \label{thm:piecewise}
  Reversible $\uid$s and $\ufd$s have finite weakly-universal models
  for $\acq$s.
\end{maintheorem}

A first complication when lifting the arity-two assumption is that
realizations cannot be created just by pairing two elements. To
satisfy the $\uid$s
 we may have to create facts that connect elements on more
than two positions, so we may need more than the bijections between two
positions that we used before.
A much more serious problem is that the positions where we connect together
elements
may still be only a subset of the positions of the
relation, which means that the other positions must be filled somehow.

We address these difficulties by defining first \emph{piecewise realizations},
which create partial facts on positions connected by $\ufd$s, similarly to the
previous section.
We show that we can get piecewise realizations by generalizing the
Binary Realizations Lemma (Lemma~\ref{lem:balwsndcb}).
Second, to find elements to reuse at other positions,
we define a notion of \emph{saturation}. We show
that, by an initial \emph{saturation
process}, we can ensure that there are existing elements that we can reuse at
positions where this will not violate $\ufd$s (the \emph{non-dangerous
positions}). Third, we define a notion of \emph{thrifty chase step} to solve
$\uid$ violations one by one. We last explain how to use 
thrifty chase steps to solve all
$\uid$ violations on saturated instances, using a piecewise realization as a template;
this is how we construct our
weakly-sound completion.

As in the previous section, we fix the instance $I_0$, reversible
constraints $\uconr$ formed of $\uid$s $\idsr$ and $\ufd$s $\ufds$, and assume 
that $I_0 \models \ufds$.

\subsection{Piecewise Realizations}

Without the arity-two assumption, we must define a new equivalence relation to
reflect the $\ufd$s, in addition to $\eqids$ which reflects the $\uid$s:

\begin{definition}
  \label{def:eqfun}
  For any two positions $R^p$ and $R^q$, we write $R^p \eqfun R^q$ whenever
  $R^p = R^q$ or $R^p
  \rightarrow R^q$ and $R^q \rightarrow R^p$ are both in~$\ufds$.
\end{definition}

By transitivity of~$\ufds$, $\eqfun$ is indeed an equivalence relation.

The definition of \defo{balanced instances} (Definition~\ref{def:balanced})
generalizes as-is to arbitrary arity.
We do not change the definition of \defo{pssinstance}
(Definition~\ref{def:completion}), and talk of them being balanced 
(Definition~\ref{def:balancedpss})
in the same
way.
Further, we know that the Balancing Lemma (Lemma~\ref{lem:hascompletion})
holds even without the arity-two assumption.

Our general scheme is the same: construct a balanced pssinstance of~$I_0$, and
use it to construct the completion~$I$.
What we need is to change the notion of \defo{realization}. We replace it
by \defo{piecewise realizations}, which are defined on $\eqfun$-classes.
We number the $\eqfun$-classes of
$\pos(\sigma)$ as $\Pi_1, \ldots, \Pi_{\neqfunc}$
and define \defp{piecewise instances} by their projections to the~$\Pi_i$:

\begin{definition}
  \label{def:piecewise}
  A \deft{piecewise instance}
  is an $\neqfunc$-tuple $\pire = (K_1, \ldots, K_{\neqfunc})$,
  where each $K_i$ is a set of $\card{\Pi_i}$-tuples, indexed by~$\Pi_i$ for
  convenience.
  The \deft{domain} of~$\pire$ is $\dom(\pire) \defeq \bigcup_i \dom(K_i)$.
  For $1 \leq i \leq \neqfunc$ and $R^p \in \Pi_i$, we define
  $\pi_{R^p}(\pire) \defeq \pi_{R^p}(K_i)$.
\end{definition}

We will realize a pssinstance~$P$, not as an instance as in the previous \secname, but
as a piecewise instance.
The tuples in each $K_i$ will be defined from~$P$, and will connect elements
that want to occur at the corresponding position in $\Pi_i$, generalizing the 
ordered pairs constructed with bijections in the proof of 
the Binary Realizations Lemma (Lemma~\ref{lem:balwsndcb}).
Let us define accordingly the notion of a \defo{piecewise realization} of a pssinstance
as a piecewise instance:

\begin{definition}
  \label{def:superbvpw}
  A piecewise instance $\pire = (K_1, \ldots, K_{\neqfunc})$ is a
  \deft{piecewise realization} of the pssinstance $P = (I, \calH, \mapb)$ if:
  \begin{itemize}
    \item $\pi_{\Pi_i}(I) \subseteq K_i$ for all $1 \leq i \leq \neqfunc$,
    \item $\dom(\pire) = \dom(I) \sqcup \calH$, 
    \item for all $1 \leq i \leq \neqfunc$, for all $R^p
  \in \Pi_i$, for every tuple $\mybf{a} \in K_i \backslash \pi_{\Pi_i}(I)$, 
  we have $a_p \in \appelem{P}{R^p}$. \qedef
\end{itemize}
\end{definition}
Notice that the definition is similar to the conditions imposed on realizations
(Definition~\ref{def:realization}), although piecewise realizations are
piecewise instances, not actual instances; so we will need one extra step to make
real instances out of them: this is done in Section~\ref{sec:completion}.

We must now
generalize the Binary Realizations Lemma (Lemma~\ref{lem:balwsndcb}) to
construct these piecewise realizations out of balanced pssinstances. For this,
we need to define what it means for a piecewise
instance $\pire$ to ``satisfy'' $\uconr$. For $\ufds$, we require that $\pire$ respects the
$\ufd$s within each $\eqfun$-class. For
$\idsr$, we define it directly from the projections of~$\pire$.

\begin{definition}
  A piecewise instance $\pire$ is \deft{$\bm{\ufds}$-compliant} if, for all $1 \leq i \leq n$, there
  are no two tuples $\mybf{a} \neq \mybf{b}$ in~$K_i$
  such that $a_p = b_p$ for some
  $R^p \in \Pi_i$.

  $\pire$ is \deft{$\bm{\idsr}$-compliant} if $\appelem{\pire}{\tau} \defeq \pi_{R^p}(\pire) \backslash
  \pi_{S^q}(\pire)$ is empty for all $\tau : \ui{R^p}{S^q}$ in~$\idsr$.

  $\pire$ is \deft{$\bm{\uconr}$-compliant} if it is $\ufds$- and $\idsr$-compliant.
\end{definition}

We can then state and prove the generalization of the Binary Realizations Lemma:

\begin{lemma}[Realizations]
  \label{lem:balwsndc}
  For any balanced pssinstance~$P$ of an instance $I$ that satisfies
  $\ufds$, we can construct a piecewise
  realization of $P$
  which is $\uconr$-compliant.
\end{lemma}

Before we prove the Realizations Lemma, we show a simple example:

\begin{example}
  Consider a $4$-ary relation $R$ and the $\uid$s $\tau: \ui{R^1}{R^2}$,
  $\tau': \ui{R^3}{R^4}$ and their reverses, and the $\ufd$s $\phi: R^1 \rightarrow
  R^2$, $\phi' : R^3 \rightarrow R^4$ and their reverses. We have
  $\Pi_1 = \{R^1, R^2\}$ and $\Pi_2 = \{R^3, R^4\}$.
  Consider $I_0 \defeq \{R(a,
  b, c, d)\}$, which is balanced, and the trivial balanced pssinstance
  $P \defeq (I_0, \emptyset, \mapb)$,
  where $\mapb$ is the empty function. A $\uconr$-compliant piecewise realization of~$P$
  is $\pire \defeq (\{(a, b), (b, a)\}, \{(c, d), (d, c)\})$.
\end{example}

We conclude the \subsecname with the proof of the Realizations Lemma:
\begin{proof}[\proofof Lemma~\ref{lem:balwsndc}]
Let $P = (I, \calH, \mapb)$ be the balanced pssinstance.
Recall
that the $\eqfun$-classes of~$\sigma$ are numbered $\Pi_1, \ldots,
\Pi_{\neqfunc}$.
By definition of~$P$ being balanced (Definition~\ref{def:balanced} applied
to arbitrary arity),
for any $\eqfun$-class $\Pi_i$,
for any two positions $R^p, R^q \in \Pi_i$,
we have $\card{\appelem{P}{R^p}} = \card{\appelem{P}{R^q}}$.
Hence, for all $1 \leq i \leq \neqfunc$,
let us write $s_i$ to denote the value of~$\card{\appelem{P}{R^p}}$
for some $R^p \in \Pi_i$.

For $1 \leq i \leq \neqfunc$, we let $m_i$ be~$\card{\Pi_i}$, and
number the positions of $\Pi_i$ as $R^{p^i_1}, \ldots, R^{p^i_{m_i}}$.
We choose for each $1 \leq i \leq \neqfunc$ and $1 \leq j \leq m_i$
an arbitrary bijection $\phi^i_j$
from $\{1, \ldots, s_i\}$ to $\appelem{P}{R^{p^i_j}}$.
We construct the piecewise realization $\pire = (K_1, \ldots, K_n)$
by setting each $K_i$ for $1 \leq i \leq \neqfunc$
to be $\pi_{\Pi_i}(I)$ plus the tuples
$(\phi^i_1(l), \ldots, \phi^i_{m_i}(l))$ for $1 \leq l \leq s_i$.

\medskip

It is clear that $\pire$
is a piecewise realization. Indeed, the first two conditions are
immediate. Further, whenever we create a tuple
$\mybf{a} \in K_i$ for any $1 \leq i \leq \neqfunc$,
then, for any $R^p \in \Pi_i$, we have
$a_p \in \appelem{P}{R^p}$.

\medskip

Let us then show that $\pire$ is $\ufds$-compliant.
Assume by contradiction that
there is $1 \leq i \leq \neqfunc$
and $\mybf{a}, \mybf{b} \in K_i$ such
that $a_l = b_l$ but $a_r \neq b_r$ for some $R^l, R^r \in \Pi_i$.
As $I$ satisfies $\ufds$, we assume without loss of generality that
$\mybf{a} \in K_i \backslash \pi_{\Pi_i}(I)$. Now either $\mybf{b} \in
\pi_{\Pi_i}(I)$ or $\mybf{b} \in K_i \backslash \pi_{\Pi_i}(I)$.
\begin{itemize}
  \item If $\mybf{b} \in \pi_{\Pi_i}(I)$,
then $b_l \in \pi_{R^l}(I)$. Yet, we know by
construction that, as $\mybf{a} \in K_i \backslash \pi_{\Pi_i}(I)$, we 
have $a_l \in \appelem{P}{R^l}$, and as $a_l = b_l$ we have $a_l \in
    \appelem{P}{R^l}$, which contradicts the fact that $\mybf{b} \in
    \pi_{\Pi_i}(I)$.
  \item If $\mybf{b} \in K_i \backslash \pi_{\Pi_i}(I)$, then,
writing $R^l = R^{p^i_j}$ and $R^r = R^{p^i_{j'}}$,
the fact that $a_l = b_l$ but $a_r \neq b_r$
contradicts the fact that $\phi^i_{j} \circ (\phi^i_{j'})^{-1}$ is injective.
\end{itemize}
Hence, $\pire$ is $\ufds$-compliant.

\medskip

Let us now show that $\pire$ is $\idsr$-compliant.
We must show
that, for every $\uid$
$\tau: \ui{R^p}{S^q}$ of~$\idsr$, we have $\appelem{\pire}{\tau} = \emptyset$, which means
that we have $\pi_{R^p}(\pire) \subseteq \pi_{S^q}(\pire)$. Let
$\Pi_i$ be the $\eqfun$-class of~$R^p$, and assume to the contrary the existence of a
tuple $\mybf{a}$ of~$K_i$ such that $a_p \notin \pi_{S^q}(\pire)$.
Either we have $a_p \in \dom(I)$ or we have $a_p \in \calH$:

\begin{itemize}
  \item 
    If $a_p \in \dom(I)$,
    then we have $a_p \notin \pi_{S^q}(\pire)$,
in particular $a_p \notin \pi_{S^q}(I)$.
    Now there are two subcases: either $\mybf{a} \in \pi_{\Pi_i}(I)$, and
    then we have $a_p \in \pi_{R^p}(I)$; or $\mybf{a} \in K_i \setminus
    \pi_{\Pi_i}(I)$, in which case we know that $a_p \in \appelem{P}{R^p}$, so
    that as $a_p \in \dom(I)$ we have $a_p \in \appelem{I}{R^p}$. In both subcases, 
$\tau$ witnesses that
$a_p \in \appelem{I}{S^q}$.
By construction of~$\pire$, then,
letting $\Pi_{i'}$ be the $\eqfun$-class of~$S^q$ and letting $S^q = S^{p^{i'}_j}$,
as $\phi^{i'}_j$ is surjective,
we must have $a_p \in \pi_{S^q}(K_{i'})$,
that is, $a_p \in \pi_{S^q}(\pire)$, a contradiction.

    \item If $a_p \in \calH$, then clearly by construction we have $a_p \in \pi_{T^r}(\pire)$
iff $T^r \in \mapb(a_p)$,
so that, given that $\tau$ witnesses $R^p \eqids S^q$,
if $a_p \in \pi_{R^p}(\pire)$ then $a_p \in \pi_{S^q}(\pire)$,
a contradiction.
\end{itemize}
We conclude that $\pire$ is indeed a $\uconr$-compliant piecewise realization
of~$P$.
\end{proof}

\subsection{Relation-Saturation}

The Realizations Lemma (Lemma~\ref{lem:balwsndc})
gives us a $\uconr$-compliant piecewise
realization which is a piecewise instance. To construct an
actual superinstance from it, we will have to expand each tuple $\mybf{t}$
of each~$K_i$, defined on the $\eqfun$-class $\Pi_i$, to an entire fact
$F_{\mybf{t}}$ of the corresponding relation.

However, to fill the other positions of $F_{\mybf{t}}$, we will need to
reuse existing elements of~$I_0$.
To do this, it is easier to assume that $I_0$ contains some $R$-fact
for every relation $R$ of the signature.

\begin{definition}
  \label{def:wsaturated}
  A superinstance $I$ of $I_0$ is \deft{relation-saturated}
  if for every $R \in \sigma$ there is an $R$-fact in $I$.
\end{definition}

We illustrate why it is easier to work with
relation-saturated instances:

\begin{example}
  Suppose our schema has two binary relations $R$ and $T$ and a unary relation $S$, the
  $\uid$s $\tau: \ui{S^1}{R^1}$, $\tau': \ui{R^2}{T^1}$ and their reverses, and no $\ufd$s. Consider the
  non-relation-saturated instance $I_0 \defeq \{S(a)\}$. It is
  balanced, so
  $P \defeq (I_0, \emptyset, \mapb)$, with $\mapb$ the empty function, is a
  pssinstance of~$I$.

  Now, a $\uconr$-compliant piecewise realization of~$P$ is $\pire = (K_1,
  \ldots, K_5)$ with $K_2 = K_4 = K_5 = \emptyset$ and $K_1 = K_3 = \{a\}$,
  where $\Pi_1$ and $\Pi_3$ are the $\eqfun$-classes of~$R^1$ and~$S^1$.
  However, we cannot easily complete $\pire$ to an actual superinstance of $I_0$
  satisfying~$\tau$ and~$\tau'$. Indeed,
  to create the fact $R(a, \bullet)$, as indicated by~$K_1$, we need to fill
  position $R^2$. Using an existing element would violate weak-soundness, and
  using a fresh element would introduce a violation
  of~$\tau'$, which $P$ and $\pire$ would not tell us how to solve.

  Consider instead the relation-saturated instance $I_1 \defeq I_0 \sqcup
  \{S(c), R(c, d), T(d ,e)\}$.
  We can complete~$I_1$ to a weakly-sound superinstance that satisfies $\tau$
  and $\tau'$, by adding the
  fact $R(a, d)$. Observe how we reused $d$ to fill position $R^2$: this
  does not violate weak-soundness or introduce new $\uid$ violations.
\end{example}

Relation-saturation can clearly be ensured by initial chasing, which
does not violate weak-soundness. We call this a \emph{saturation process} to
ensure relation-saturation:

\begin{lemma}[Relation-saturated solutions]
  \label{lem:wtsol}
  For any reversible $\uid$s $\idsr$, $\ufd$s $\ufds$, and instance $I_0$ satisfying
  $\ufds$, there exists a finite number $n \in \NN$ such that the result of
  performing $n$ chase rounds on~$I_0$ by $\idsr$ is
  a weakly-sound relation-saturated superinstance of~$I_0$ that satisfies
  $\ufds$.
\end{lemma}

This allows us to assume that $I_0$ was preprocessed with initial chasing if
needed, so we can assume it to be relation-saturated.
To show the lemma, and also for further use, we make a simple observation on
weak-soundness:

\begin{lemma}[Weak-soundness transitivity]
  \label{lem:wsndtrans}
  If $I'$ is a weakly-sound superinstance of~$I$, and $I$ is a weakly-sound
  superinstance of~$I_0$, then $I'$ is a weakly-sound superinstance of~$I_0$.
\end{lemma}

\begin{proof}
  Let $a \in \dom(I')$, and let us show that it does not witness a violation of
  the weak-soundness of~$I'$ for~$I_0$. We distinguish three cases:
  \begin{itemize}
    \item If $a \in \dom(I_0)$, then in particular $a \in \dom(I)$. Hence,
      letting $S^q$ be any position such that $a \in \pi_{S^q}(I')$, as
      $I'$ is a weakly-sound superinstance of~$I$, either $a \in
      \pi_{S^q}(I)$ or we have $a \in
      \appelem{I}{S^q}$. Let $R^p$ be a position such that $a \in
      \pi_{R^p}(I)$, and such that $R^p = S^q$ (in the first case) or $\ui{R^p}{S^q}$
      is in~$\idsr$ (in the second case). As $I$ is a weakly-sound
      superinstance of~$I_0$, either $a \in \pi_{R^p}(I_0)$ or $a \in
      \appelem{I_0}{R^p}$. As $\idsr$ is transitively closed, we conclude that
      $a \in \appelem{I_0}{S^q}$ or $a \in \pi_{S^q}(I_0)$.
    Hence, the fact that $a$ occurs at position $S^q$ in~$I'$
    does not cause a violation of weak-soundness in~$I'$ for~$I_0$.
  \item If $a \in \dom(I) \backslash \dom(I_0)$, we must show that for any two
    positions $R^p$, $S^q$ where $a$ occurs in $I'$, we have $R^p
    \eqids S^q$. Let us fix two such positions, i.e., we have $a \in
    \pi_{R^p}(I')$ and $a \in \pi_{S^q}(I')$. As $I'$ is a weakly-sound
    superinstance of~$I$, 
    we have either $a \in \pi_{R^p}(I)$ or $a \in \appelem{I}{R^p}$, and
    we have either $a \in \pi_{S^q}(I)$ or $a \in \appelem{I}{S^q}$.
    As in the previous case, let $T^v$ and $U^w$ be positions such that $a \in
    \pi_{T^v}(I)$
    and $a \in \pi_{U^w}(I)$, and $T^v = R^p$ or the $\uid$ $\tau: \ui{T^v}{R^p}$
    is in~$\idsr$, and $U^w = S^q$ or the $\uid$ $\tau': \ui{U^w}{S^q}$ is
    in~$\idsr$. As $I$ is a weakly-sound superinstance of~$I_0$, and $a \notin
    \dom(I_0)$, we know that $T^v \eqids U^w$. By the reversibility assumption
    and as $\idsr$ is transitively closed, we deduce (using $\tau$ and $\tau'$
    if necessary) that $R^p \eqids S^q$, which is what we wanted to show. Hence,
    the fact that $a$ occurs at positions $R^p$ and $S^q$ in~$I'$ does not cause
    a violation of weak-soundness in~$I'$ for~$I_0$.
  \item If $a \in \dom(I') \backslash \dom(I)$, then from the fact that $I'$ is
    a weakly-sound superinstance of~$I$, we deduce immediately about $a$ what is
    needed to show that it does not witness a violation of the weak-soundness of
    $I'$ for $I_0$.
  \end{itemize}
  So we conclude that $I'$ is a weakly-sound instance of $I_0$, as desired.
\end{proof}

We conclude the \subsecname
by proving the Relation-Saturated Solutions Lemma (Lemma~\ref{lem:wtsol}):

\begin{proof}[\proofof Lemma~\ref{lem:wtsol}]
  Remember that the signature $\sigma$ was assumed without loss of generality
  not to contain any useless relation. Hence, for every relation $R \in \sigma$,
  there is an $R$-fact in $\chase{I_0}{\idsr}$,
  which was generated at the $n_R$-th round of the chase, for some $n_R \in \NN$.
  Let $n \defeq \max_{R \in \sigma} n_R$, which is finite because the number of
  relations in $\sigma$ is finite.
  We take $I$ to be the result of applying
  $n$ chase rounds
  to~$I_0$.
  
  It is clear that $I$ is relation-saturated. The fact that $I$ is
  weakly-sound is
  by the Weak-Soundness Transitivity Lemma (Lemma~\ref{lem:wsndtrans}), 
  because each chase step
  clearly preserves weak-soundness: the exported element occurs at a position
  where it wants to occur, so we can use the reversibility assumption
  and new elements only occur at one position.
\end{proof}

\subsection{Thrifty Chase Steps}

We have explained why $I_0$ can be assumed to be relation-saturated, and we know
we can build a $\uconr$-compliant piecewise realization $\pire$ of a balanced
pssinstance. Our goal is now to satisfy the $\uid$s using $\pire$. We
will do so by a \emph{completion process} that fixes each violation one by
one, following $\pire$. This \subsecname presents the tool that we use
for this, and the next \subsecname describes the actual process.

Our tool is a form of chase step, a \defo{thrifty chase step}, which adds a new
fact~$F_{\fnew}$ to satisfy a $\uid$
violation. For some of the positions,
the elements of~$F_{\fnew}$ will be defined from the
realization $\pire$, using one of its tuples.
For each of these elements, either $F_{\fnew}$ makes them occur at a position that they
want to be (thus satisfying another violation) or
these elements are helpers that did not occur already in the domain.
At any other position $S^r$ of $F_{\fnew}$, we may
either reuse an existing element (by relation saturation, one can always
reuse an element that already occurs in that position) 
or create a
fresh element (arguing that no $\uid$ will be violated on that element).
This depends on whether $S^r$ is \defo{dangerous}
or \defo{non-dangerous}:

\begin{definition}
  \label{def:dangerdef}
  We say a position $S^r \in \pos(\sigma)$ is \deft{dangerous} for the position
  $S^q \neq S^r$ if $S^r \rightarrow S^q$ is in~$\ufds$, and write $S^r \in \danger(S^q)$.
  Otherwise, still assuming $S^q \neq S^r$ $S^r$ is
  \deft{non-dangerous} for~$S^q$, written $S^r \in \nondanger(S^q)$. Note that
  $\{S^q\} \sqcup \danger(S^q) \sqcup \nondanger(S^q) = \pos(S)$.
\end{definition}

We can now define \defo{thrifty chase steps}. The details of the definition are
designed for the completion process defined in the next \subsecname
(Proposition~\ref{prp:rtcompr}), and for the
specialized notions that we will introduce later in this \subsecname as well as
in the following \secnames.

\begin{definition}
  \label{def:thrifty}
  Let $I$ be a superinstance of~$I_0$, let $\tau : \ui{R^p}{S^q}$
  be a $\uid$ of~$\idsr$, and let $F_{\factive} = R(\mybf{a})$ be an active fact for
  $\tau$ in $I$. We call $S^q$ the \deft{exported position}, and write $\Pi_i$
  for its $\eqfun$-class.
  
  Applying a \deft{thrifty chase step} to~$F_{\factive}$ (or~$a$) in~$I$ by~$\tau$ yields a 
  superinstance $I'$ of~$I_0$ which is $I$ plus a single new fact $F_{\fnew} =
  S(\mybf{b})$. We require the following on $b_r$ for all $S^r \in \pos(S)$:
  
  \begin{itemize}
  \item For $S^r = S^q$, we require $b_q = a_p$ and $b_q \in \appelem{I}{\tau}$;
  \item For $S^r \in \Pi_i \backslash \{S^q\}$, we require that one of the following holds:
    \begin{itemize}
      \item $b_r \in \appelem{I}{S^r}$;
      \item $b_r \notin \dom(I)$ and for all $S^s \in \Pi_i$,
    such that $b_r = b_s$, we have $S^r \eqids S^s$;
    \end{itemize}
  \item For $S^r \in \danger(S^q) \backslash \Pi_i$, we require $b_r$ to be
    fresh and occur only at that position;
  \item For $S^r \in \nondanger(S^q)$, we require that $b_r \in \pi_{S^r}(I)$.
    \qedef
  \end{itemize}
\end{definition}

Thrifty chase steps eliminate $\uid$ violations on the element at the exported
position $S^q$ of the new fact (which
is why we call them ``chase steps''),
and also eliminate violations on positions in the same $\eqfun$-class as
$S^q$, unless a fresh element is used there.
The completion process that we will define in the next \subsecname will
\emph{only} apply thrifty chase steps (namely, \emph{relation-thrifty steps},
which we will define shortly), and indeed
this will be true of \emph{all} completion processes used in \thispaper.

For now, we can observe that thrifty chase steps cannot break weak-soundness:

\begin{lemma}[Thrifty preserves weak-soundness]
  \label{lem:thriftyp}
  For any weakly-sound superinstance $I$ of an instance $I_0$, letting $I'$ be
  the result of applying a thrifty chase step on~$I$, we have that $I'$ is a weakly-sound
  superinstance of~$I_0$.
\end{lemma}

\begin{proof}
  By the Weak-Soundness Transitivity Lemma (Lemma~\ref{lem:wsndtrans}),
  it suffices to show that $I'$ is a
  weakly-sound superinstance of~$I$. It suffices to check this
  for the elements occurring in the one fact $F_{\fnew}
  = S(\mybf{b})$ of $I' \backslash I$, as the other elements occur at the same
  positions as before. Let us show for each $b_r$ for $S^r \in \pos(S)$ that
  $b_r$ does not cause a violation of weak-soundness:
  \begin{itemize}
    \item For $S^r = S^q$, we have $b_r \in \appelem{I}{S^r}$, so $b_r$ does not
      violate weak-soundness;
    \item For $S^r \in \Pi_i \backslash \{S^q\}$, there are two
      possible cases:
      \begin{itemize}
        \item $b_r \in \appelem{I}{S^r}$, so $b_r$ does not violate
          weak-soundness;
        \item $b_r \notin \dom(I)$ and $b_r$ occurs only at positions related by
          $\eqids$, so $b_r$ does not violate weak-soundness;
      \end{itemize}
    \item For $S^r \in \danger(S^q) \backslash \Pi_i$, $b_r$ is fresh and occurs
      at a single position in~$I'$, so $b_r$ does not violate weak-soundness;
    \item For $S^r \in \nondanger(S^q)$, as $b_r \in \pi_{S^r}(I)$, 
      $b_r$ does not violate weak-soundness. \qedhere
  \end{itemize}
\end{proof}

Thrifty chase steps may introduce $\ufd$ violations. For this reason, we
introduce the special case of
\defo{relation-thrifty} chase steps, which can not introduce such
violations. (Relation-thrifty chase steps may still introduce
$\fd$ violations; we will deal with this in \secref{sec:hfds}.)

\begin{definition}[Relation-thrifty]
  A \deft{relation-thrifty chase step} is a thrifty chase step where, reusing
  the notation of Definition~\ref{def:thrifty}, we choose one fact
  $F_\rr = S(\mybf{c})$ of $I$,
  and use $F_\rr$ to define $b_r \defeq c_r$ for all
  $S^r \in \nondanger(S^q)$.
\end{definition}

Remember that relation-saturation ensures that such a fact $S(\mybf{c})$ can
always be found, so clearly any $\uid$ violation can be solved
on a relation-saturated instance by applying some relation-thrifty chase step. 
Further, we can show that relation-thrifty chase steps, unlike thrifty chase
steps, preserve $\ufd$s:

\begin{lemma}[Relation-thrifty preservation]
  \label{lem:rthriftyp}
  For any superinstance $I$ of an instance~$I_0$ such that $I$ satisfies $\ufds$, letting $I'$ be
  the result of applying a relation-thrifty chase step on~$I$, then $I'$ satisfies
  $\ufds$. Further, if $I$ is relation-saturated, then $I'$ is
  relation-saturated.
\end{lemma}

\begin{proof}
  Assume to the contrary the existence of two facts $F = S(\mybf{a})$ and $F'
  = S(\mybf{b})$ in $I'$ that
  witness a violation of some $\ufd$ $\phi: S^r \rightarrow S^p$ of~$\ufds$. As $I
  \models \ufds$, we may assume without loss of generality that $F'$ is
  $F_{\fnew} =
  S(\mybf{b})$, the unique fact of~$I' \backslash I$.
  Write $\tau: \ui{R^p}{S^q}$ the $\uid$ of~$\idsr$ applied in the
  relation-thrifty chase step.
  
  We first note that we must
  have $S^r$ in $\nondanger(S^q)$. Indeed, assuming to the contrary that $S^r
  = S^q$ or $S^r \in
  \danger(S^q)$, the definition of thrifty chase steps requires that either $b_r \notin \dom(I)$ or $b_r \in
  \appelem{I}{S^r}$, so that in either case $b_r \notin \pi_{S^r}(I)$.
  Yet, as $a_r = b_r$,
  $F$ witnesses that $b_r \in \pi_{S^r}(I)$, a contradiction. Thus, $S^r \in
  \nondanger(S^q)$.

  Now, because $\phi$ is
  in $\ufds$ and $\ufds$ is closed under the $\ufd$ transitivity rule,
  unwinding the definitions we can see that $S^p
  \in \nondanger(S^q)$ as well. Now, let $F_\rr = S(\mybf{c})$ be the chosen fact
  for the relation-thrifty chase step. Observe that we must have $F \neq F_\rr$: this follows because we have
  $\pi_{S^r}(F_\rr) = c_q = b_q$ but $\pi_{S^r}(F) = a_q$ and $a_q \neq b_q$ by
  definition of a $\ufd$ violation. Remember now that the definition
  of $F_{\fnew}$ from $F_\rr$ ensures that
  $b_q = c_q$ and $b_r = c_r$. As we also showed that $F \neq F_\rr$, we know
  that $F$ and $F_\rr$ are also a violation of
  $\phi$. But as $F$ and $F_\rr$ are in $I$, this contradicts the fact that $I \models \ufds$.

  The second part of the claim is immediate.
\end{proof}

To summarize: we have defined the general tool used in our completion process, \emph{thrifty
chase steps}, along with a special case that preserves $\ufd$s,
\emph{relation-thrifty chase steps}, which applies to
relation-saturated instances. 
We now
move to the last part of this \secname, where we use this tool to satisfy $\uid$
violations, also using the tools previously defined in this \secname.

\subsection{Relation-Thrifty Completions}
\label{sec:completion}

To prove Theorem~\ref{thm:piecewise},
let us start by taking our initial finite instance~$I_0$, which satisfies~$\ufds$,
and use the Relation-Saturated Solutions Lemma (Lemma~\ref{lem:wtsol})
to obtain a finite weakly-sound
superinstance $I_0'$ which is relation-saturated and still satisfies~$\ufds$.
We now obtain our weakly-sound superinstance from~$I_0'$
by performing a \defo{completion
process} by relation-thrifty chase steps, which we phrase as follows:

\begin{proposition}[Reversible relation-thrifty completion]
  \label{prp:rtcompr}
  For any reversible $\ufds$ and $\idsr$,
  for any finite relation-saturated instance $I_0'$ that satisfies $\ufds$,
  we can use relation-thrifty chase steps
  to construct 
  a finite weakly-sound superinstance $I_\f$ of~$I_0'$ that satisfies $\uconr =
  \idsr \formcup \ufds$.
\end{proposition}

Indeed, once this result is proven, we can immediately conclude the proof of
Theorem~\ref{thm:piecewise} with it, by applying it to~$I_0'$ and obtaining
$I_\f$ which is a weakly-sound superinstance of~$I_0'$, hence of~$I_0$
by the Weak-Soundness Transitivity Lemma (Lemma~\ref{lem:wsndtrans}). So we conclude
the \secname with the proof
of this proposition.

Recall that we number $\Pi_1, \ldots, \Pi_n$ the $\eqfun$-classes of
$\pos(\sigma)$. For two classes $\Pi_i, \Pi_j$  over  a relation $R$, we write $\Pi_i \rightarrow \Pi_j$ to
mean 
that for all $R^p \in \Pi_i$ and $R^q \in \Pi_j$ the $\ufd$ $R^p \rightarrow R^q$ is
in~$\ufds$. Note that this is true iff it is true for some pair of positions (by definition of a $\eqfun$-class and the fact that
$\ufds$ is transitively closed).
We first define the \defo{inner} classes, where creating elements may cause
$\uid$ violations, and the \defo{outer} classes, where this cannot happen
because no position of the class occurs in any $\uid$:

\begin{definition}
  \label{def:inout}
  We say that $\Pi_j$ is an \deft{inner} $\eqfun$-class if it contains a
  position occurring in~$\idsr$; otherwise, it is an \deft{outer} $\eqfun$-class.
\end{definition}

The fundamental property is:

\begin{lemma}
  \label{lem:inout}
  For any $1 \leq i, j \leq n$ with $i \neq j$, if $\Pi_i$ is inner and $\Pi_j \rightarrow \Pi_i$
  then $\Pi_j$ is outer.
\end{lemma}

\begin{proof}
  Assume to the contrary that $\Pi_j$ is inner.
  This means that it contains a position $R^q$
  that occurs in $\idsr$.
  As $\Pi_i$ is inner, pick any $R^p \in \Pi_i$ that occurs in~$\idsr$.
  As $\Pi_j \rightarrow \Pi_i$,
  $\phi: R^q \rightarrow R^p$ is in $\ufds$. Hence, by
  the reversibility assumption, $\phi^{-1}$ also is in $\ufds$. But then we have $R^p
  \eqfun R^q$, contradicting the maximality of $\eqfun$-classes $\Pi_i$ and~$\Pi_j$.
\end{proof}

Let us now start the actual proof
of Proposition~\ref{prp:rtcompr},
and fix the finite relation-saturated instance~$I_0'$ that satisfies $\ufds$.
We start by constructing a balanced pssinstance~$P$
of~$I_0'$ using the Balancing Lemma (Lemma~\ref{lem:hascompletion}),
and a finite $\uconr$-compliant piecewise realization
$\pire = (K_1, \ldots, K_{\neqfunc})$ of~$P$
by the Realizations Lemma (Lemma~\ref{lem:balwsndc}).
Let $\calF$ be an infinite set of fresh elements (not in $\dom(P)$)
from which we will take the (finitely many) fresh elements 
that we will introduce (only at dangerous positions, in outer classes)
during the relation-thrifty chase steps.

We will use $\pire$ to construct a weakly-sound superinstance $I_\f$ by
relation-thrifty chase steps. We
maintain the following invariant when doing so:

\begin{definition}
  \label{def:follows}
  A superinstance $I$ of the instance $I_0'$
  \deft{follows} the piecewise realization $\pire = (K_1, \ldots, K_{\neqfunc})$
  if for every inner $\eqfun$-class $\Pi_i$, we have
  $\pi_{\Pi_i}(I) \subseteq K_i$.
\end{definition}

We prove the Reversible Relation-Thrifty Completion Proposition (Proposition~\ref{prp:rtcompr})
by satisfying $\uid$ violations in~$I_0'$
with relation-thrifty chase steps
using the piecewise realization~$\pire$.
We call $I$ the current state of our superinstance, starting at $I \defeq I_0'$,
and we perform relation-thrifty chase steps on $I$ to satisfy $\uid$ violations,
until we reach a finite weakly-sound superinstance $I_\f$ of~$I_0'$ such that
$I_\f$ satisfies $\idsr$ and
$I_\f$ follows $\pire$. This $I_\f$ will be the final result of the
Reversible Relation-Thrifty Completion Proposition.

Chasing by relation-thrifty chase steps preserves the following invariants:
\begin{description}
  \item[\descindent \invar{sub}:] $I_0' \subseteq I$ (this is clearly monotone);
  \item[\descindent \invar{wsnd}:] $I$ is weakly-sound (by Lemma~\ref{lem:thriftyp});
  \item[\descindent \invar{fun}:] $I \models \ufds$ (by Lemma~\ref{lem:rthriftyp});
  \item[\descindent \invar{rsat}:] $I$ is relation-saturated (by Lemma~\ref{lem:rthriftyp}).
\end{description}
Further, we maintain the following invariants:
\begin{description}
  \item[\descindent \invar{fw}:] $I$ follows $\pire$;
  \item[\descindent \invar{help}:] For a position $R^p$ of an outer class, $\pi_{R^p}(I)$
    and $\calH$ (the set of helper elements), are disjoint.
\end{description}
Let us show that any $\uid$ violation in~$I$ at any stage of the
construction can be solved by applying a relation-thrifty chase step that
preserves these invariants. To show this, let
$a \in
\appelem{I}{\tau}$ be an element to which some $\uid$ $\tau:
\ui{R^p}{S^q}$ of~$\idsr$ is applicable. Let $F_{\factive} = R(\mybf{a})$ be the active fact,
with $a = a_p$.
Let $\Pi_{i}, \Pi_j$ denote the $\eqfun$-classes
of~$R^p$ and~$S^q$ respectively.
The $\uid$ $\tau$ witnesses that $\Pi_{i}$
is inner,
so by invariant \invar{fw} we have $a \in \pi_{R^p}(\pire)$.
As $\pire$ is $\idsr$-compliant, we must have $a \in \pi_{S^q}(\pire)$,
and there is a $\card{\Pi_{j}}$-tuple $\mybf{t} \in K_{j}$
such that $t_q = a$; in fact, by $\ufds$-compliance, there is exactly one such
tuple. What is more, this tuple must be in $K_j \setminus \pi_{\Pi_j}(I)$, as otherwise
we would have $a \in \Pi_{S^q}(I)$, contradicting the applicability of~$\tau$
to~$F_{\factive}$.

Let $F_\rr = S(\mybf{c})$ be an $S$-fact of $I_0'$, which is
possible by invariant \invar{rsat}.
We create a new fact $F_{\fnew} = S(\mybf{b})$ with the relation-thrifty chase step
defined as follows:
\begin{itemize}
  \item For the exported position $S^q$, we set $b_q \defeq a_p$.
  \item For any $S^r \in \Pi_j$, we set $b_r \defeq t_r$.
  \item For any position $S^r \in \danger(S^q) \backslash \Pi_j$,
    we take $b_r$ to be a fresh element $f_r$ from~$\calF$.
  \item For any position $S^r \in \nondanger(S^q)$, we set $b_r \defeq c_r$.
\end{itemize}
We first verify that this satisfies the conditions of thrifty chase steps.
We have set $b_q = a$, and by definition of $F_\rr$ it is immediate that $b_r \in
\pi_{S^r}(I)$ for $S^r \in \nondanger(S^q)$.
For $S^r \in \danger(S^q) \backslash \Pi_j$, we use a fresh element $f_r$ from~$\calF$ which occurs only at
position~$S^r$,
as we should.

The last case to check is for $S^r \in \Pi_j \backslash \{S^q\}$. The first
case is if
$b_r \notin \dom(I)$, in which case we must show that all positions at which
$b_r$ occurs are
$\eqids$-equivalent. Assume that $b_r$ occurs at some other position $S^s \in
\Pi_j$. Now as $b_r$ is in~$\pi_{S^s}(\pire)$, by definition of $\pire$
being a piecewise realization of $P$, we have $b_r \in \appelem{P}{S^s}$. Now,
as $b_r \notin \dom(I)$, by invariant \invar{sub} we also have $b_r \notin
\dom(I_0')$.
But as $b_r \in \dom(\pire)$, we must have $b_r \in \calH$. So by
definition of a pssinstance we have $S^s \in \mapb(b_r)$. Now, observe that we
have $b_r \in
\appelem{P}{S^r}$ because $b_r = t_r$ and $\mybf{t} \in K_j \setminus
\pi_{\Pi_j}(I)$ implies by definition of a piecewise realization that $t_r \in
\appelem{P}{S^r}$.
Hence, we have $S^r \in \mapb(b_r)$. By definition of
$\mapb(b_r)$ being an $\eqids$-class, this means that $S^r
\eqids S^s$, as required.

The second case is $b_r \in \dom(I)$. We will show that we have $b_r \in
\appelem{I}{S^r}$.
Observe first that $b_r \notin \pi_{S^r}(I)$. Indeed, assuming to the contrary that $b_r \in
\pi_{S^r}(I)$, 
let $F = S(\mybf{d})$ be a
witnessing fact in~$I$.
Now, $\tau$ witnesses that $\Pi_j$ is inner, so
by invariant
\invar{fw}, we deduce that
$\pi_{\Pi_j}(\mybf{d}) \in \pi_{\Pi_j}(\pire)$.
Now, as $d_r = t_r$ and $\pire$ is $\ufds$-compliant, we deduce that
$\mybf{d} = \mybf{t}$, so that $F$ witnesses that
$d_q$ is in $\pi_{S^q}(I)$.
As we have $d_q = t_q = a$,
this contradicts the applicability of~$\tau$ to~$a$.
Hence, we have $b_r \notin \pi_{S^r}(I)$.

Second, observe that we have $t_r \in \appelem{P}{S^r}$. Indeed, we have $b_r =
t_r$ which is in $\pi_{S^r}(\pire)$, and we cannot have $\mybf{t} \in
\pi_{\Pi_j}(I)$, as otherwise this would contradict the applicability of
$\tau$ to $a$, as we showed; so in particular, by invariant \invar{sub}, we cannot have
$\mybf{t} \in \pi_{\Pi_j}(I_0')$. Thus, by definition of a piecewise
realization, we have $t_r \in \appelem{P}{S^r}$.

Now, as $t_r \in \appelem{P}{S^r}$, by definition of $\appelem{P}{S^r}$,
there are two cases:
\begin{itemize}
  \item We have $t_r \in \dom(I_0')$ and $t_r \in
\appelem{I_0'}{S^r}$. In this case, as we have shown that $t_r \notin
\pi_{S^r}(I)$, we conclude immediately that $\t_r \in \appelem{I}{S^r}$.

\item We have $t_r \in \calH$ and $S^r \in \mapb(t_r)$. In this case,
consider a fact $F'$ of~$I$ witnessing $t_r \in \dom(I)$, where $t_r$ occurs at a
position $T^l$; let $\Pi_{i'}$ be the $\eqfun$-class of~$T^l$.
As $t_r \in \calH$, by invariant \invar{help}, $\Pi_{i'}$ is inner,
so by invariant \invar{fw} there is a tuple $\mybf{t}'$ of~$K_{i'}$
such that $t'_l = t_r$.
Now, as $t_r \in \calH$,
by definition of piecewise realizations,
we have $T^l \in \mapb(t_r)$.
Hence, either the $\uid$ $\tau' : \ui{T^l}{S^r}$ is in~$\idsr$
or we have $T^l = S^r$.
As $t_r \in \pi_{T^l}(I)$ and we have shown earlier that
$t_r \notin \pi_{S^r}(I)$, we know that $T^l \neq S^r$, so $\tau'$ is in
$\idsr$.
Hence, as $F'$ witnesses that $t_r \in \pi_{T^l}(I)$,
and as $t_r \notin \pi_{S^r}(I)$,
we conclude that $t_r \in \appelem{I}{S^r}$.
\end{itemize}
Hence, in either case we have $t_r \in \appelem{I}{S^r}$, as claimed.
This concludes the proof of the fact that we have indeed defined a
thrifty chase step. Further, the step is clearly relation-thrifty by construction.
The last thing to do is to check that invariants \invar{fw} and \invar{help} are
preserved by the relation-thrifty chase step:
\begin{itemize}
  \item For invariant \invar{fw},
$\tau$ witnesses that the class $\Pi_j$ of~$S^q$ is inner. Hence, for any
$S^r \in \danger(S^q) \backslash \Pi_j$, by Lemma~\ref{lem:inout}, the
$\eqfun$-class of $S^r$ is outer. Thus, to show that \invar{fw} is preserved,
it suffices to show it for the $\eqfun$-class $\Pi_j$ and on the $\eqfun$-classes
included in $\nondanger(S^q)$ (clearly no $\eqfun$-class includes both a
position of $\danger(S^q)$ and a position of $\nondanger(S^q)$). For $\Pi_j$,
    the new fact~$F_{\fnew}$ is defined following $\mybf{t}$; for the classes in
$\nondanger(S^q)$, it is defined
following an existing fact of~$I$. Hence, invariant \invar{fw} is preserved.

  \item Invariant \invar{help}
    is preserved because the only new elements of~$F_{\fnew}$ that may be in $\calH$ are 
those used at positions of $\Pi_j$, which is inner. 
\end{itemize}
Let $I_\f$ be the result of the process that we have described. It satisfies
$\ids$ by definition, and it is a
weakly-sound superinstance of~$I_0'$ that satisfies $\ids$, by invariants
\invar{wsnd}, \invar{sub}, and \invar{fun}. Further, it follows $\pire$ by
invariant \invar{fw}, and $\pire$ is finite. This implies that $I_\f$ is finite,
because we apply chase steps by $\idsr$, so each chase step makes an element
of~$\dom(\pire)$ occur at a new position, so we only applied finitely many chase
steps. 
This concludes the proof of the Reversible Relation-Thrifty Completion Proposition (Proposition~\ref{prp:rtcompr}), and
concludes the \secname.

\section{Ensuring \lowercase{$k$}-Universality}
\label{sec:ksound}
We build on the constructions of the previous \secname to replace weak-soundness
by $k$-soundness for acyclic queries in $\acq$, for some $k > 0$
fixed in this \secname. That is, we aim to prove:

\begin{maintheorem}
  \label{thm:ksound}
  Reversible $\uid$s and $\ufd$s have finite $k$-universal models
  for $\acq$s.
\end{maintheorem}

We first introduce the concept of \emph{aligned superinstances}, which
give us an invariant 
that ensures $k$-soundness. We then give the \emph{fact-saturation} process that
generalizes relation-saturation, and a related notion of
\emph{fact-thrifty chase step}.
We then define \emph{essentiality}, which must additionally be ensured for us
to be able to reuse the weakly-sound completions of the previous \secname. We conclude
by the construction of a generalized completion process that uses these chase
steps to repair $\uid$ violations in the instance while preserving
$k$-soundness.

In this \secname, we still make the reversibility assumption on $\idsr$ and $\ufds$.
However, we will also be considering a superset $\idsb$ of
$\idsr$, which we assume to be transitively closed, but which may not satisfy
the reversibility assumption.
To prove Theorem~\ref{thm:ksound}, it suffices to define $\idsb \defeq \idsr$, so
\emph{the distinction can be safely ignored on first reading}. The reason for
the distinction will become apparent in the next \secname.

\subsection{Aligned Superinstances}

In this \subsecname,
we only work with the superset $\idsb$, and we do not use the reversibility
assumption.
We ensure $k$-soundness relative to~$\idsb$ by maintaining a \defo{$k$-bounded simulation} from our
superinstance of $I_0$ to the chase $\chase{I_0}{\idsb}$.

\begin{definition}
  \label{def:bbsim}
  For $I$, $I'$ two instances, $a \in \dom(I)$, $b \in \dom(I')$, and $n \in
  \NN$,
  we write $(I, a) \bsim_n (I', b)$ if,
  for any fact $R(\mybf{a})$ of~$I$ with $a_p = a$ for some~$R^p \in \pos(R)$,
  there exists a fact $R(\mybf{b})$ of~$I'$ such that $b_p = b$,
  and $(I, a_q) \bsim_{n-1} (I', b_q)$ for all~$R^q \in \pos(R)$ (note that this
  is tautological for $R^q = R^p$).
  The base case $(I, a) \bsim_0 (I', b)$ always holds.

  An \deft{$\bm{n}$-bounded simulation} from $I$ to $I'$ is a mapping $\cov$ 
  such that for all $a \in \dom(I)$, we have $(I, a) \bsim_n (I', \cov(a))$.

  We write $a \bbsim_n b$ for $a, b \in \dom(I)$ if both $(I, a) \bsim_n (I, b)$ and
  $(I, b) \bsim_n (I, a)$; this is an equivalence relation on $\dom(I)$.
\end{definition}

\begin{example}
  \label{exa:bsim}
  We illustrate in Figure~\ref{fig:bsim} some examples of 2-bounded simulations from
  one instance to another, on a binary signature. For any element $a$ in a left
  instance $I$ and image $a'$ of~$a$ in the right instance $I'$ by the
  2-bounded simulation (represented by the dashed red arrows), we have $(I, a)
  \bsim_2 (I', a')$. This means that, for any
  element $b$ in $I$ which is adjacent to~$a$ by some relation $R$, there must be an element
  $b'$ in~$I'$ which is adjacent to $a'$ by $R$ and satisfies $(I, b) \bsim_1
  (I', b')$; however, note that $b'$ need not be the image of~$b$ by the bounded
  simulation.

  Figure~\ref{fig:bsim_homom} illustrates how a homomorphism is a special case
  of a $2$-bounded simulation (indeed, it is an $n$-bounded simulation for any
  $n \in \mathbb{N}$).

  Figure~\ref{fig:bsim_acq} illustrates how 
  a 2-bounded simulation from $I$ to
  $I'$ does not guarantee that any $\acq$ satisfied by $I$ is also true
  in~$I'$: for this example, consider the query $\exists x y z u v w ~ R(x, y) \wedge S(y, z)
  \wedge T(z, u) \wedge U(u, v) \wedge V(v, w)$. However, we will soon see that
  $n$-bounded simulations preserve $\acq$ of size $\leq n$
  (Lemma~\ref{lem:ksimacq}).

  Figure~\ref{fig:bsim_cq} shows that a 2-bounded simulation does not preserve
  $\cq$s that are not $\acq$s, as witnessed by $\exists x y z ~ R(x, y) \wedge
  S(y, z) \wedge T(z, x)$. More generally, $n$-bounded simulations for all~$n$
  do not generally preserve this $\cq$, or other such $\cq$s.
\end{example}

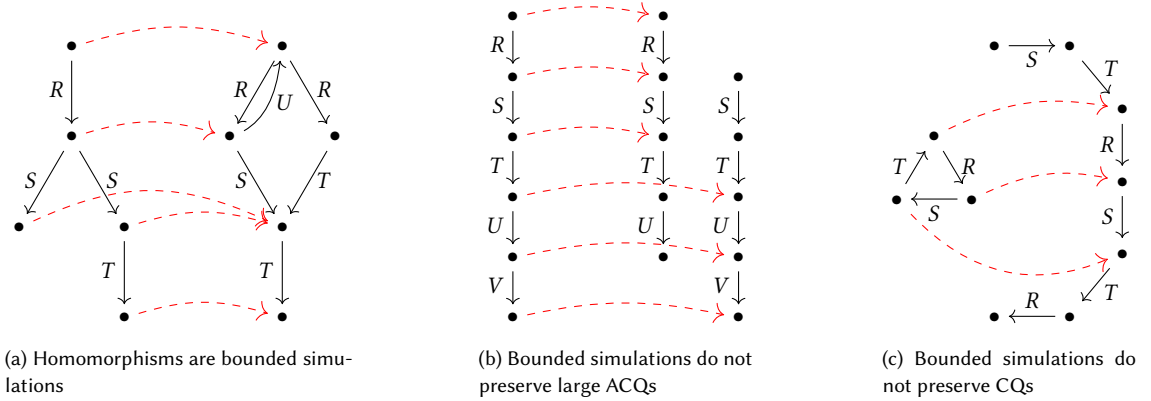
\begin{figure}
  {
  \settowidth\limage{\begin{tikzpicture}[
  text height=1ex,text depth=0ex,xscale=.7,yscale=1.2,
  mispos/.append style={red, font=\scriptsize},
  newfact/.append style={red, dashed, -{>[scale=1.5]}}
]
  \node (a) at (1, 0) {$\bullet$};
  \node (b) at (1, -1) {$\bullet$};
  \node (c2) at (0, -2) {$\bullet$};
  \node (c) at (2, -2) {$\bullet$};
  \node (d) at (2, -3) {$\bullet$};

  \draw[->] (a) -- (b) node[left,midway] {$R$};
  \draw[->] (b) -- (c2) node[left,midway] {$S$};
  \draw[->] (b) -- (c) node[right,midway] {$S$};
  \draw[->] (c) -- (d) node[left,midway] {$T$};

  \node(aa) at (5, 0) {$\bullet$};
  \node(bb) at (4, -1) {$\bullet$};
  \node(bb2) at (6, -1) {$\bullet$};
  \node(cc) at (5, -2) {$\bullet$};
  \node(dd) at (5, -3) {$\bullet$};
  
  \draw[->] (aa) -- (bb) node[left,midway] {$R$};
  \draw[->] (aa) -- (bb2) node[right,midway] {$R$};
  \draw[->] (bb2) -- (cc) node[right,midway] {$T$};
  \draw[->] (bb) -- (cc) node[left,midway] {$S$};
  \draw[->] (cc) -- (dd) node[left,midway] {$T$};
  \path (bb) edge[->,bend right] node[right,midway] {$U$} (aa);

  \path (a) edge[newfact,bend left=10] (aa);
  \path (b) edge[newfact,bend left=10] (bb);
  \path (c) edge[newfact,bend left=10] (cc);
  \path (c2) edge[newfact,bend left=15] (cc);
  \path (d) edge[newfact,bend left=10] (dd);
\end{tikzpicture}}
  \settowidth\rimage{\begin{tikzpicture}[
  text height=1ex,text depth=0ex,xscale=1,yscale=.8,
  mispos/.append style={red, font=\scriptsize},
  newfact/.append style={red, dashed, -{>[scale=1.5]}}
]
  \node (a) at (0, 0) {$\bullet$};
  \node (b) at (0, -1) {$\bullet$};
  \node (c) at (0, -2) {$\bullet$};
  \node (d) at (0, -3) {$\bullet$};
  \node (e) at (0, -4) {$\bullet$};
  \node (f) at (0, -5) {$\bullet$};

  \draw[->] (a) -- (b) node[left,midway] {$R$};
  \draw[->] (b) -- (c) node[left,midway] {$S$};
  \draw[->] (c) -- (d) node[left,midway] {$T$};
  \draw[->] (d) -- (e) node[left,midway] {$U$};
  \draw[->] (e) -- (f) node[left,midway] {$V$};

  \node (a2) at (2, 0) {$\bullet$};
  \node (b2) at (2, -1) {$\bullet$};
  \node (c2) at (2, -2) {$\bullet$};
  \node (d2) at (2, -3) {$\bullet$};
  \node (e2) at (2, -4) {$\bullet$};

  \draw[->] (a2) -- (b2) node[left,midway] {$R$};
  \draw[->] (b2) -- (c2) node[left,midway] {$S$};
  \draw[->] (c2) -- (d2) node[left,midway] {$T$};
  \draw[->] (d2) -- (e2) node[left,midway] {$U$};

  \node (b3) at (3, -1) {$\bullet$};
  \node (c3) at (3, -2) {$\bullet$};
  \node (d3) at (3, -3) {$\bullet$};
  \node (e3) at (3, -4) {$\bullet$};
  \node (f3) at (3, -5) {$\bullet$};

  \draw[->] (b3) -- (c3) node[left,midway] {$S$};
  \draw[->] (c3) -- (d3) node[left,midway] {$T$};
  \draw[->] (d3) -- (e3) node[left,midway] {$U$};
  \draw[->] (e3) -- (f3) node[left,midway] {$V$};

  \path (a) edge[newfact,bend left=15] (a2);
  \path (b) edge[newfact,bend left=15] (b2);
  \path (c) edge[newfact,bend left=15] (c2);
  
  \path (d) edge[newfact,bend left=15] (d3);
  \path (e) edge[newfact,bend left=15] (e3);
  \path (f) edge[newfact,bend left=15] (f3);
\end{tikzpicture}}
  \settowidth\rrimage{\begin{tikzpicture}[
  text height=1.3ex,text depth=0ex,xscale=1,yscale=1.2,
  mispos/.append style={red, font=\scriptsize},
  newfact/.append style={red, dashed, -{>[scale=1.5]}}
]
  \node (a) at (.5, -2.5) {$\bullet$};
  \node (b) at (1, -3.2) {$\bullet$};
  \node (c) at (0, -3.2) {$\bullet$};

  \draw[->] (a) -- (b) node[right,midway] {$R$};
  \draw[->] (b) -- (c) node[below,midway] {$S$};
  \draw[->] (c) -- (a) node[left,midway] {$T$};

  \node (y2) at (1.3, -1.5) {$\bullet$};
  \node (z2) at (2.3, -1.5) {$\bullet$};
  \node (a2) at (3, -2.2) {$\bullet$};
  \node (b2) at (3, -3) {$\bullet$};
  \node (c2) at (3, -3.8) {$\bullet$};
  \node (d2) at (2.3, -4.5) {$\bullet$};
  \node (e2) at (1.3, -4.5) {$\bullet$};

  \draw[->] (y2) -- (z2) node[below,midway] {$S$};
  \draw[->] (z2) -- (a2) node[right,midway,yshift=.8ex] {$T$};
  \draw[->] (a2) -- (b2) node[left,midway] {$R$};
  \draw[->] (b2) -- (c2) node[left,midway] {$S$};
  \draw[->] (c2) -- (d2) node[right,midway,yshift=-.8ex] {$T$};
  \draw[->] (d2) -- (e2) node[above,midway] {$R$};

  \path (a) edge[newfact,bend left=15] (a2);
  \path (b) edge[newfact,bend left=15] (b2);
  \path (c) edge[newfact,bend right=30] (c2);
\end{tikzpicture}}
  \noindent\begin{subfigure}[b]{\limage}
    \centering
    \begin{tikzpicture}[
  text height=1ex,text depth=0ex,xscale=.7,yscale=1.2,
  mispos/.append style={red, font=\scriptsize},
  newfact/.append style={red, dashed, -{>[scale=1.5]}}
]
  \node (a) at (1, 0) {$\bullet$};
  \node (b) at (1, -1) {$\bullet$};
  \node (c2) at (0, -2) {$\bullet$};
  \node (c) at (2, -2) {$\bullet$};
  \node (d) at (2, -3) {$\bullet$};

  \draw[->] (a) -- (b) node[left,midway] {$R$};
  \draw[->] (b) -- (c2) node[left,midway] {$S$};
  \draw[->] (b) -- (c) node[right,midway] {$S$};
  \draw[->] (c) -- (d) node[left,midway] {$T$};

  \node(aa) at (5, 0) {$\bullet$};
  \node(bb) at (4, -1) {$\bullet$};
  \node(bb2) at (6, -1) {$\bullet$};
  \node(cc) at (5, -2) {$\bullet$};
  \node(dd) at (5, -3) {$\bullet$};
  
  \draw[->] (aa) -- (bb) node[left,midway] {$R$};
  \draw[->] (aa) -- (bb2) node[right,midway] {$R$};
  \draw[->] (bb2) -- (cc) node[right,midway] {$T$};
  \draw[->] (bb) -- (cc) node[left,midway] {$S$};
  \draw[->] (cc) -- (dd) node[left,midway] {$T$};
  \path (bb) edge[->,bend right] node[right,midway] {$U$} (aa);

  \path (a) edge[newfact,bend left=10] (aa);
  \path (b) edge[newfact,bend left=10] (bb);
  \path (c) edge[newfact,bend left=10] (cc);
  \path (c2) edge[newfact,bend left=15] (cc);
  \path (d) edge[newfact,bend left=10] (dd);
\end{tikzpicture}
    \caption{Homomorphisms are bounded simulations}
    \label{fig:bsim_homom}
  \end{subfigure}\hfill\begin{subfigure}[b]{\rimage}
    \centering
    \begin{tikzpicture}[
  text height=1ex,text depth=0ex,xscale=1,yscale=.8,
  mispos/.append style={red, font=\scriptsize},
  newfact/.append style={red, dashed, -{>[scale=1.5]}}
]
  \node (a) at (0, 0) {$\bullet$};
  \node (b) at (0, -1) {$\bullet$};
  \node (c) at (0, -2) {$\bullet$};
  \node (d) at (0, -3) {$\bullet$};
  \node (e) at (0, -4) {$\bullet$};
  \node (f) at (0, -5) {$\bullet$};

  \draw[->] (a) -- (b) node[left,midway] {$R$};
  \draw[->] (b) -- (c) node[left,midway] {$S$};
  \draw[->] (c) -- (d) node[left,midway] {$T$};
  \draw[->] (d) -- (e) node[left,midway] {$U$};
  \draw[->] (e) -- (f) node[left,midway] {$V$};

  \node (a2) at (2, 0) {$\bullet$};
  \node (b2) at (2, -1) {$\bullet$};
  \node (c2) at (2, -2) {$\bullet$};
  \node (d2) at (2, -3) {$\bullet$};
  \node (e2) at (2, -4) {$\bullet$};

  \draw[->] (a2) -- (b2) node[left,midway] {$R$};
  \draw[->] (b2) -- (c2) node[left,midway] {$S$};
  \draw[->] (c2) -- (d2) node[left,midway] {$T$};
  \draw[->] (d2) -- (e2) node[left,midway] {$U$};

  \node (b3) at (3, -1) {$\bullet$};
  \node (c3) at (3, -2) {$\bullet$};
  \node (d3) at (3, -3) {$\bullet$};
  \node (e3) at (3, -4) {$\bullet$};
  \node (f3) at (3, -5) {$\bullet$};

  \draw[->] (b3) -- (c3) node[left,midway] {$S$};
  \draw[->] (c3) -- (d3) node[left,midway] {$T$};
  \draw[->] (d3) -- (e3) node[left,midway] {$U$};
  \draw[->] (e3) -- (f3) node[left,midway] {$V$};

  \path (a) edge[newfact,bend left=15] (a2);
  \path (b) edge[newfact,bend left=15] (b2);
  \path (c) edge[newfact,bend left=15] (c2);
  
  \path (d) edge[newfact,bend left=15] (d3);
  \path (e) edge[newfact,bend left=15] (e3);
  \path (f) edge[newfact,bend left=15] (f3);
\end{tikzpicture}
    \caption{Bounded simulations do not preserve large $\acq$s}
    \label{fig:bsim_acq}
  \end{subfigure}\hfill\begin{subfigure}[b]{\rrimage}
    \centering
    \begin{tikzpicture}[
  text height=1.3ex,text depth=0ex,xscale=1,yscale=1.2,
  mispos/.append style={red, font=\scriptsize},
  newfact/.append style={red, dashed, -{>[scale=1.5]}}
]
  \node (a) at (.5, -2.5) {$\bullet$};
  \node (b) at (1, -3.2) {$\bullet$};
  \node (c) at (0, -3.2) {$\bullet$};

  \draw[->] (a) -- (b) node[right,midway] {$R$};
  \draw[->] (b) -- (c) node[below,midway] {$S$};
  \draw[->] (c) -- (a) node[left,midway] {$T$};

  \node (y2) at (1.3, -1.5) {$\bullet$};
  \node (z2) at (2.3, -1.5) {$\bullet$};
  \node (a2) at (3, -2.2) {$\bullet$};
  \node (b2) at (3, -3) {$\bullet$};
  \node (c2) at (3, -3.8) {$\bullet$};
  \node (d2) at (2.3, -4.5) {$\bullet$};
  \node (e2) at (1.3, -4.5) {$\bullet$};

  \draw[->] (y2) -- (z2) node[below,midway] {$S$};
  \draw[->] (z2) -- (a2) node[right,midway,yshift=.8ex] {$T$};
  \draw[->] (a2) -- (b2) node[left,midway] {$R$};
  \draw[->] (b2) -- (c2) node[left,midway] {$S$};
  \draw[->] (c2) -- (d2) node[right,midway,yshift=-.8ex] {$T$};
  \draw[->] (d2) -- (e2) node[above,midway] {$R$};

  \path (a) edge[newfact,bend left=15] (a2);
  \path (b) edge[newfact,bend left=15] (b2);
  \path (c) edge[newfact,bend right=30] (c2);
\end{tikzpicture}
    \caption{Bounded simulations do not preserve $\cq$s}
    \label{fig:bsim_cq}
  \end{subfigure}
}
\caption{Examples of 2-bounded simulations (represented as dashed red lines): see Example~\ref{exa:bsim}}
\label{fig:bsim}
\end{figure}

The point of bounded simulations is that they preserve acyclic queries of size
smaller than the bound:

\begin{lemma}[$\acq$ preservation]
  \label{lem:ksimacq}
  For any instance $I$ and $\acq$ $q$ of size $\leq n$ such that $I \models
q$, if there is an $n$-bounded simulation from $I$ to $I'$, then $I' \models q$.
\end{lemma}

To show this lemma, we introduce a different way to write queries in $\acq$.
Consider the following alternate query language:

\begin{definition}
  \label{def:pointed}
  We inductively define a special kind of query with at most one free variable,
a \deft{pointed query}.
The base case is that of a tautological query with no atoms.
Inductively, pointed queries include all queries of the 
  form:
  \[
    q(x): \bigwedge_i \bigg(\exists \mybf{y}^{\mybf{i}} ~
      \Big(
        A^i(x, \mybf{y}^{\mybf{i}}) \wedge \bigwedge_{\is{y^i_j \in
  \mybf{y}^{\mybf{i}}}} q^i_j(y^i_j)\Big)\bigg)\]
  where the $\mybf{y}^{\mybf{i}}$ are vectors of pairwise distinct variables
  (also distinct from $x$),
  $A^i$ are atoms with free variables as indicated and with no repeated
  variables (each free variable occurs at exactly one position), and the $q^i_j$ are
  pointed queries.  

  The \deft{size} $\card{q}$ of a pointed query $q$ is the total number of atoms in
  $q$, including its subqueries.
\end{definition}

It is easily seen that, for any pointed query $q'$, the query $q: \exists x ~ q'(x)$ is an
$\acq$. Conversely, we can show:

\begin{lemma}
  \label{lem:acqstruct}
  For any (Boolean) $\acq$ $q$ and variable $x$ of~$q$, we can rewrite $q$ as
  $\exists x ~ q'(x)$ with $q'$ a pointed query such that $\card{q} = \card{q'}$.
\end{lemma}

\begin{proof}
  We show the claim by induction on the size of $q$. It is clearly true for the
  empty query.

  Otherwise, let $\calA = A_1, \ldots, A_m$ be the atoms of~$q$ where $x$ occurs. Because $q$ is an
  $\acq$, $x$ occurs exactly once in each of them, and each variable $y$
  occurring in one of the $A_i$ occurs exactly once in them overall: $y$ cannot occur
  twice in the same atom, nor can 
  occur in two different atoms $A_p$ and $A_q$ (as in this case $A_p$, $y$, $A_q$,
  $x$ would be a Berge cycle of~$q$). Let $\calY$ be the set
  of the variables occurring in~$\calA$, not including $x$.
  
  Consider the
  incidence multigraph $G$ of $q$ (Definition~\ref{def:acq}).
  Remember that we assume queries to be connected, so $q$ is connected, and $G$ is connected.
  Let $\calZ$ be the variables of $q$ which are not in 
  $\calY \cup \{x\}$. For each $z \in \calZ$, there must be a path $p_z$ from
  $x$ to $z$ in~$G$, written $x = w^z_1,
  \ldots, w^z_{n_z} = z$. Observe that, by definition of
  $\calY$, we must have $w^z_2 \in \calY$ for any such path. Further, for each
  $z \in \calZ$, we claim that there is a \emph{single} $y_z \in \calY$ such
  that $w^z_2 = y_z$ for any such path. Indeed, assuming to the contrary that
  there are $y_z \neq y_z'$ in $\calY$, a path $p_z$ whose second element is
  $y_z$, and a path $p_z'$ whose second element is $y_z'$, we deduce from $p_z$
  and $p_z'$ a Berge cycle in~$q$.
  
  Thus we can partition $\calZ$ into sets of variables $\calZ_y$ for $y \in
  \calY$, where $\calZ_y$ contains all variables $z$ of~$\calZ$ such that $y$ is the
  variable used to reach $z$ from~$x$. Let $\calA_y$ for $y \in \calY$ be the
  atoms of~$q$ whose variables are a subset of $\calZ_y \cup \{y\}$. It is clear that
  $\calA$ and the $\calA_y$ are a partition of the atoms of~$q$: no atom $A$ can
  include a variable $z$ from $\calZ_y$ and a variable $z'$ from $\calZ_{y'}$ for $y \neq
  y'$ in~$\calY$, as otherwise a path from $x$ to $z$ and a path from $x$ to
  $z'$, together with $A$, imply that $q$ has a Berge cycle.
  
  Now, we form for each $y \in \calY$ a query $q_y$ as the set of
  atoms $\calZ_y$, with all variables existentially quantified except for~$y$. As the queries
  $\exists y \, q_y(y)$ are connected queries in $\acq$ which are strictly
  smaller than $q$, by induction we can
  rewrite $q_y$ to a pointed query of the same size. Hence, we have shown
  that $q$ can be rewritten as a pointed query built from the $A_i$ and, for each $i$,
  the $q_y$ for $y \in \calY$.
\end{proof}

We use this normal form to prove 
the $\acq$ Preservation Lemma (Lemma~\ref{lem:ksimacq}):

\begin{proof}[\proofof Lemma~\ref{lem:ksimacq}]
  Fix the instances $I$ and $I'$, and the $\acq$ $q$.
  We show, by induction on $n \in \NN$, the following claim:
  for any $n \in
  \mathbb{N}$,
  for any \emph{pointed query} $q$ such that $\card{q} \leq n$,
  for any $a \in \dom(I)$,
  if $I \models q(a)$,
  then for any $a' \in \dom(I')$ such that $(I, a) \bsim_n (I', a')$,
  we have $I' \models q(a')$.
  Clearly this claim implies the statement of the Lemma, as by
  Lemma~\ref{lem:acqstruct} any $\acq$ query can be written as $\exists x ~
  q(x)$ with $q$ a pointed query.
  The case of the trivial query is immediate.
 
  For the induction step, consider a pointed query $q(x)$ of size $n \defeq \card{q}$, $n>0$,
  written in the form of Definition~\ref{def:pointed}, and fix $a \in \dom(I)$.
  Consider a match $h$ of $q(a)$ on $I$, which must map~$x$ to~$a$. 
  Let $a' \in \dom(I')$ be such that $(I, a) \bsim_{n} (I', a')$. We show that
  $I' \models q'(a)$.

  Using notation from Definition~\ref{def:pointed},
  write $\mybf{y}$ the (disjoint) union of the~$\mybf{y}^{\mybf{i}}$, write
  $\calA = A^1, \ldots, A^n$, and write $q^i_j(y^i_j)$ the subqueries.
  Let 
  $b^i_j \defeq h(y^i_j)$ for all $y^i_j \in \mybf{y}$.
  We show that there is a match $h_{\calA}$ of~$\calA$ on $I'$ that maps $x$ to
  $a'$ and such that every $y^i_j \in \mybf{y}$ is mapped to some element $(b^i_j)'$ of
  $I'$ such that $(I, b^i_j) \bsim_{n-1} (I', (b^i_j)')$. Indeed, start by
  fixing $h_{\calA}(x) \defeq
  a'$. Now, for each atom $A^i = R(x, \mybf{y}^{\mybf{i}})$ of $\calA$, the
  variable $x$ occurs at some
  position, say $R^p$, and $h(A^i) = R(\mybf{b}^{\mybf{i}})$ is a fact of~$I$ where $h(x) = a$ occurs at
  position $R^p$. As each variable in $\mybf{y}^{\mybf{i}}$ occurs at precisely
  one position of~$A^i$, we index each of these variables by the one position in~$A^i$ where it
  occurs.
  Now, as $(I, a) \bsim_{n} (I', a')$, there is a fact $(A^i)'
  = R((\mybf{b}^{\mybf{i}})')$ of~$I'$ such that $(b^i_p)' = a'$ and, for all $1 \leq j \leq
  \arity{R}$ with $p \neq j$, we have $(I, b^i_j) \bsim_{n-1} (I', (b^i_j)')$. We
  define $h_{\calA}(y^i_j) \defeq
  (b^i_j)'$ for all~$i$ and~$j$. As each variable of $\mybf{y}$ occurs exactly once in~$\calA$ overall,
  the definitions cannot conflict, so this correctly defines a function
  $h_{\calA}$
  which is clearly a match of $\calA$ on $I'$ with the claimed
  properties.
  
  Now, each of the $q^i_j$ is a pointed query which is strictly smaller than
  $q$. Further,
  the restriction of~$h$ to the variables of~$q^i_j$
  is a match of $q^i_j$ on~$I$ that maps each $y^i_j$ (indexing the variables
  of~$y^i$ in
  the same way as before) to $b^i_j \in \dom(I)$.
  As we have $(I, b^i_j) \bsim_{n-1} (I', (b^i_j)')$, then we can apply the induction
  hypothesis to show that each of the $q^i_j$ has a match $h_{i,j}$ in $I'$
  that maps $y^i_j$ to
  $(b^i_j)'$. As these queries have disjoint sets of variables, the range of the
  $h_{i,j}$ is disjoint, and the range of each $h_{i,j}$ overlaps with $h_{\calA}$ 
  only on $\{y^i_j\}$, where we have $h_{\calA}(y^i_j) = h_{i,j}(y^i_j) = (b^i_j)'$. Thus,
  we can combine the $h_{i,j}$
  and the previously defined $h_{\calA}$ to obtain an overall
  match of $q$ in $I'$ that matches $x$ to $a'$. This concludes the proof of the
  induction step, and proves our claim on pointed queries.
\end{proof}

This implies that any superinstance of~$I_0$ that has a $k$-bounded simulation to
$\chase{I_0}{\idsb}$ must be $k$-sound for~$\uconb$ (no matter whether it
satisfies $\uconb$ or not). Indeed, the chase is a universal model for $\idsb$, and it
satisfies $\ufds$ (by the Unique Witness Property, and because $I_0$ does).
Hence, the chase is in particular $k$-universal for~$\uconb$. Hence,
by the $\acq$ Preservation Lemma (Lemma~\ref{lem:ksimacq}),
any superinstance with a $k$-bounded simulation to the chase is $k$-sound.

We give a name to such superinstances.
For convenience, we also require them to be finite and satisfy $\ufds$.
For technical reasons we require that the simulation is the identity on $I_0$,
that it does not map other elements to $I_0$, and that elements occur in the
superinstance
at least at the position where their $\cov$-image was introduced in the chase
(the \defo{directionality condition}):

\begin{definition}
  An \deft{aligned superinstance} $J = (I, \cov)$ of~$I_0$
  (for $\ufds$ and $\idsb$)
  consists of a finite superinstance
  $I$ of~$I_0$
  that satisfies $\ufds$,
  and a $k$-bounded simulation
  $\cov$ from $I$ to $\chase{I_0}{\idsb}$ such that $\restr{\cov}{\dom(I_0)}$ is the
  identity and $\restr{\cov}{\dom(I\backslash I_0)}$ maps to $\chase{I_0}{\idsb}
  \backslash I_0$.
  
  Further, for any $a \in \dom(I) \backslash \dom(I_0)$,
  letting $R^p$ be the position where $\cov(a)$ was introduced in
  $\chase{I_0}{\idsb}$, we require that $a \in \pi_{R^p}(I)$. We call this the
  \deft{directionality condition}.

  We write $\dom(J)$ to mean $\dom(I)$, and extend other existing notation in the
  same manner when relevant, e.g., $\appelem{J}{\tau}$ means $\appelem{I}{\tau}$.
\end{definition}

\subsection{Fact-Saturation}

Before we perform the \emph{completion process} that allows us to satisfy
the $\uid$s $\idsr$, we need to perform a \emph{saturation process}.
Like aligned superinstances,
this process is defined with respect to the
superset $\idsb$, and does not depend on the reversibility assumption.
The process generalizes relation-saturation from the previous \secname:
instead of achieving all relations, we want the
aligned superinstance to achieve all \defo{fact classes}:

\begin{definition}
  \label{def:factclass}
  A \deft{fact class} is a pair $(R^p, \mybf{C})$ of a position $R^p
  \in \pos(\sigma)$
  and a $\arity{R}$-tuple of 
  $\bbsim_k$-classes of elements of
  $\chase{I_0}{\idsb}$, with $\bbsim_k$ as in Definition~\ref{def:bbsim}.

  The \deft{fact class} of a fact $F = R(\mybf{a})$ of~$\chase{I_0}{\idsb}
  \backslash I_0$ is $(R^p, \mybf{C})$, where $a_p$ is the exported element
  of~$F$ and $C_i$ is the $\bbsim_k$-class of~$a_i$ in $\chase{I_0}{\idsb}$ for
  all $R^i \in \pos(R)$.
  
  A fact class $(R^p, \mybf{C})$ is
  \deft{achieved} 
  in $\chase{I_0}{\idsb}$ if
  $\nondanger(R^p) \neq \emptyset$ and
  if it is the fact class of some fact of
  $\chase{I_0}{\idsb} \backslash I_0$.
Such a fact is an \emph{achiever} of the fact class.
  We write $\rfcl$ for the set of all achieved 
  fact classes.

  For brevity, the dependence on~$I_0$, $\idsb$, and $k$ is omitted from 
  this notation.
\end{definition}

The requirement that $\nondanger(R^p)$ is non-empty is a technicality that will
prove useful in \secref{sec:hfds}. The following is easy to see:

\begin{lemma}
  \label{lem:clasfino}
  For any initial instance $I_0$, set $\idsb$ of~$\uid$s, and $k \in \mathbb{N}$, $\rfcl$ is finite.
\end{lemma}

\begin{proof}
We first show that $\bbsim_k$ has only a finite number
of equivalence classes on $\chase{I_0}{\idsb}$. Indeed, for any element $a \in
\dom(\chase{I_0}{\idsb})$, by the Unique Witness Property,
the number of facts in which $a$ occurs is bounded by a constant depending only
  on $I_0$ and $\idsb$.  We use the standard notion of the \defo{Gaifman graph}
  of $\chase{I_0}{\idsb}$, which is the infinite undirected graph having vertices the elements of
  $\chase{I_0}{\idsb}$, with an edge between any pair of elements that
  occur in the same fact.
The elements of~$\dom(\chase{I_0}{\idsb})$ which are relevant to
determine the $\bbsim_k$-class of~$a$  are only those whose  distance to~$a$ in
the Gaifman graph is~$\leq k$.
From the bound above, we see that there is a constant $M$ depending only on $I_0$,
$\idsb$, and
$k$, bounding the number of such elements. Since there are only finitely many isomorphism
types of structures on $M$ elements, we get a finite bound on the number of equivalence classes.

This reasoning clearly implies that $\rfcl$ is finite, because the number of~$m$-tuples of
equivalence classes of~$\bbsim_k$ that occur in $\chase{I_0}{\idsb}$ is then finite
for any $m \leq \max_{R \in \sigma} \arity{R}$, and $\positions(\sigma)$ is finite.
\end{proof}

We define \defo{fact-saturated} superinstances, which achieve all fact
classes in~$\rfcl$:

\begin{definition}
  \label{def:factsat}
  An aligned superinstance $J = (I, \cov)$ of~$I_0$ is \deft{fact-saturated} if, for any achieved fact
  class $D = (R^p, \mybf{C})$ in $\rfcl$, there is a fact $F_D =
  R(\mybf{a})$ of~$I \backslash I_0$
  such that $\cov(a_i) \in C_i$ for all $R^i \in \pos(R)$.
  We say that $F_D$ \deft{achieves}~$D$ in~$J$.

  Note that this definition does not depend on the position~$R^p$ of the fact
  class.
\end{definition}

The point of fact-saturation is that, when we perform thrifty chase steps, we
can reuse elements from a suitable achiever at the non-dangerous positions. With
relation-saturation, the facts were of the right relation; with fact-saturation,
they further achieve the right 
\emph{fact class}, which will be important to maintain the bounded simulation
$\cov$.

The fact-saturation completion process, which replaces the relation-saturation
process of the previous \secname, works in the same way.

\begin{lemma}[Fact-saturated solutions]
  \label{lem:nondetexhaust}
  For any $\uid$s $\idsb$, $\ufd$s $\ufds$, and instance~$I_0$,
  there exists a finite number $n \in \NN$ such that
  the result $I$ of performing $n$ chase
  rounds on $I_0$ is such that $J_0 = (I, \id)$ is a fact-saturated aligned
  superinstance of~$I_0$.
\end{lemma}

\begin{proof}
For every $D \in \rfcl$, let $n_{D} \in \mathbb{N}$ be such
that $D$ is achieved
by a fact of~$\chase{I_0}{\idsb}$ created at round $n_{D}$.
As $\rfcl$ is
finite, $n \defeq \max_{D \in \rfcl} n_D$ is finite. Hence, all classes
of $\rfcl$ are achieved after $n$ chase rounds on $I_0$.

Consider now $I_0'$ obtained from the aligned superinstance $I_0$ by
$n$ rounds
of the $\uid$ chase, and $J_0 = (I_0', \id)$.
It is clear that for any $D \in \rfcl$,
considering an achiever of $D$ in $\chase{I_0}{\idsb}$, the corresponding fact
in~$J_0$ is an achiever of~$D$ in~$J_0$. Hence, $J_0$ is indeed fact-saturated.
\end{proof}

We thus obtain a fact-saturated aligned superinstance $J_0$ of our initial
instance~$I_0$, which we now want to complete to one that satisfies the $\uid$s
we are interested in, namely $\idsr$.

\subsection{Fact-Thrifty Steps}

In the previous \secname, we defined
relation-thrifty chase steps, which reused non-dangerous elements from any
fact of the correct relation, assuming relation-saturation. We now define \defo{fact-thrifty} steps, which are thrifty
steps that reuse elements from a fact achieving the right fact class, thanks to
fact-saturation.
To do so, however, we must first refine the notion
of \defo{thrifty} chase step, to make them
apply to aligned superinstances.
We will \emph{always} apply them to aligned superinstances
for~$\idsb$ and~$\ufds$; however, we will always chase by the $\uid$s of~$\idsr$.

\begin{definition}[Applying thrifty chase steps to aligned superinstances]
  \label{def:ft2}
  Let $J = (I, \cov)$ be an aligned superinstance of~$I_0$ for $\idsb$ and
  $\ufds$, let $\tau:\ui{R^p}{S^q}$ be a $\uid$ of~$\idsr$,
  and let $a \in \appelem{J}{\tau}$.
  The result of applying a \deft{thrifty} chase step to~$a$ in~$J$ by~$\tau$ is
  a pair $(I', \cov')$ where:
  \begin{itemize}
    \item The instance $I'$ is the result of applying some thrifty step to~$a$
      in~$I$ by~$\tau$, as in Definition~\ref{def:thrifty} (note that this
      \emph{only} depends on~$\idsr$ and~$\ufds$, \emph{not} on~$\idsb$).
    \item The mapping $\cov'$ extends $\cov$ to elements of $\dom(I') \backslash
  \dom(I)$ as follows.
  Because $\cov$ is a $k$-bounded simulation and $k > 0$, it is in particular a
  $1$-bounded simulation, so we have $\cov(a) \in
  \pi_{R^p}(\chase{I_0}{\idsb})$. Hence, because $\tau \in \idsr \subseteq \idsb$, there
  is a fact $F_\w = S(\mybf{b}')$ in $\chase{I_0}{\idsb}$
  with $b'_q = \cov(a)$. We call $F_\w$ the \deft{chase witness}. 
      For any $b \in \dom(I') \backslash \dom(I)$, let $S(\vec b)$ be the
      unique fact of~$I' \backslash I$ and let $b_r$ be the only element of that
      fact such that $b = b_r$;
we define
  $\cov'(b_r) \defeq b'_r$.\qedef
  \end{itemize}
\end{definition}

We do not know yet whether the result $(I', \cov')$ of a thrifty chase step on
an aligned superinstance $(I, \cov)$ is still an aligned superinstance; we will
investigate this later.

Now that we have defined thrifty chase steps on aligned superinstances, we can
clarify the role of the directionality condition. Its goal is to ensure, intuitively, that
as chase steps go ``downwards'' in the original chase,
thrifty chase steps on aligned superinstances makes the $\cov$ mapping
go ``downwards'' in the chase as well. Formally:

\begin{lemma}[Directionality]
  \label{lem:wadef}
  Let $J$ be an aligned superinstance of~$I_0$ for $\idsb$ and $\ufds$, and consider
  the application of a thrifty chase step for a $\uid$ $\tau: \ui{R^p}{S^q}$. 
  Consider the chase witness $F_\w = S(\mybf{b}')$. Then $b'_q$ is the
  exported element of~$F_\w$.
\end{lemma}

\begin{proof}
  Let $F_{\active} = R(\mybf{a})$ be the active fact in $J$,
  let $F_{\fnew} = S(\mybf{b})$ be the new fact of~$J'$,
  and let $\tau: \ui{R^p}{S^q}$ be the $\uid$,
  so $a_p = b_q$ is the exported element of this chase step.
  Let $F_\w = S(\mybf{b}')$ be the chase witness in $\chase{I_0}{\idsb}$.
  Assume by way of contradiction 
  that $b'_q$ was not the exported element in~$F_\w$,
  so that it was introduced in $F_\w$.
  In this case, as $\cov(a_p) = \cov(b_q) = b'_q$,
  by the directionality condition in the definition of aligned superinstances,
  we have $a_p \in \pi_{S^q}(J)$,
  which contradicts the fact that $a_p \in \appelem{J}{\tau}$.
  Hence, we have proved by contradiction
  that $b'_q$ was the exported element in~$F_\w$.
\end{proof}

This observation will be important to connect fact-saturation to the
\defo{fact-thrifty chase steps} that we now define:

\begin{definition}
  \label{def:factthrifty}
  We define a \deft{fact-thrifty chase step},
  using the notation of Definition~\ref{def:thrifty}, as follows:
  if $\nondanger(S^q)$ is non-empty, choose one fact
  $F_\rr = S(\mybf{c})$ of $I \backslash I_0$ that achieves the fact class of
  $F_\w = S(\mybf{b}')$ (that is, $\cov(c_i) \bbsim_k b'_i$ for all
  $i$), and use $F_\rr$ to define $b_r \defeq c_r$ for all $S^r \in \nondanger(S^q)$.

  We also call a 
  fact-thrifty chase step \deft{fresh} if for all $S^r \in \danger(S^q)$,
  we take $b_r$ to be a fresh element only occurring at that position (and
  extend $\cov'$ accordingly).
\end{definition}

We first show that, on \emph{fact-saturated} instances, any $\uid$ violation can
be repaired by a fact-thrifty chase step; this uses Lemma~\ref{lem:wadef}. More
specifically, we show that, for
any relation-thrifty chase step that we could want to apply, we could apply a
fact-thrifty chase step instead.

\begin{lemma}[Fact-thrifty applicability]
  \label{lem:ftappl}
  For any fact-saturated superinstance $J$ of an instance $I_0$,
  for any $\uid$ $\tau: \ui{R^p}{S^q}$ of~$\idsr$,
  for any element $a \in \appelem{J}{\tau}$, we can apply a fact-thrifty
  chase step on $a$ with $\tau$ to satisfy this violation. Further, for any new fact $S(\mybf{e})$
  that we can create by chasing on~$a$ with $\tau$ with a relation-thrifty chase
  step, we can instead apply a fact-thrifty chase step on $a$ with $\tau$ to create a fact
  $S(\mybf{b})$ with $b_r = e_r$ for all $S^r \in \pos(S) \backslash
  \nondanger(S^r)$.
\end{lemma}

\begin{proof}
  We prove the first part of the statement by justifying the existence of the
  fact~$F_\rr$, which only needs to be done if $\nondanger(S^q)$ is non-empty. In
  this case, considering the fact $F_\w = S(\mybf{b}')$ in $\chase{I_0}{\idsb}$, we know by
  Lemma~\ref{lem:wadef} that $b'_q$ is the exported element in $F_\w$. Hence,
  letting $D$ be the fact class of $F_\w$, we have $D = (S^q, \mybf{C})$ for
  some $\mybf{C}$, and $D$ is in $\rfcl$ because
  $\nondanger(S^q)$ is non-empty. Hence, by definition of fact-saturation, there
  is a fact $F_\rr = S(\mybf{c})$ in $J$ such that, for all $S^r \in \pos(R)$, we
  have $\cov(c_i) \in C_i$, i.e., $\cov(c_i) \bbsim_k b'_i$ in
  $\chase{I_0}{\idsb}$. This proves the first part of the claim.

  For the second part of the claim, observe that the definition of fact-thrifty
  chase steps only imposes conditions on the non-dangerous positions, so
  considering any new fact $S(\mybf{e})$ created by a relation-thrifty chase
  step, changing its non-dangerous positions to follow the definition of
  fact-thrifty chase steps, we can create it with a fact-thrifty chase step.
\end{proof}

We now look at which properties are preserved on the result $(I', \cov')$ of fact-thrifty chase steps.
First note that fact-thrifty chase steps are in particular relation-thrifty, so
$I'$ is still weakly-sound and still satisfies $\ufds$
(by Lemmas~\ref{lem:thriftyp} and~\ref{lem:rthriftyp}). However, we do not
know yet whether $(I', \cov')$ is an aligned superinstance for $\ufds$ and
$\idsb$.

For now, we show that it is the case for \emph{fresh} fact-thrifty chase steps:

\begin{lemma}[Fresh fact-thrifty preservation]
  \label{lem:fftp}
  For any fact-saturated aligned superinstance $J$ of~$I_0$ (for $\ufds$ and
  $\idsb$), the result $J'$ of
  a fresh fact-thrifty chase step on $J$ is still a fact-saturated aligned
  superinstance of~$I_0$.
\end{lemma}

We prove this result in the rest of the \subsecname. For \emph{non-fresh} fact-thrifty
chase steps, the analogous claim is not true in general: it requires us to
introduce \emph{essentiality}, the focus of the next \subsecname, and relies on
the reversibility assumption that we made on $\idsr$ and $\ufds$.

To prove the Fresh Fact-Thrifty Preservation Lemma, we first make a general
claim about how we can extend a superinstance by adding a fact, and
preserve bounded simulations.

\begin{lemma}
  \label{lem:addfact}
  Let $n \in \mathbb{N}$.
  Let $I_1$ and $I$ be instances and $\cov$ be a $n$-bounded simulation from
  $I_1$ to $I$. Let $I_2$ be a superinstance of~$I_1$ 
  defined by adding one fact $F_{\fnew} = R(\mybf{a})$ to $I_1$,
  and let $\cov'$ be a mapping from $\dom(I_2)$
  to $\dom(I)$ such that $\restr{\cov'}{\dom(I_1)} = \cov$. Assume there is a fact $F_\w =
  R(\mybf{b})$ in $I$ such that, for all $R^i \in \pos(R)$, $\cov'(a_i) \bbsim_n b_i$. Then
  $\cov'$ is an $n$-bounded simulation from $I_2$ to $I$.
\end{lemma}

\begin{proof}
  We prove the claim by induction on $n$. The base case of $n = 0$ is immediate.

  Let $n > 0$, assume that the claim holds for $n-1$, and show that it holds for
  $n$. As $\cov$ is an $n$-bounded simulation, it is an $(n-1)$-bounded
  simulation, so we know by the induction hypothesis that $\cov'$ is an
  $(n-1)$-bounded simulation.

  Let us now show that it is an $n$-bounded simulation.
  Let $a \in \dom(I_2)$ be an element and
  show that $(I_2, a) \bsim_{n} (I, \cov'(a))$. 
  Hence, for any $F = S(\mybf{a})$ a fact of~$I_2$ with $a_p = a$ for some $p$, 
  we must show that there exists a fact $F' =
  S(\mybf{a}')$ of~$I$ with $a'_p = \cov'(a_p)$ and 
  $(I_2, a_q) \bsim_{n-1} (I, a'_q)$ for all $S^q \in \pos(S)$.
  
  The first possibility is that $F$ is the new fact $F_{\fnew} = R(\mybf{a})$.
  In this case, as we have $(I, b_p) \bsim_n (I, \cov'(a_p))$,
  considering $F_\w$, we deduce the existence of a fact
  $F_\w' = R(\mybf{c})$ in~$I$ such that $c_p = \cov'(a_p)$
  and $(I, b_q) \bsim_{n-1} (I, c_q)$ for all $1 \leq q \leq \arity{R}$.
  We take $F' = F_\w'$ as our witness fact for~$F$.
  By construction we have $c_p = \cov'(a_p)$.
  Fixing $1 \leq q \leq \arity{R}$,
  to show that $(I_2, a_q) \bsim_{n-1} (I, c_q)$,
  we use the fact that $\cov'$ is an $(n-1)$-bounded simulation
  to deduce that $(I_2, a_q) \bsim_{n-1} (I, \cov'(a_q))$.
  Now, we have $(I, \cov'(a_q)) \bsim_{n-1} (I, b_q)$,
  and as we explained we have $(I, b_q) \bsim_{n-1} (I, c_q)$,
  so we conclude by transitivity.
  
  If $F$ is another fact, then it is a fact of~$I_1$,
  so its elements are in $\dom(I_1)$,
  and as $\cov'$ coincides with $\cov$ on such elements,
  we conclude because $\cov$ is a $n$-bounded simulation.
\end{proof}

We now prove 
the Fresh Fact-Thrifty Preservation Lemma (Lemma~\ref{lem:fftp}), which concludes the
\subsecname:

\begin{proof}[\proofof Lemma~\ref{lem:fftp}]

It is immediate that, letting $J' = (I', \cov')$ be the result of the
fact-thrifty chase step,
$I'$ is still a finite superinstance of~$I_0$, and it still satisfies $\ufds$,
because fact-thrifty chase steps are relation-thrifty
chase steps, so we can still apply
Lemma~\ref{lem:rthriftyp}.

To show that $\cov'$ is still a $k$-bounded simulation, we apply Lemma~\ref{lem:addfact}
  with $F_{\fnew} = S(\mybf{b})$ and $F_\w = S(\mybf{b}')$. Indeed, letting $\tau : \ui{R^p}{S^q}$
be the applied $\uid$ in~$\idsr$, we have $\cov'(b_q) = b'_q$ by definition, and have set
$\cov'(b_r) \defeq b'_r$ for all $S^r \in \danger(S^q)$ (note that each such
$b_r$ occurs at only one position). For $S^r \in \nondanger(S^q)$, we have
$\cov'(b_r) \bbsim_k b'_r$ in $\chase{I_0}{\idsb}$ by definition of a
fact-thrifty chase step. Hence, by Lemma~\ref{lem:addfact}, $\cov'$ is still a
$k$-bounded simulation from $I'$ to $\chase{I_0}{\idsb}$.

We now check the directionality condition on elements of 
$\dom(I') \backslash \dom(I)$, namely, we show:
for $S^r \neq S^q$, if $b_r \in \dom(I') \backslash \dom(I)$,
then $b_r$ occurs in $J'$ at the position where
$\cov'(b_r)$ was introduced in $\chase{I_0}{\idsb}$.
By the Directionality Lemma (Lemma~\ref{lem:wadef}),
we know that $b'_q$ was the exported element of $F_\w$. Hence, as $\cov'(b_r)
\defeq b'_r$, we know that $b'_r$ was introduced at position $S^r$
in~$F_\w$ in~$\chase{I_0}{\idsb}$, so the condition is respected.

Last, the preservation of fact-saturation is immediate, and
the fact that $\cov'$ is the identity on~$I_0$
is immediate because $\restr{\cov'}{\dom(I_0)} = \restr{\cov}{\dom(I_0)}$.
We show that $\restr{\cov'}{\dom(I' \backslash I_0)}$ maps
to $\chase{I_0}{\idsb} \backslash I_0$, using the directionality condition.
Indeed, for all elements $b_r \in \dom(I') \backslash \dom(I)$
(with $S^r \neq S^q$),
which are clearly not in~$I_0$,
we have fixed $\cov'(b_r) \defeq b_r'$,
and as we explained $b'_r$ is introduced in $F_\w$ in $\chase{I_0}{\idsb}$
so it cannot be an element of~$I_0$; hence $b'_r$ is indeed
an element of $\chase{I_0}{\idsb} \backslash I_0$. This is the last point we had
to verify.
\end{proof}

\subsection{Essentiality}

The problem of non-fresh fact-thrifty chase steps is that, while they try to preserve
$k$-soundness on the non-dangerous positions, they may not preserve it overall:

\begin{figure}
  \noindent\null\begin{minipage}[b]{.25\linewidth}
    \centering
%
%
%
\begin{tikzpicture}[
  text height=1.3ex,text depth=0ex,xscale=1.5,yscale=.85,
  mispos/.append style={red, font=\scriptsize},
  newfact/.append style={red, dashed, -{>[scale=1.5]}}
]
  \node (bound1) at (-.25, -.5) {};
  \node (bound2) at (1.25, -.5) {};
  \node (y) at (0, 0) {$u$};
  \node (z) at (1, 0) {$v$};
  \node (a) at (0, -1) {$a$};
  \node (b) at (1, -1) {$b$};

  \node[red] (a2) at (0, -2) {$a_1$};
  \node[red] (a3) at (0, -3) {$a_2$};

  \node[red] (b2) at (1, -2) {$b_1$};
  \node[red] (b3) at (1, -3) {$b_2$};

  \draw[->] (a) -- (y) node[right,midway] {$U$};
  \draw[->] (z) -- (b) node[right,midway] {$V$};
  \draw[->] (a) -- (b) node[above,midway] {$R$};

  \path[newfact,bend left=25] (b) edge node[right, midway] {$R$} (b2)
    (b2) edge node[right, midway] {$R$} (b3)
    (b3) edge node[below, midway] {$R$} (a3)
    (a3) edge node[left, midway] {$R$} (a2)
    (a2) edge node[left, midway] {$R$} (a);
\end{tikzpicture}
    \caption{Soundness and fact-thrifty chase steps (see
    Example~\ref{exa:essentiality})}
    \label{fig:essentiality}
  \end{minipage}\hspace{.02\linewidth}\begin{minipage}[b]{.73\linewidth}
    \centering
    \null\hfill\begin{tikzpicture}[
  text height=1.3ex,text depth=0ex,xscale=.7,yscale=.9,
  mispos/.append style={red, font=\scriptsize},
  newfact/.append style={red, dashed, -{>[scale=1.5]}},
  sccp/.append style={pattern=north west lines, pattern
  color=darkgreen,opacity=.2},
  scc/.append style={darkgreen}
]
  \node (a1) at (-1, 1) {$v$};
  \node (a2) at (0, 1) {$x$};
  \draw (-1.2, 1.2) rectangle (.2, .8);
  \node (A) at (-1.5,1) {$A$};

  \node (s1) at (0, 0) {$x$};
  \node (s2) at (1, 0) {$y$};
  \draw (-.2, -.2) rectangle (1.2, .2);
  \node (R) at (-.5,0) {$S$};

  \path (a2) edge[bend left=45,->] (s1);

  \node (r1) at (1, -1) {$y$};
  \node (r2) at (2, -1) {$z$};
  \node (r3) at (3, -1) {$a$};
  \draw (.8, -1.2) rectangle (3.2, -.8);
  \node (R) at (.5,-1) {$R$};
  
  \path (s2) edge[very thick,bend left=45,->] (r1);

  \node (u1) at (1, -2) {$b$};
  \node (u2) at (2, -2) {$c$};
  \node (u3) at (3, -2) {$a$};
  \draw (.8, -2.2) rectangle (3.2, -1.8);
  \node (U) at (.5,-2) {$U$};

  \path (r3) edge[very thick,bend left=45,->] (u3);
  
  \node (t1) at (4.5, -2) {$a$};
  \node (t2) at (5.5, -2) {$d$};
  \draw (4.3, -2.2) rectangle (5.7, -1.8);
  \node (T) at (4,-2) {$T$};

  \path (r3) edge[bend left=15,->] (t1);
  
  \node (v1) at (1, -3) {$b$};
  \node (v2) at (2, -3) {$e$};
  \draw (.8, -3.2) rectangle (2.2, -2.8);
  \node (V) at (.5,-3) {$V$};

  \path (u1) edge[very thick,bend left=45,->] (v1);
\end{tikzpicture}\hfill
\begin{tikzpicture}[
  text height=1.3ex,text depth=0ex,xscale=.7,yscale=.9,
  mispos/.append style={red, font=\scriptsize},
  newfact/.append style={red, dashed, -{>[scale=1.5]}},
  sccp/.append style={pattern=north west lines, pattern
  color=darkgreen,opacity=.2},
  scc/.append style={darkgreen}
]
  \node (b1) at (0, 1) {$w'$};
  \node (b2) at (1, 1) {$e'$};
  \draw (-.3, 1.2) rectangle (1.2, .8);
  \node (B) at (-.6,1) {$B$};

  \node (v1) at (0, 0) {$b'$};
  \node (v2) at (1, 0) {$e'$};
  \draw (-.3, -.2) rectangle (1.2, .2);
  \node (V) at (-.6,0) {$V$};

  \path (b2) edge[bend left=45,->] (v2);

  \node (u1) at (0, -1) {$b'$};
  \node (u2) at (1, -1) {$c'$};
  \node (u3) at (2, -1) {$a'$};
  \draw (-.3, -1.2) rectangle (2.2, -0.8);
  \node (U) at (-.6,-1) {$U$};

  \path (v1) edge[very thick,bend left=45,->] (u1);

  \node (r1) at (0, -2) {$y'$};
  \node (r2) at (1, -2) {$z'$};
  \node (r3) at (2, -2) {$a'$};
  \draw (-.3, -2.2) rectangle (2.2, -1.8);
  \node (R) at (-.6,-2) {$R$};

  \path (u3) edge[very thick,bend left=45,->] (r3);
  
  \node (t1) at (3.5, -2) {$a'$};
  \node (t2) at (4.5, -2) {$d'$};
  \draw (3.2, -2.2) rectangle (4.7, -1.8);
  \node (T) at (2.9,-2) {$T$};

  \path (u3) edge[bend left=15,->] (t1);

  \node (s1) at (-1, -3) {$x'$};
  \node (s2) at (0, -3) {$y'$};
  \draw (-1.3, -3.2) rectangle (.2, -2.8);
  \node (R) at (-1.6,-3) {$S$};

  \path (r1) edge[very thick,bend left=45,->] (s2);
\end{tikzpicture}\hfill\null
    \caption{$\uid$ Chase Similarity Theorem (Theorem~\ref{thm:similarityb}); see Example~\ref{exa:similarity}.}
    \label{fig:chase}
  \end{minipage}
\end{figure}

\begin{example}
  \label{exa:essentiality}
  Consider the instance $I_0 = \{U(a, u), R(a, b), V(v, b)\}$ depicted as the
  solid black elements and edges in
  Figure~\ref{fig:essentiality}. Consider the $\uid$
  $\tau: \ui{R^1}{R^2}$, and the $\ufd$ $\phi: R^1 \rightarrow R^2$.
  We define $\idsr = \idsb = \{\tau, \tau^{-1}\}$ and $\ufds = \{\phi,
  \phi^{-1}\}$, so that $\ufds$ and $\idsr$ are reversible. We have $a \in
  \appelem{I}{\tau}$ and $b \in \appelem{I}{\tau^{-1}}$.
  To satisfy these violations, we can apply a fact-thrifty chase step by
  $\tau$
  on $a$ and create $F = R(b, a)$, noting that there are no non-dangerous positions.
  However, the superinstance $I_0 \sqcup \{F\}$
  is not a $k$-sound superinstance of~$I_0$ for $k \geq 3$. For instance, it makes the
  following $\acq$ true, which is not true in $\chase{I_0}{\idsb}$:
  $\exists x y z w ~ V(x, y), R(y, z), U(z, w)$.

  Instead, for any value of $k$, this problem can be avoided as follows.
  First, apply $k$ fresh fact-thrifty chase steps by~$\tau^{-1}$ to create the chain $R(b,
  b_1), \allowbreak R(b_1, b_2), \ldots, \allowbreak R(b_{k-1}, b_k)$. Then apply $k$ fresh fact-thrifty
chase steps by~$\tau$ to create $R(a_1, a), \allowbreak R(a_2, a_1),
\ldots, \allowbreak R(a_k,
  a_{k-1})$.
  Now we can apply a non-fresh fact-thrifty chase step by~$\tau$ on $a_k$ and
  create $R(b_k, a_k)$, and this does not make any new $\acq$ of size $\leq k$
  true. This process is illustrated with red elements and red dashed
  edges in Figure~\ref{fig:essentiality} for $k=2$.
\end{example}

The intuition behind this example is that non-fresh fact-thrifty chase steps may connect together elements at the
dangerous positions, but their image by~$\cov$ may be too dissimilar, so the
bounded simulation does not extend.
This implies that, in general, the result of a fact-thrifty chase step may not
be an aligned superinstance.
As the example shows, however, we can avoid that
problem if we chase for sufficiently long, so that the ``history'' of elements
no longer contains anything specific about them.

We first formalize this notion for elements of the chase $\chase{I_0}{\idsb}$, which we call
\defo{essentiality}. We will then define it for aligned superinstances using
the $\cov$ mapping.

\begin{definition}
  \label{def:chaserev}
  We define a forest structure on the facts of
$\chase{I_0}{\idsb}$: the facts of~$I_0$ are the roots, and the parent of a fact
$F$ not in $I_0$
is the fact~$F'$ that was the active fact for which $F$ was created,
so that $F'$ and~$F$ share the exported element of~$F$.

For $a \in \dom(\chase{I_0}{\idsb})$, if $a$ was introduced at position $S^r$ of an
$S$-fact $F = S(\mybf{a})$ created by applying the $\uid$ $\tau: \ui{R^p}{S^q}$
(with $S^q \neq S^r$) to its parent fact~$F'$,
we call $\tau$ the \deft{last $\uid$} of~$a$. The \deft{last two
$\uid$s} of~$a$ are $(\tau, \tau')$ where $\tau'$ is the last $\uid$ of the
exported element~$a_q$ of~$F$ (which was introduced in~$F'$).
For $n \in \mathbb{N}$,
we define the \deft{last $n$ $\uid$s} in the same way, for elements of
$\chase{I_0}{\idsb}$ introduced after sufficiently many rounds.

We say that $a$ is
  \deft{$\bm{n}$-essential} if its last $n$ $\uid$s are reversible in~$\idsb$. This is in
particular the case if these last $\uid$s are in $\idsr$: indeed, $\idsr$
satisfies the reversibility assumption so for any $\tau \in \idsr$, we have
$\tau^{-1} \in \idsr$, so that $\tau^{-1} \in \idsb$.
\end{definition}

The point of this definition is the following result, which we state without
proof for now. We will prove it in Section~\ref{sec:similarity}: 

\mydefthm{similarityb}{theorem}{$\uid$ chase similarity theorem}{
  For any instance~$I_0$, transitively closed set of $\uid$s $\idsb$,
  and $n \in \mathbb{N}$,
  for any two
  elements $a$ and $b$
  respectively introduced at positions~$R^p$ and~$S^q$ in $\chase{I_0}{\idsb}$,
  if $a$ and $b$ are $n$-essential,
  and if $\ui{R^p}{S^q}$ and $\ui{S^q}{R^p}$ are in $\idsb$,
  then $a \bbsim_n b$.
}

In other words, in the chase, when the last $n$ $\uid$s of an element were
reversible, then the
$\bbsim_n$-class of that element only depends on the 
position where it was introduced.

We use this to define a corresponding notion on aligned
superinstances: an aligned superinstance is \defo{$n$-essential} if, for all elements
that witness a violation of the $\uid$s~$\idsr$ that we wish to solve, their
$\cov$ image is an $n$-essential element of the chase, introduced at a suitable
position. In fact, we introduce a more general definition, which does not
require the superinstance to be aligned, i.e., does not require that~$\cov$ is a $k$-bounded
simulation.

\begin{definition}
  \label{def:nreverse}
  Let $J = (I, \cov)$ be a pair of a superinstance $I$ of~$I_0$
  and a mapping $\cov$ from $I$ to $\chase{I_0}{\idsb}$. Let $k \in \NN$.
  We call $a \in \dom(I)$ \mbox{\deft{$\bm{n}$-essential}} in~$J$ for $\idsr$ if the following are true:
  \begin{itemize}
    \item $\cov(a)$ is
  an $n$-essential element of~$\chase{I_0}{\idsb}$;
\item for any
  position $S^q \in \pos(\sigma)$
  such that $a \in \appelemb{\idsr}{I}{S^q}$, the position $T^v$ where
  $\cov(a)$ was introduced in $\chase{I_0}{\idsb}$ is such that $\ui{T^v}{S^q}$
  and $\ui{S^q}{T^v}$ are in~$\idsr$, which we write $T^v \eqids S^q$ as in the
  previous \secnames.
  \end{itemize}
  Note that the second point is vacuous if there is no $\uid$ of~$\idsr$ applicable to~$a$.
  We call $J$ \deft{$\bm{n}$-essential} for~$\idsr$ if, for all $a \in \dom(J)$, $a$ is $n$-essential in~$J$
  for~$\idsr$.
\end{definition}

We now show that, as we assumed the $\uid$s of $\idsr$ to be reversible
(by the reversibility assumption), fresh fact-thrifty chase steps by~$\idsr$ never make
essentiality decrease, and even make it \emph{increase} on new elements:

\begin{lemma}[Thrifty steps and essentiality]
  \label{lem:ftok2}
  For any $n$-essential aligned superinstance $J$, letting $J' = (I',
  \cov')$ be the result of a thrifty chase step on $J$ by a $\uid$ of
  $\idsr$, $J'$ is still $n$-essential. Further,
  all elements of $\dom(J') \backslash \dom(J)$ are $(n+1)$-essential in~$J'$.
\end{lemma}

\begin{proof}
  Fix $J$ and $J'$; note that $J'$ may not be an aligned superinstance.
  Consider first the elements of~$\dom(J)$ in~$J'$. For any $a \in \dom(J)$, by definition of thrifty chase
  steps, we know that for any $T^v \in \pos(\sigma)$ such that $a \in
  \pi_{T^v}(J')$, we have either $a \in \pi_{T^v}(J)$ or $a \in
  \appelemb{\idsr}{J}{T^v}$. Hence, as $\idsr$ is transitively closed, for any
  $U^w \in \pos(\sigma)$ such that $a \in \appelemb{\idsr}{J'}{U^w}$, we have
  also $a \in \appelemb{\idsr}{J}{U^w}$, and as $J$ is $n$-essential, we
  conclude that $a$ is $n$-essential in~$J'$.
  Hence, it suffices to show that any element in
  $\dom(J') \backslash \dom(J)$ is $(n+1)$-essential in~$J'$.

  To do this, write $\tau : \ui{R^p}{S^q}$ the $\uid$ applied in the chase step,
  and let $F_{\fnew} = S(\mybf{b})$ be the new fact.
  By definition of thrifty chase steps, and as $\tau \in \idsr$, we had $b_q \in
  \appelemb{\idsr}{J}{S^q}$, so  $b_q$ was
  $n$-essential because $J$ was. Hence, $\cov(b_q)$ is $n$-essential in
  $\chase{I_0}{\idsb}$.
  By the Directionality Lemma (Lemma~\ref{lem:wadef}), 
  $b'_q \defeq \cov(b_q)$ is also the exported element of the chase
  witness $F_{\w} = S(\mybf{b}')$, and as $b_q$ is $n$-essential, we know by the
  second part of Definition~\ref{def:nreverse} that $b'_q$ was
  introduced in $\chase{I_0}{\idsb}$ at some position~$T^v$ such that $T^v
  \eqids S^q$. This means that $F_{\w}$ was created in $\chase{I_0}{\idsb}$ by
  applying the $\uid$ $\ui{T^v}{S^q}$ which is in~$\idsr$. This implies that, for all $S^r \in
  \pos(S) \backslash \{S^q\}$, the element~$b'_r$ is $(n+1)$-essential
  in $\chase{I_0}{\idsb}$, and is introduced at position~$S^r$.

  Now, let $a \in \dom(J') \backslash \dom(J)$,
  and let $T^v$ such that $a \in \appelemb{\idsr}{J'}{T^v}$. Let $U^w \in \pos(\sigma)$
  and $\tau' : \ui{U^w}{T^v}$ that witness this, i.e., $\tau' \in \idsr$ and $a
  \in \pi_{U^w}(J')$. By the reversibility assumption, we have $U^w \eqids T^v$
  in~$\idsr$.
  Now, by definition of thrifty chase
  steps on aligned superinstances, we know that we defined $\cov'(a) \defeq
  b'_r$ for some $S^r$ where $a$ occurred in~$F_{\fnew}$. Further,
  by definition of thrifty chase
  steps, we know that all positions in which $a$ occurs in~$F_{\fnew}$, and thus all positions
  where it occurs in~$J'$, are $\eqids$-equivalent in~$\idsr$; in particular $S^r
  \eqids U^w$, hence by transitivity $S^r \eqids T^v$.
  By the previous
  paragraph, $\cov(a) = b'_r$ is an $(n+1)$-reversible element introduced
  in $\chase{I_0}{\idsb}$
  at position~$S^r$, and we have $S^r \eqids T^v$. This
  shows that $a$ is $(n+1)$-reversible in~$J'$.
  
  Hence, $J'$ is indeed
  $n$-reversible, and the elements of $\dom(J') \backslash \dom(J)$ are indeed
  $(n+1)$-reversible, which concludes the proof.
\end{proof}

In conjunction with the Fresh Fact-Thrifty Preservation Lemma (Lemma~\ref{lem:fftp}),
this implies that
applying sufficiently many fresh fact-thrifty chase rounds yields an
$n$-essential aligned superinstance:

\begin{lemma}[Ensuring essentiality]
  \label{lem:achkrev}
  For any $n \in \mathbb{N}$, applying $n+1$ fresh fact-thrifty chase rounds on
  a fact-saturated aligned superinstance $J$ by the
  $\uid$s of~$\idsr$ yields an $n$-essential aligned superinstance $J'$.
\end{lemma}

\begin{proof}
  Fix the aligned superinstance $J = (I, \cov)$.
  We use the Fresh Fact-Thrifty Preservation Lemma (Lemma~\ref{lem:fftp})
  to show that the property of being aligned is
  preserved, so we only show that the result is $n$-essential.
  We prove this claim by induction on~$n$.

  For the base case, we must show that the result $J' = (I', \cov')$ of a fresh fact-thrifty
  chase round on~$J$ by~$\idsr$ is $0$-essential. Let $S^q \in \pos(\sigma)$ and $a \in
  \appelemb{\idsr}{J'}{S^q}$. As $\idsr$ is transitively closed, by definition of a chase round, we have $a \in
  \dom(J') \backslash \dom(J)$, because $\uid$ violations on elements
  of~$\dom(J)$ must have been solved in~$J'$; hence, $a$ was created by a fact-thrifty chase step
  on~$J$. By similar reasoning as in the proof of Lemma~\ref{lem:ftok2},
  considering the chase witness $F_\w$ for this chase step, we conclude that
  $\cov(a)$ was introduced at a position $T^v$ in $\chase{I_0}{\idsb}$ such that
  $T^v \eqids S^q$ in~$\idsr$. Further, $\cov(a)$ is vacuously $0$-essential.
  Hence, $J'$ is indeed $0$-essential.

  For the induction step, let $J'$ be the result of $n+1$ fresh fact-thrifty
  chase rounds on $J$, and show that it is $n$-essential. By induction
  hypothesis, the result $J'' = (I'', \cov'')$ of $n$ fresh fact-thrifty chase rounds is
  $(n-1)$-essential.
  Now, again by definition of a chase round, for any position $S^q \in
  \pos(\sigma)$ and $a \in \appelemb{\idsr}{J''}{S^q}$, we must have $a \in
  \dom(J') \backslash \dom(J'')$, so that $a$ was created by applying a
  fact-thrifty chase step on an element~$a''$ in~$J''$ which witnessed a
  violation of a $\uid$ of~$\idsr$. As $J''$ is $(n-1)$-essential, $a'$ was
  $(n-1)$-essential in~$J''$, so we conclude by Lemma~\ref{lem:ftok2} that $a$ is
  $n$-essential in~$J'$. Hence, we conclude that $J'$ is indeed $n$-essential.
\end{proof}

Hence, we can ensure $k$-essentiality. The point of essentiality is to
guarantee that the result of \emph{non-fresh} fact-thrifty chase steps on a
$k$-essential aligned superinstance is also
an aligned superinstance.

\begin{lemma}[Fact-thrifty preservation]
  \label{lem:usekrev}
  For any fact-saturated $k$-essential aligned superinstance $J$ for $\idsb$ and
  $\ufds$, the result $J'$ 
  of \emph{any} fact-thrifty chase step on $J$ by a $\uid$ of $\idsr$ is still a fact-saturated and $k$-essential 
  aligned superinstance.
\end{lemma}

\begin{proof}
  Fix $J' = (I', \cov')$, the $\uid$ $\tau : \ui{R^p}{S^q}$ of~$\idsr$, which is reversible by
  the reversibility assumption, and the element $a \in \dom(J)$ to which it is
  applied.

  The fact that $k$-essentiality is preserved is by Lemma~\ref{lem:ftok2}, and
  fact-saturation is clearly preserved, so we must
  only show that~$J'$ is still an aligned superinstance. The fact that $J'$ is
  a finite superinstance of~$I_0$ is immediate, and it still satisfies $\ufds$
  by Lemma~\ref{lem:rthriftyp} because fact-thrifty chase steps are
  relation-thrifty chase steps. 
  The directionality condition is clearly respected because any new element in $\dom(J')
  \backslash \dom(J)$ occurs at least at the position at which its $\cov'$-image
  was introduced in the chase (namely, the position where it occurs in~$F_\w$),
  and the additional conditions on $\restr{\cov'}{\dom(I_0)}$ and $\restr{\cov'}{\dom(J'
  \backslash I_0)}$ still hold.
  
  The only thing to show is that $\cov'$ is still a
  $k$-bounded simulation.
  Let $F_{\fnew} = S(\mybf{b})$ be the new fact and $F_\w = S(\mybf{b}')$ be the
  chase witness. Now, 
  as in the proof of the Thrifty Steps And Essentiality Lemma (Lemma~\ref{lem:ftok2}), 
  and using the Directionality Lemma (Lemma~\ref{lem:wadef}),
  all elements of $F_\w$ are
  $n$-essential (and, except for $b'_q$, they were introduced at their position of
  $F_\w$).

  Now, to show that $\cov'$ is a $k$-bounded simulation, we use
  Lemma~\ref{lem:addfact}, so it suffices to show that
  we have $\cov(b_r) \bbsim_k b'_r$ for all $r$. This is the case whenever we
  have $\cov(b_r) = b'_r$, which is guaranteed by definition for $S^r = S^q$ and
  for elements in $\danger(S^q)$ such that $S^r \eqfun S^q$ does not hold. For
  non-dangerous elements, the fact that $\cov(b_r) \bbsim_k b'_r$ is by
  definition of fact-thrifty chase steps. For the other positions, there are two
  cases:
  \begin{itemize}
    \item $b_r \in \dom(I)$, in which case $b_r \in \appelemb{\idsr}{I}{S^r}$. As $J$
      is $n$-essential, $\cov(b_r)$ is an $n$-essential element of
      $\chase{I_0}{\idsb}$ introduced at a position $T^v$ such that $T^v \eqids
      S^r$ holds in~$\idsr$. Now, $b'_r$ is an $n$-essential element of
      $\chase{I_0}{\idsb}$ introduced at position $S^r$. By the
      $\uid$ Chase Similarity Theorem (Theorem~\ref{thm:similarityb}), we then have $\cov'(b_r) \bbsim_k b'_r$ in
      $\chase{I_0}{\idsb}$
    \item $b_r \notin \dom(I)$, in which case the claim is immediate unless it occurs
      at multiple positions. However, by definition of thrifty chase steps, all
      positions at which it occurs are related by $\eqids$ in~$\idsr$, so the corresponding
      elements of $F_\w$ are also $\bbsim_k$-equivalent by the $\uid$ Chase
      Similarity Theorem: hence we have
      $\cov'(b_r) \bbsim_k b'_r$.
  \end{itemize}
  We conclude by Lemma~\ref{lem:addfact} that $J'$ is indeed an aligned
  superinstance, which concludes the proof.
\end{proof}

We can now conclude the proof of Theorem~\ref{thm:ksound}.
Let $I_0$ be the initial instance, and consider $J_0 = (I_0, \id)$
which is trivially an aligned superinstance of~$I_0$.
Apply the Fact-Saturated Solutions Lemma (Lemma~\ref{lem:nondetexhaust})
to obtain a fact-saturated aligned
superinstance $J_0' = (I_0', \cov')$. We must now show that we can
complete $J_0'$ to a superinstance that satisfies $\idsr$ as well, which we do with
the following variant of 
the Reversible Relation-Thrifty Completion Proposition (Proposition~\ref{prp:rtcompr}):

\begin{proposition}[Reversible Fact-Thrifty Completion]
  \label{prp:ftcompr}
  For any reversible $\ufds$ and $\idsr$,
  for any transitively closed $\uid$s $\idsb \supseteq \idsr$,
  for any fact-saturated aligned superinstance $J_0'$ of~$I_0$ 
  (for~$\ufds$ and~$\idsb$),
  we can use fact-thrifty chase steps by $\uid$s of $\idsr$
  to construct 
  an aligned fact-saturated superinstance $J_\f$ of~$I_0$ (for~$\ufds$ and
  $\idsb$)
  that satisfies $\idsr$.
\end{proposition}

We conclude this section by proving this proposition. To do so,
we first apply 
the Ensuring Essentiality Lemma (Lemma~\ref{lem:achkrev})
with the $\uid$s of $\idsr$ to make $J_0'$ $k$-essential.
By the Fresh Fact-Thrifty Preservation Lemma (Lemma~\ref{lem:fftp}),
the result $J_1 = (I_1, \cov_1)$ is then a fact-saturated $k$-essential aligned
superinstance of~$I_0$ (for $\idsb$ and $\ufds$). 

We will then use the
Reversible Relation-Thrifty Completion Proposition (Proposition~\ref{prp:rtcompr}) on~$J_1$;
but we must refine it to a stronger claim. We
do so using
the following definition:

\begin{definition}
  \label{def:tsec}
  A \deft{thrifty sequence} on an instance $I$ for $\uid$s $\ids$ and $\ufd$s
  $\ufds$ is a sequence $L$ defined inductively as follows, with an
  \deft{output} $L(I)$ which is a superinstance of~$I$ 
  that we also define inductively:
  \begin{itemize}
    \item The empty sequence $L = ()$ is a thrifty sequence, with $L(I) = I$
    \item Let $L'$ be a thrifty sequence, let $I' = L'(I)$ be the output of
      $L'$, and let $t = (a, \tau, \mybf{b})$ be a triple
      formed of an element $a \in \dom(I')$, a $\uid$ $\tau:
      \ui{R^p}{S^q}$ of $\ids$, and an $\arity{S}$-tuple $\mybf{b}$.
      We require that the fact $S(\mybf{b})$ can be created in $I'$ by applying a thrifty chase step
      to $a$ in $L'(I)$ by $\tau$ (Definition~\ref{def:thrifty}). Then the
      concatenation $L$ of $L'$ and $t$ is a thrifty sequence, and its output
      $L(I)$ is the result of performing this chase step on $L'(I)$, namely,
      $L(I) \defeq L'(I) \sqcup \{S(\mybf{b})\}$. 
  \end{itemize}
  The \deft{length} of~$L$ is written $\card{L}$ and the elements of~$L$ are indexed by
  $L_1, \ldots, L_{\card{L}}$.
  We define a \deft{relation-thrifty sequence} in the same way with
  relation-thrifty steps, and likewise define a
  \deft{fact-thrifty sequence}.
\end{definition}

With this definition, the 
Reversible Relation-Thrifty Completion Proposition (Proposition~\ref{prp:rtcompr}) implies that
there is a relation-thrifty sequence $L$ such that $L(I'_0)$ is
a finite weakly-sound superinstance $I_\f$ of~$I_0$ that satisfies $\ufds$
and~$\idsr$. 
Our goal to prove 
the Reversible Fact-Thrifty Completion Proposition (Proposition~\ref{prp:ftcompr})
is to rewrite $L$ to a fact-thrifty sequence.
To do this, we first need to show that
any two thrifty sequences that
coincide on non-dangerous positions have the same effect in terms of $\uid$
violations:

\begin{definition}
  Let $\ids$ be $\uid$s and $\ufds$ be
  $\ufd$s, let $I_0$ be an instance, and $L$ and $L'$ be thrifty sequences
  on~$I_0$.
  We say that $L$ and $L'$ \deft{non-dangerously match} if $\card{L} = \card{L'}$ and that
  for all $1 \leq i \leq \card{L}$, writing $L_i = (a, \tau, \mybf{b})$ and
  $L_i' = (a', \tau', \mybf{b}')$, we have $a = a'$, $\tau = \tau'$, and,
  writing $\tau: \ui{R^p}{S^q}$, we have $b_r = b'_r$ for all $S^r \in
  \pos(S) \backslash \nondanger(S^q)$.
\end{definition}

\begin{lemma}[Thrifty sequence rewriting]
  \label{lem:thsrw}
  Let $\ids$ be $\uid$s and $\ufds$ be
  $\ufd$s, let $I_0$ be an instance, and let $L$ and $L'$ be thrifty
  sequences on~$I_0$ that non-dangerously match.
  Then $L(I_0)$ satisfies $\ids$ iff $L'(I_0)$ satisfies $\ids$.
\end{lemma}

\begin{proof}
  We prove by induction on the common length of $L$ and $L'$ that,
  if $L$ and $L'$ non-dangerously match,
  then, for all $U^v \in \pos(\sigma)$, we have $\pi_{U^v}(L(I_0)) =
  \pi_{U^v}(L'(I_0))$. If both $L$ and $L'$ have
  length~$0$, the claim is trivial. For the induction step, 
  write $I \defeq L(I_0)$ and $I' \defeq L'(I_0)$.
  Write $L$ as the
  concatenation of $L_2$ and its last tuple $t = (a, \tau, \mybf{b})$, and write
  similarly $L'$ as the concatenation of $L_2'$ and the last tuple $t' = (a', \tau',
  \mybf{b}')$. Let $U^v \in \pos(\sigma)$ and show that $\pi_{U^v}(L(I_0)) =
  \pi_{U^v}(L'(I_0))$.
  Clearly $L_2$ and $L_2'$ non-dangerously match and are strictly shorter than $L$ and $L'$, respectively, so by
  the induction hypothesis, writing $I_2 \defeq L_2(I_0)$ and $I_2' \defeq
  L_2'(I_0)$, we have $\pi_{U^v}(I_2) = \pi_{U^v}(I_2')$. Further, we have $\tau
  = \tau'$; write them  as $\ui{R^p}{S^q}$.
  We then have $I = I_2 \sqcup S(\mybf{b})$, and $I' = I_2' \sqcup S(\mybf{b}')$.
  As we must have $b_r = b'_r$ if $U^v \notin \nondanger(S^q)$, there is nothing
  to show unless we have $U^v \in \nondanger(S^q)$. However, in this case,
  writing $U^v$ as $S^r$, then, by
  definition of thrifty chase steps, we have $b_r \in \pi_{S^r}(I_2)$, so that
  $\pi_{S^r}(I) = \pi_{S^r}(I_2)$. Likewise, $\pi_{S^r}(I') = \pi_{S^r}(I_2')$,
  hence $\pi_{S^r}(I) = \pi_{S^r}(I')$. This concludes the induction proof.

  We now prove the lemma.
  Fix $\tau: \ui{R^p}{S^q}$ in $\ids$. We have $L(I_0) \models \tau$ iff
  $\pi_{R^p}(L(I_0)) \backslash \pi_{S^q}(L(I_0)) = \emptyset$, and likewise for
  $L'(I_0)$. By the result proved in the paragraph above, 
  these conditions are equivalent, and thus we have $L(I_0) \models \tau$ iff $L'(I_0) \models \tau$.
\end{proof}

Hence, consider our fact-saturated aligned superinstance $J_1 = (I_1, \cov_1)$ (for~$\idsb$
and~$\ufds$). As we explained,
the Reversible Relation-Thrifty Completion Proposition (Proposition~\ref{prp:rtcompr})
implies that 
there is a relation-thrifty sequence $L$ such that $L(I_1)$ satisfies $\ids$. We modify
$L$ 
inductively to obtain a \emph{fact-thrifty sequence} $L'$ that non-dangerously
matches~$L$,
in the following manner. Whenever $L$ applies a relation-thrifty step
$t = (a, \tau, \mybf{b})$ to the previous instance $L_2(I_1)$, then observe
that $L_2(I_1)$ is fact-saturated, because $I_1$ was fact-saturated and fact-thrifty chase
steps preserve fact-saturation,
by the Fact-Thrifty Preservation Lemma (Lemma~\ref{lem:usekrev}).
Hence, by that lemma,
instead of applying the relation-thrifty step described by~$t$, 
we can 
choose to apply a fact-thrifty step
on $a$ with $\tau$, defining the new fact using~$\mybf{b}$
except on the non-dangerous positions. By Lemma~\ref{lem:thsrw}, the
resulting $L'$ also ensures that $L'(I_1)$ satisfies $\ids$.

Considering now the fact-thrifty sequence $L'$, as $J_1$ is a fact-saturated $k$-essential
aligned superinstance of~$I_0$ (for $\idsb$ and $\ufds$), 
letting $I_\f \defeq L'(I_1)$,
we can use the Fact-Thrifty Preservation Lemma (Lemma~\ref{lem:usekrev})
to define an aligned fact-saturated superinstance $J_\f = (I_\f, \cov_\f)$ (for
$\idsb$ and $\ufds$),
following each fact-thrifty step, and we have shown that
$I_\f$ satisfies~$\ids$. Hence, we have proven the
Reversible Fact-Thrifty Completion Proposition (Proposition~\ref{prp:ftcompr}).

To prove Theorem~\ref{thm:ksound},
we can simply apply the proposition with $\idsb = \idsr$, and the resulting aligned
superinstance $J_\f = (I_\f, \cov_\f)$ of $I_0$ satisfies $\ids$ and is $k$-sound for $\ucon$ and $\acq$. Further,
it satisfies $\ufds$ and is finite, by definition of being an aligned
superinstance. Hence, $I_\f$ is the desired $k$-universal model, which proves
Theorem~\ref{thm:ksound}.

\subsection{$\uid$ Chase Similarity Theorem}
\label{sec:similarity}

We conclude the \secname by proving the $\uid$ Chase Similarity Theorem:

\similarityb*

Note that this result does not involve $\fd$s, and applies to any arbitrary transitively closed set of~$\uid$s,  not relying
on any finite closure properties, or on the reversibility assumption.
It only
assumes that the last $n$ dependencies used to create~$a$ and~$b$ were
reversible.

\begin{example}
  \label{exa:similarity}
  Consider Figure~\ref{fig:chase} on page~\pageref{fig:chase}, which illustrates
  the neighborhood of two elements, $a$ and $a'$, in the $\uid$ chase by some
  $\uid$s. Each rectangle represents a higher-arity fact, and edges
  represent the $\uid$s used in the chase, with thick edges representing reversible
  $\uid$s.

  The last $\uid$ applied to create $a$ was $\ui{S^2}{R^1}$, and the last $\uid$
  for $a'$ is $\ui{V^1}{U^1}$; they are reversible. Further, $a$ is introduced
  at position $R^3$ and $a'$ at position $U^3$, and $\ui{R^3}{U^3}$ holds and is
  reversible. The theorem claims that $a$ that $a'$ are $1$-bounded-bisimilar,
  which is easily verified; in fact, they are $2$-bounded-bisimilar. This is
  intuitively because all child facts of the $R$-fact at the left must occur at
  the right by definition of the $\uid$ chase, and the parent fact must occur as
  well because of the reverse of the last $\uid$ for~$a$; a similar argument
  ensures that the facts at the right must be reflected at the left.
  
  However, note that $a$ and $a'$ are not $3$-bounded-bisimilar: the $A$-fact at
  the left is not reflected at the right, and vice-versa for the $B$-fact,
  because these $\uid$s are not reversible,
\end{example}

\medskip

To prove the theorem,
fix the instance $I_0$ and the set $\idsb$ of~$\uid$s. 
We first show the following easy lemma:

\begin{lemma}
  \label{lem:samepos}
For any $n > 0$ and position $R^p$, for any two elements $a, b$ of
$\chase{I_0}{\idsb}$ introduced at position $R^p$ in two facts $F_a$ and $F_b$,
letting $a'$ and $b'$ be the exported elements of~$F_a$ and $F_b$, if $a'
  \bbsim_{n-1} b'$, then $a \bbsim_n b$.
\end{lemma}

\begin{proof}
  We proceed by induction on~$n$.
  By symmetry, it suffices to show that $(\chase{I_0}{\idsb}, a) \bsim_n
  (\chase{I_0}{\idsb}, b)$.

  For the base case  $n=1$, 
  observe that, for every fact $F$ of
  $\chase{I_0}{\idsb}$ where $a$ occurs at some position $S^q$,
  there are two cases. Either $F = F_a$, so
  we can pick $F_b$ as the representative fact, or the $\uid$ $\ui{R^p}{S^q}$
  is in $\idsb$ so we can pick a corresponding fact for $b$ by definition of
  the chase. Hence, the claim is shown for $n =1$.

  For the induction step, we proceed in the same way. If $F = F_a$, we pick
  $F_b$ as representative fact, and use either the hypothesis on $a'$ and $b'$ or the induction
  hypothesis (for other elements of~$F_a$ and $F_b$) to justify that $F_b$ is
  suitable. Otherwise, we pick the corresponding fact for $b$ which must
  exist by definition of the chase, and apply the induction hypothesis to the
  other elements of the fact to conclude.
\end{proof}

We now prove the $\uid$ Chase Similarity Theorem (Theorem~\ref{thm:similarityb}). Throughout the proof, we write
$R^p \eqids S^q$ as shorthand to mean that $\ui{R^p}{S^q}$ and $\ui{S^q}{R^p}$
are in~$\idsb$: it is still the case that $\eqids$ is an equivalence relation,
even without the reversibility assumption.

We prove the main claim by induction on~$n$:
for any positions $R^p$ and $S^q$ such that $R^p \eqrev
S^q$, for any two $n$-essential elements $a$ and $b$ respectively introduced at
positions $R^p$ and $S^q$, we have $a \bbsim_n b$.
By symmetry it suffices to show that
$a \bsim_n b$ in $\chase{I_0}{\idsb}$, formally, $(\chase{I_0}{\idsb}, a)
\bsim_n (\chase{I_0}{\idsb}, b)$. 

The base case of~$n = 0$ is immediate.

For the induction step, fix $n > 0$, and assume that the result holds
for $n-1$. Fix $R^p$ and $S^q$ such that $R^p \eqids S^q$, and let $a,
b$ be two $n$-essential elements introduced respectively at $R^p$ and $S^q$ in
facts $F_a$ and $F_b$.
Note that by the induction hypothesis we already
know that $a \bsim_{n-1} b$ in $\chase{I_0}{\idsb}$; we must
show that this holds for~$n$.

First, observe that, as~$a$ and~$b$ are $n$-essential with $n > 0$, they are not
elements of~$I_0$. Hence, by definition of the chase, for each one of them, the
following is true: for each fact of the chase where the element occurs, it only
occurs at one position, and all other elements co-occurring with it in a fact of the chase
occur only at one position and in exactly one of these facts.
Thus, to prove the claim,
it suffices to construct a mapping $\phi$ from the set $N_1(a)$ of the
facts of~$\chase{I_0}{\idsb}$ where $a$ occurs, to the set $N_1(b)$ of the facts
where $b$ occurs, such that the following holds:
for every fact $F = T(\mybf{a})$ of~$N_1(a)$, letting $T^c$ be the one position of
$F$ such that $a_c = a$, the element $b$ occurs at position $T^c$ in $\phi(F) = T(\mybf{b})$, and for every
$i$, we have $a_i \bsim_{n-1} b_i$.

By construction of the chase (using the
Unique Witness Property), $N_1(a)$ consists of exactly the following facts:

\begin{itemize}
  \item The fact $F_a = R(\mybf{a})$, where $a_d = a'$ is the
    exported element (for a certain $R^d \neq R^p$) and
    $a_p = a$ was introduced at $R^p$ in $F_a$. Further, 
    for $i \notin \{p, d\}$, the element $a_i$ was
    introduced at $R^i$ in $F_a$.
  \item For every $\uid$ $\tau : \ui{R^p}{V^g}$ of~$\idsb$, a $V$-fact $F^{\tau}_a$
    where the element at position $V^g$ is~$a$. Further, for $i \neq g$, the
    element at position $V^i$ in $F^{\tau}_a$ was introduced at this position in that fact.
\end{itemize}
A similar characterization holds for $b$: we write the corresponding facts $F_b$
and~$F^{\tau}_b$.
We construct the mapping $\phi$ as follows:
\begin{itemize}
  \item If $R^p = S^q$ then set $\phi(F_a) \defeq F_b$; otherwise, as $\tau : \ui{S^q}{R^p}$
    is in $\idsb$, set $\phi(F_a) \defeq F^{\tau}_b$.
  \item For every $\uid$ $\tau : \ui{R^p}{V^g}$ of~$\idsb$, as $R^p \eqids S^q$,
    by transitivity,
    either $S^q = V^g$ or the $\uid$ $\tau': \ui{S^q}{V^g}$ is in $\idsb$. 
    In the first case, set $\phi(F^{\tau}_a) \defeq F_b$.
    In the second case, set $\phi(F^{\tau}_a) \defeq F^{\tau'}_b$.
\end{itemize}
We must now show that this mapping $\phi$ from $N_1(a)$ to $N_1(b)$ 
satisfies the required conditions. 
Verify that indeed, by construction, whenever $a$ occurs at position $T^c$ in
$F$, then $\phi(F)$ is a $T$-fact where $b$ occurs
at position $T^c$. So we must show that for any $F \in N_1(a)$, writing $F =
T(\mybf{a})$ and $\phi(F) = T(\mybf{b})$, with $a_c = a$ and $b_c = b$ for
some $c$, we have indeed $a_i \bsim_{n-1} b_i$ for all $T^i \in \pos(T)$.
If $n = 1$ there is nothing to show and we are done, so we assume $n \geq 2$.
If $i = c$ then the claim is immediate by the induction hypothesis; otherwise, we distinguish two cases:
\begin{enumerate}
  \item  $F = F_a$ (so that $T = R$ and $c = p$), or $F = F_a^{\tau}$ such that the $\uid$
    $\tau : \ui{R^p}{T^c}$ is
    reversible, meaning that $\tau^{-1} \in \idsb$. In this case, by construction,
    either $\phi(F) = F_b$ or $\phi(F) = F^{\tau'}_b$ for $\tau' :
    \ui{S^q}{T^c}$; $\tau'$ is then reversible,
    because $R^p \eqids S^q$ and $R^p \eqids T^c$.

    We show that for all $1 \leq i \leq \arity{T}$ such that $i \neq c$, the
    element $a_i$ is
    $(n-1)$-essential and was introduced in~$\chase{I_0}{\idsb}$ at a position
    in the $\eqids$-class of~$T^i$. Once we have proved this, we can
    show the same for all~$b_i$ in a symmetric way, so that we can conclude that $a_i \bsim_{n-1}
    b_i$ by induction hypothesis.
    To see why the claim holds, we distinguish two subcases.
    Either $a_i$ was introduced in~$F$, or we have $F = F_a$, $i = d$
    and $a_i$ is the exported element for~$a$.

    In the first subcase, $a_i$ was created by applying the reversible $\uid$
    $\tau$ and the exported element was~$a$, which is $n$-essential, so $a_i$ is
    $(n-1)$-essential (in fact it is even $(n+1)$-essential), and $a_i$ is introduced
    at position~$T^i$.
    
    In the second subcase, $a_i$ is the exported element used to
    create~$a$, which is $n$-essential, so $a_i$ is $(n-1)$-essential; and as
    $n \geq 2$, the last dependency applied to create $a_i$ is reversible, so
    that $a_i$ was introduced at a position in the same $\eqids$-class as~$T^i$.

    Hence, we have proved the desired claim for the first case. 

  \item $F = F_a^{\tau}$ such that $\tau : \ui{R^p}{T^c}$ is not reversible. In this
    case, we cannot have $T^c = S^q$ (because we have $R^p \eqids S^q$),
    so we must have $\phi(F) = F_b^{\tau'}$ with $\tau' : \ui{S^q}{T^c}$.
    Now, all $a_i$ for $i \neq c$ were introduced in~$F$ at position $T^i$,
    and likewise for the~$b_i$ in~$\phi(F)$. 
    Using Lemma~\ref{lem:samepos}, as $a \bbsim_{n-1} b$, we conclude
    that $a_i \bbsim_{n} b_i$, hence $a_i \bsim_{n-1} b$.
\end{enumerate}
This concludes the proof of 
the $\uid$ Chase Similarity Theorem (Theorem~\ref{thm:similarityb}), 
thus completing the proof of Theorem~\ref{thm:ksound}.

\section{Decomposing the Constraints}
\label{sec:manyscc}
In this \secname, we lift the reversibility assumption, proving:

\begin{maintheorem}
  \label{thm:manyscc}
  Finitely-closed $\uid$s and $\ufd$s have finite $k$-universal models
  for $\acq$s.
\end{maintheorem}

\subsection{Partitioning the $\uid$s}

We write the $\ufd$s as $\ufds$ and the $\uid$s as $\idsb$.
We will proceed by partitioning $\idsb$
into subsets of $\uid$s which are either reversible or
are much simpler to deal with.

Our desired notion of partition respects an order on $\uid$, which we now
define. As we will show (Lemma~\ref{lem:sccbyscc}),
the order is also respected by thrifty chase steps.

\begin{definition}
  \label{def:opartition}
  For any $\tau, \tau' \in \idsb$,
  we write $\tau \idprec \tau'$ when we can write $\tau: \ui{R^p}{S^q}$ and
  $\tau':
  \ui{S^r}{T^v}$ with $S^q \neq S^r$, and the $\ufd$ $S^r \rightarrow S^q$ is in
  $\ufds$.
  An \deft{ordered partition} $(P_1, \ldots, P_{\neqidsc})$ of~$\idsb$ is a
  partition of $\idsb$ (i.e., $\idsb = \bigsqcup_i P_i$)
  such that for any $\tau \in P_i$, $\tau' \in P_j$, if
  $\tau \idprec \tau'$ then $i \leq j$.
\end{definition}

The point of partitioning $\idsb$ is to be able to control the
structure of the $\uid$s in each class:

\begin{definition}
  \label{def:manageable}
  We call $P \subseteq \idsb$ \deft{reversible} if $P$ and $\ufds$ are
  reversible (Definition~\ref{def:reversible}).
  We say $P
  \subseteq \idsb$ is \deft{trivial} if we have $P = \{\tau\}$ for some $\tau
  \in \idsb$ such that $\tau \not\idprec \tau$. A partition is
  \deft{manageable} if it is ordered and all of its classes are either reversible or trivial.
\end{definition}

As we will show in Section~\ref{sec:manageable},
we can always construct a manageable partition of~$\idsb$:

\begin{proposition}
  \label{prp:nontrivconv}
  Any conjunction $\idsb$ of~$\uid$s closed under finite implication has a
  manageable partition.
\end{proposition}

\begin{figure}
  \centering
  \begin{tikzpicture}[
  text height=1.3ex,text depth=0ex,xscale=1.5,yscale=1.75,
  mispos/.append style={red, font=\scriptsize},
  newfact/.append style={red, dashed, -{>[scale=1.5]}},
  sccp/.append style={pattern=north west lines, pattern
  color=darkgreen,opacity=.2},
  scc/.append style={darkgreen}
]
  \node (r1) at (0, 0) {1};
  \node (r2) at (1, 0) {2};
  \node (r3) at (2, 0) {3};

  \node (s1) at (4, 0) {1};
  \node (s2) at (5, 0) {2};
  \node (s3) at (6, 0) {3};

  \path (r1) edge[<->,bend left=25] (r2);
  \path (s2) edge[<->,bend left=25] (s3);
  \path (r3) edge[->,bend left=12.5] (s1);
  \path (r1) edge[newfact,<->,bend right=45] (r2);
  \path (s2) edge[newfact,<->,bend right=45] (s3);
  \path (r1) edge[newfact,bend left=40] (r3);
  \path (r2) edge[newfact,bend left=40] (r3);
  \path (s1) edge[newfact,bend left=40] (s2);
  \path (s1) edge[newfact,bend left=40] (s3);

  \node (R) at (-.5,0) {$R$};
  \draw (-.3, -.3) rectangle (2.3, .3);

  \node (S) at (6.5,0) {$S$};
  \draw (3.7, -.3) rectangle (6.3, .3);

  \draw[scc] (-.2, -.2) rectangle (1.2, .2);
  \draw[sccp] (-.2, -.2) rectangle (1.2, .2);

  \draw[scc] (4.8, -.2) rectangle (6.2, .2);
  \draw[sccp] (4.8, -.2) rectangle (6.2, .2);

  \draw[scc] (1.8, -.2) rectangle (4.2, .2);
  \draw[sccp] (1.8, -.2) rectangle (4.2, .2);
\end{tikzpicture}
  \caption{Manageable partition (see Example~\ref{exa:scc})}
  \label{fig:scc}
\end{figure}
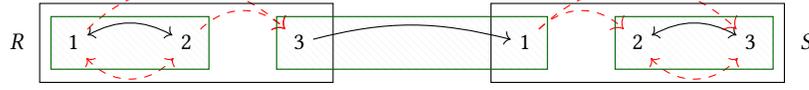

\begin{example}
  \label{exa:scc}
  Consider two ternary relations $R$ and $S$.
  Consider the $\uid$s $\tau_R: \ui{R^1}{R^2}$,
  $\tau_S: \ui{S^2}{S^3}$, $\tau_{RS}: \ui{R^3}{S^1}$, and the $\ufd$s
  $\phi_R: R^1 \rightarrow R^2$,
  $\phi_S: S^2 \rightarrow S^3$,
  $\phi_R': R^3 \rightarrow R^1$, and
  $\phi_S' : S^3 \rightarrow S^1$.
  The $\uid$s $\tau_R^{-1}$ and $\tau_S^{-1}$, and the $\ufd$s
  $\phi_R^{-1}$, $\phi_S^{-1}$, and $R^3 \rightarrow R^2$, $S^2 \rightarrow
  S^1$, are finitely implied. The two relations $R$ and~$S$ are illustrated in
  Figure~\ref{fig:scc}, where $\uid$s are drawn as solid black edges, and
  $\ufd$s as dashed red edges that are \emph{reversed} (this follows
  Definition~\ref{def:constraint}).

  A manageable partition of the $\uid$s of the finite closure is $(\{\tau_R,
  \tau_R^{-1}\}, \{\tau_{RS}\}, \{\tau_S, \tau_S^{-1}\})$, where the first and
  third classes are reversible and the second is trivial. The classes of the
  partition are drawn as green hatched rectangles in Figure~\ref{fig:scc}; they are intuitively related to
  a topological sort of the graph of the black and red edges (see
  Definition~\ref{def:partfromsort}).
\end{example}

\subsection{Using Manageable Partitions}

Fix the instance $I_0$ and the finitely closed constraints $\uconb$ formed of
$\uid$s $\idsb$ and $\ufd$s $\ufds$.
To prove Theorem~\ref{thm:manyscc}, starting with the initial aligned
superinstance $J_0 = (I_0, \id)$ of $I_0$ (for $\idsb$ and $\ufds$), we first note that 
the Fact-Saturated Solutions Lemma (Lemma~\ref{lem:nondetexhaust})
does not use the reversibility assumption. Hence, we apply it (with~$\idsb$) to obtain from $I_0$ an
aligned fact-saturated superinstance $J_1$ of~$I_0$ (for~$\ufds$ and~$\idsb$).
This is the \deft{saturation process}.

The goal is now to apply a \deft{completion process} to satisfy $\idsb$, which
we formalize as the following proposition. Recall the definition of thrifty
sequences (Definition~\ref{def:tsec}). We refine the definition below.

\begin{definition}
  \label{def:preserving}
  We define a  \deft{preserving
  fact-thrifty sequence}  $L$
  (for $\uid$s $\idsb$ and $\ufd$s $\ufds$)
  on any
  \emph{fact-saturated aligned superinstance} $J$ of~$I_0$
  in the following inductive way,
  with its \deft{output} $L(J)$ also being a fact-saturated aligned
  superinstance:
  \begin{itemize}
    \item The empty list $L = ()$ is a preserving fact-thrifty sequence, with
      output $L(J) \defeq J$.
    \item Let $L$ be the concatenation of a preserving fact-thrifty sequence $L'$ and a
      triple $t = (a, \tau, \mybf{b})$. Let $J' \defeq L'(J)$ be the output of $L'$.
      We call $L$ \deft{preserving}, 
      if one of the following holds:
      \begin{itemize}
        \item $t$ is a fresh fact-thrifty chase step. 
          In this case,
          by the Fresh Fact-Thrifty Preservation Lemma (Lemma~\ref{lem:fftp}),
          $J'$ is indeed a fact-saturated aligned
          superinstance of~$I_0$. 
        \item $J'$ is $k$-essential for some subset $\idsr$ of
          $\idsb$ such that $\tau \in \idsr$ and $\idsr$ and $\ufds$ are
          reversible.
          In this case, 
          by the Fact-Thrifty Preservation Lemma (Lemma~\ref{lem:usekrev}),
          $J$ is a fact-saturated aligned superinstance
          of~$I_0$ (which is also $k$-essential for the same subset).
      \end{itemize}
      In either case, the output of $L$ is the aligned superinstance obtained
      as the result of applying the fact-thrifty chase step 
      represented by~$t$ on~$J'$. \qedef
  \end{itemize}
\end{definition}

We can now state our intended result, which implies Theorem~\ref{thm:manyscc}:
  
\begin{proposition}[Fact-thrifty completion]
  \label{prp:ftcomp}
  Let $\uconb = \ufds \sqcup \idsb$ be finitely closed $\ufd$s and $\uid$s,
  and  $I_0$ be an instance  satisfying $\ufd$s.
  For any fact-saturated aligned superinstance $J$ of $I_0$ for $\uconb$,
  there is a preserving fact-thrifty sequence $L$ such that $L(J)$ satisfies~$\idsb$.
\end{proposition}

We prove Proposition~\ref{prp:ftcomp}, and from it Theorem~\ref{thm:manyscc}, in
the rest of the \subsecname.
We construct a manageable partition $\mybf{P} = (P_1, \ldots, P_{\neqidsc})$ of~$\idsb$
using Proposition~\ref{prp:nontrivconv}.
Now, for $1 \leq i \leq \neqidsc$, we use fact-thrifty chase steps by $\uid$s
of~$P_i$ to
extend the
fact-saturated aligned superinstance $J_i$ to a larger one~$J_{i+1}$
that satisfies~$P_i$.

The crucial point is that we can apply fact-thrifty chase steps to satisfy~$P_i$
without creating any new violations of~$P_j$ for $j < i$, and hence 
we can make progress following the partition. The reason for this is the following
easy fact about thrifty chase steps:

\begin{lemma}
  \label{lem:sccbyscc}
  Let $J$ be an aligned superinstance of~$I_0$ and $J'$ be the result of
  applying a thrifty chase step on $J$ for a $\uid$ $\tau$ of\/~$\idsb$.
  Assume that a $\uid$ $\tau'$ of\/~$\idsb$ was satisfied by $J$
  but is not satisfied by $J'$. Then $\tau \idprec \tau'$.
\end{lemma}

\begin{proof}
Fix $J$, $J'$, $\tau : \ui{R^p}{S^q}$ and $\tau'$. As chase steps add a
single fact, the only new $\uid$ violations in $J'$ relative to $I$ are on
  elements in the newly created fact $F_{\fnew} = S(\mybf{b})$, As $\idsb$ is
  transitively closed, $F_{\fnew}$ can introduce no new violation on the exported element
$b_q$. Now, as thrifty chase steps
always reuse existing elements at
non-dangerous positions, we know that if $S^r \in \nondanger(S^q)$ then
no new $\uid$ can be applicable to $b_r$. Hence, if a new $\uid$ is applicable
to $b_r$ for $S^r \in \pos(S)$, then necessarily $S^r \in \danger(S^q)$. By
definition of dangerous positions, the $\ufd$ $S^r \rightarrow S^q$ is then in
$\ufds$, and we have $S^r \neq S^q$.
Hence, writing $\tau' : \ui{S^r}{T^r}$, we see that $\tau \idprec \tau'$.
\end{proof}

The lemma justifies our definition of ordered partition, since it will allow us to 
prove Proposition~\ref{prp:ftcomp} inductively.
Using
 the fact that $\mybf{P}$ is
ordered ensures that we can indeed apply fact-thrifty chase steps to satisfy
each $P_i$ individually, dealing with them in the order of the partition.

Thus, to prove Proposition~\ref{prp:ftcomp}, consider each
class $P_i$ in order.
As $\mybf{P}$ is manageable, there are two cases: either $P_i$ is
trivial or it is reversible.

First, if $P_i$ is trivial, it can simply be satisfied by a preserving fact-thrifty sequence $L_i$ of
fresh fact-thrifty chase steps using the one $\uid$ of~$P_i$. This follows from
Lemma~\ref{lem:sccbyscc}.

\begin{lemma}
  \label{lem:trivscc}
  For any trivial class $\{\tau\}$, 
  performing one chase round on an aligned fact-saturated superinstance $J$ of~$I_0$ by
  fresh fact-thrifty chase steps
  for $\tau$ yields an aligned superinstance $J'$ of~$I_0$ that satisfies~$\tau$.
\end{lemma}

\begin{proof}
Fix $J$, $J'$ and $\tau$. All violations of~$\tau$ in $J$ have been satisfied in
$J'$ by definition of~$J'$, so we only have to show that no new violations of
$\tau$ were introduced in $J'$. But by
Lemma~\ref{lem:sccbyscc}, as $\tau \not\idprec \tau$, each fresh fact-thrifty chase
step cannot introduce such a violation, hence there is no new violation of
$\tau$ in $J'$. Hence, $J' \models \tau$.
\end{proof}

Second,
returning to the proof of Proposition~\ref{prp:ftcomp},
the interesting case is that of a \emph{reversible} $P_i$,
for which we have done the work of the last three \secnames. We satisfy a reversible
$P_i$ by a preserving fact-thrifty sequence~$L_i$
obtained using the
Reversible Fact-Thrifty Completion Proposition (Proposition~\ref{prp:ftcompr}).
Indeed, $J_i$ is a
fact-saturated aligned superinstance of $I_0$ for $\ufds$ and $\idsb$, and by
definition of $P_i$ being reversible, letting $\idsr \defeq P_i$, the constraints
$\ufds$ and $\idsr$ are reversible. By the
Reversible Fact-Thrifty Completion Proposition, we can thus construct a
fact-thrifty sequence~$L_i$ (by $\uid$s of~$\idsr$)
such that $J_{i+1} \defeq L_i(J_i)$ is a
fact-saturated 
aligned
superinstance of~$I_0$ for $\ufds$ and $\idsb$ that
satisfies $P_i$. Further, from the proof, it is clear that $L_i$ is preserving.

Hence, in either of the two cases, we construct a preserving fact-thrifty sequence~$L_i$
and $J_{i+1} \defeq L_i(J_i)$ satisfies
$P_i$. Further, as $L_i$ only performs fact-thrifty chase steps by $\uid$s of~$P_i$,
$J_{i+1}$ actually satisfies $\bigcup_{j \leq i} P_j$, thanks
to Lemma~\ref{lem:sccbyscc}.

The concatenation of the preserving fact-thrifty sequences $L_i$ for each $P_i$
is
thus a preserving fact-thrifty sequence $L$ whose final result $L(J) = J_{n+1}$ is thus an aligned superinstance of
$I_0$ that satisfies $\idsb$, which proves 
the Fact-Thrifty Completion Proposition (Proposition~\ref{prp:ftcomp}).
As an aligned superinstance, $J_{n+1}$
is also finite, satisfies $\ufds$, and is $k$-sound for $\acq$; so it is
$k$-universal for $\uconb$ and $\acq$. This concludes the proof of
Theorem~\ref{thm:manyscc}.

\subsection{Building Manageable Partitions}
\label{sec:manageable}

The only missing part is to show how manageable partitions are constructed
(Proposition~\ref{prp:nontrivconv}),
which we show  in this \subsecname.
We will construct the
manageable partition using a \defo{constraint graph} defined
from the dependencies, inspired by the multigraph used in the proof of 
Theorem~\ref{thm:maincosm} in \cite{cosm}.

\begin{definition}
  \label{def:constraint}
  Given a set $\uconb$ of finitely closed $\uid$s and $\ufd$s
  on signature~$\sigma$,
  the \deft{constraint graph} $G(\uconb)$ is the directed graph
  with vertex set $\positions(\sigma)$ and with the following edges:
  \begin{itemize}
    \item For each $\uid$ $\ui{R^p}{S^q}$ in $\uconb$, an edge from $R^p$ to $S^q$
    \item For each $\ufd$ $R^a \rightarrow R^b$ in $\uconb$,
      an edge from $R^b$ to $R^a$.
  \end{itemize}
  As we forbid trivial $\uid$s and $\ufd$s, $G(\uconb)$ has no self-loop, but it
  may contain both the edge $(R^p, S^q)$ and $(S^q, R^p)$. However, we do not
  represent multiple edges in $G(\uconb)$: for instance, if the $\uid$
  $\ui{R^a}{R^b}$ and the $\ufd$ $R^b \rightarrow R^a$ are in $G(\uconb)$, we
  only create a single copy of the edge $(R^a, R^b)$.
\end{definition}

Hence, fix the finitely closed $\uid$s and $\ufd$s $\uconb \defeq \idsb \wedge
\ufds$, and construct
the graph $G(\uconb)$. As observed by \cite{cosm}, the graph
$G(\uconb)$ has the following property, which will be needed to show that classes
are reversible:

\begin{lemma}
  \label{lem:cyclereverse}
  For any edge $e$ occurring in a cycle in $G(\uconb)$, for any dependency $\tau$
  which caused the creation of $e$, the reverse $\tau^{-1}$ of $\tau$ is
  in~$\uconb$.
\end{lemma}

\begin{proof}
  Let $e_1$ be the edge, and $e_1, \ldots, e_n$ be the cycle (the first vertex
  of $e_1$ is the second vertex of $e_n$), and let $\tau$ be the dependency.
  Consider a cycle of dependencies $\tau_1, \ldots, \tau_n$, with $\tau_1 =
  \tau$, such that each $\tau_i$ caused the creation of edge $e_i$ in
  $G(\uconb)$. We must show that the reverse
  $\tau^{-1}$ of $\tau$ is in $\uconb$.

  If all the $\tau_i$ are $\uid$s, then, as $\idsb$ is closed under the
  $\uid$ transitivity rule, we apply it to $\tau_2, \ldots, \tau_n$
  and deduce that $\tau_1^{-1}$ is in $\idsb$. Likewise, if all the
  $\tau_i$ are $\ufd$s, then we proceed in the same way because $\ufds$ is
  closed under the $\ufd$ transitivity rule.

  If the $\tau_i$ are of alternating types (alternatively $\uid$s and $\ufd$s),
  then, recalling that $\uconb$ is closed under the \emph{cycle rule} (see
  Section~\ref{sec:implication})
  we deduce that $\tau_i^{-1}$ is in $\uconb$ for all $i$.

  In the general case, 
  consider the maximal subsequence $\tau_j, \ldots, \tau_n, \tau_1, \ldots,
  \tau_i$ ($i < j$) of consecutive dependencies in the cycle
  that includes $\tau$ and where all dependencies are of the same type. Let
  $\tau_{\mathrm{m}}$ be
  the result of combining these dependencies by the $\uid$ or $\ufd$
  transitivity rule (depending on whether they are $\uid$s or $\ufd$s),
  and consider the cycle $\tau_{\mathrm{m}}, \tau_{i+1}, \ldots, \tau_n, \tau_1,
  \ldots, \tau_{j-1}$. Collapsing all other consecutive sequences of
  dependencies to a single dependency using the $\uid$ and $\ufd$ transitivity
  rules, and applying the
  cycle rule as in the previous case,
  we deduce
  that $\tau_{\mathrm{m}}^{-1}$ is in~$\uconb$. Hence, the cycle $\tau_j,
  \ldots, \tau_n, \tau_1, \ldots, \tau_i, \tau_{\mathrm{m}}^{-1}$ is a cycle of
  dependencies of the same type as~$\tau$, and it includes~$\tau$, so we
  conclude as in the first two cases that $\tau^{-1}$ is in $\uconb$.

  Hence, in all cases $\tau^{-1}$ is in $\uconb$. This concludes the proof.
\end{proof}

Compute the strongly connected components of~$G(\uconb)$, ordered following a
topological sort: we label them $V_1, \ldots, V_n$.
The order of the $V_i$ guarantees that there are no edges in $G(\uconb)$
from $V_i$ to $V_j$ unless $i \leq j$.

We will build each class of the manageable partition, either as the set of
$\uid$s within the positions of an SCC (a \emph{reversible} class), or as a
singleton $\uid$ going from a class $V_i$ to a class $V_j$ with $j > i$ (a
\emph{trivial} class). Formally:

\begin{definition}
  \label{def:partfromsort}
  The topological sort of the SCCs of $G(\uconb)$, written $V_1, \ldots, V_n$, defines a
  partition $\mybf{P}$ of the $\uid$s of $\idsb$, in the following manner. For each $V_i$,
  if there are any non-trivial $\uid$s of the form $\ui{R^p}{S^q}$ with $R^p, S^q \in V_i$,
  create a class of $\uid$s (the \emph{main} class) containing all of
  them. Then, for each $\uid$ of the form $\ui{R^p}{S^q}$ with $R^p \in V_i$ and
  $S^q \in V_j$ with $j > i$, create a singleton class of $\uid$s containing
  exactly that $\uid$ (a \emph{satellite} class). The partition $\mybf{P}$ is obtained by
  taking the concatenation, for $i$ from $1$ to $n$, of the main class of~$V_i$ (if it exists)
  and then all satellite classes of~$V_i$ (if any) in an arbitrary order.
\end{definition}

Remember that, while the constraint graph reflects both the $\uid$s and the
$\ufd$s, the partition $\mybf{P}$ that we define is a partition of~$\idsb$, that
is, a partition of $\uid$s, and does not contain $\ufd$s. We first show that $\mybf{P}$
is indeed a partition, and then that it is an ordered partition.

\begin{lemma}
  $\mybf{P}$ is indeed a partition of $\idsb$.
\end{lemma}

\begin{proof}
  As the SCCs of $G(\uconb)$ partition the vertex set of $G(\uconb)$,
it is clear by construction that any $\uid$ occurs in at most a single class of
the partition, which must be a class for the SCC of its first position, and
either a satellite class or the main class depending on the SCC of its second
position.

Conversely, each $\uid$ $\tau$ is reflected in some class of the partition, for the SCC
$V_i$ of its first position: either the second position of $\tau$ is also
in~$V_i$, so $\tau$
is in the main class for~$V_i$; or 
the second position of~$\tau$ is in an SCC $V_j$ with $i \neq j$, in
which case $i < j$ by definition of a topological sort, and $\tau$ is in some
satellite class for~$V_i$.
Hence, $\mybf{P}$ is indeed a partition of~$\ids$.
\end{proof}

\begin{lemma}
  $\mybf{P}$ is an ordered partition.
\end{lemma}

\begin{proof}
Assume by way of contradiction that there are two classes $P_i$ and $P_j$ and
$\tau \in P_i$, $\tau' \in P_j$, such that $\tau \idprec \tau'$ but $i > j$. Let
$V_p$ and $V_q$ be the SCCs in which $P_i$ and $P_j$ were created. We must have
$p \geq q$, so there are two
possibilities.

First, if $p =
q$, then the first positions of $\tau$ and $\tau'$ must both be in $V_p = V_q$,
and as $P_i$ is not the first class created for $V_p = V_q$, it must be a
satellite class.
Hence, the second position of $\tau$ is in another
SCC, say $V_r$, with $r > p$. Now, as $\tau \idprec \tau'$,
there is a $\ufd$ from the first position of $\tau'$ to the
second position of $\tau$, which implies that there is an edge from $V_r$ to
$V_p$ in $G(\uconb)$. As $r > p$, this contradicts the fact that the SCCs are ordered
following a topological sort.

Second, if we have $p > q$, then again the first position of $\tau$ must be in
$V_p$, and the first position of $\tau'$ is in $V_q$. Let $V_r$ be the SCC of
the second position of $\tau$. As $\tau \idprec \tau'$, the $\ufd$ from the first
position of $\tau'$ to the second position of $\tau$ witnesses that there is an
edge in $G(\uconb)$ from $V_r$ to $V_q$. Hence, we must have $r \leq q$. But $\tau$
witnesses that there is an edge from $p$ to $r$ in $G(\uconb)$, so that we must have $p
\leq r$. Hence, $p \leq q$, but we had assumed $p > q$, a contradiction.
\end{proof}

We now show that $\mybf{P}$ is manageable, by considering each class and showing that it
is either trivial or that it is reversible:

\begin{lemma}
  Each satellite class in $\mybf{P}$ is trivial.
\end{lemma}

\begin{proof}
  Each satellite class consists by construction of a singleton dependency
  $\tau: \ui{R^p}{S^q}$, implying the existence of an edge in the
  constraint graph $G(\uconb)$ from $R^p$ to $S^q$. Assume by way of contradiction that
  $\tau \idprec \tau$. This implies that $R^p \rightarrow S^q$ is in $\ufds$,
  so there is an edge in $G(\uconb)$ from $S^q$ to $R^p$. Hence, $\{R^p, S^q\}$ is
  strongly connected, so $R^p$ and $S^q$ belong to the same SCC, which
  contradicts the definition of a satellite class.
\end{proof}

\begin{lemma}
  Each main class in the partition is reversible.
\end{lemma}

\begin{proof}
  Let $P_i$ be the class and $V_i$ be the corresponding SCC.
  We first show that $P_i$ is transitively closed. Consider two $\uid$s
  $\tau$ and $\tau'$ of $P_i$ that would be combined by the transitivity rule to
  the $\uid$ $\tau''$. As $\idsb$ is transitively closed, we have $\tau'' \in
  \idsb$. Now, if both $\tau$ and $\tau'$  have both positions in $V_i$,
  then so does $\tau''$, so we also have $\tau'' \in P_i$.

  Second, to see that every $\uid$ $\tau$ in $P_i$ is reversible, consider a
  $\uid$ $\tau:
  \ui{R^p}{S^q}$ of~$P_i$, with $R^p, S^q \in V_i$. We forbid trivial $\uid$s, so $R^p
  \neq S^q$. As $V_i$ is strongly connected, consider a directed path $\pi$ of
  edges of $G(\uconb)$ from $S^q$ to $R^p$. Combining $\pi$ with the edge created in
  $G(\uconb)$ for the $\uid$ $\tau$, we deduce the existence of a cycle in
  $G(\uconb)$. Hence, by Lemma~\ref{lem:cyclereverse}, the $\uid$ $\tau^{-1}$
  is in $\idsb$, and it also has both positions in~$V_i$, so $\tau^{-1}$ is
  in~$P_i$.

  Third, we prove the claim about $\ufd$s. Consider a $\ufd$ $\phi: R^p \rightarrow
  R^q$ of $\ufds$, with $R^p \neq R^q$. Assume that $R^p$ and $R^q$ occur in a
  $\uid$ of $P_i$; thus  $R^p, R^q \in V_i$. Reasoning
  as before, we find a cycle in $G(\uconb)$ that includes the edge that
  corresponds to~$\phi$,
  and deduce that $\phi^{-1}$ is in $\ufds$.
\end{proof}

Hence, $\mybf{P}$ is an ordered partition
of~$\idsb$ where each class is either reversible or trivial, i.e., it is a
manageable partition. This concludes the proof of
Proposition~\ref{prp:nontrivconv}.

\section{Higher-Arity FD\lowercase{s}}
\label{sec:hfds}
The goal of this \secname is to generalize our results to functional dependencies
of arbitrary arity:

\begin{maintheorem}
  \label{thm:hfds}
  Finitely-closed $\uid$s and $\fd$s have finite universal models
  for $\acq$s.
\end{maintheorem}

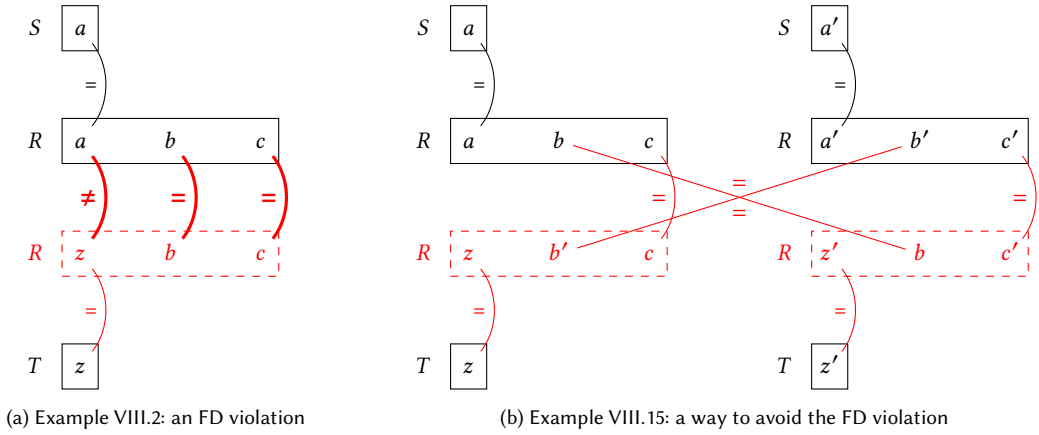
\begin{figure}
\noindent\begin{subfigure}[b]{.4\linewidth}
  \centering
  \begin{tikzpicture}[
  text height=1.3ex,text depth=0ex,xscale=1.2,yscale=1.5,
  mispos/.append style={red, font=\scriptsize},
  newfact/.append style={red, dashed, -{>[scale=1.5]}},
  sccp/.append style={pattern=north west lines, pattern
  color=darkgreen,opacity=.2},
  scc/.append style={darkgreen}
]
  \node (s1) at (0, 0) {$a$};
  \draw (-.2, -.2) rectangle (0.2, .2);
  \node (S) at (-.5,0) {$S$};

  \node (r1) at (0, -1) {$a$};
  \node (r2) at (1, -1) {$b$};
  \node (r3) at (2, -1) {$c$};
  \draw (-.2, -1.2) rectangle (2.2, -0.8);
  \node (R) at (-.5,-1) {$R$};

  \path (s1) edge[bend left=45] node[left] {=} (r1);

  \node (t1) at (0, -3) {$z$};
  \draw (-.2, -3.2) rectangle (0.2, -2.8);
  \node (T) at (-.5,-3) {$T$};

  \node[red] (rr1) at (0, -2) {$z$};
  \node[red] (rr2) at (1, -2) {$b$};
  \node[red] (rr3) at (2, -2) {$c$};
  \draw[red,dashed] (-.2, -2.2) rectangle (2.2, -1.8);
  \node[red] (RR) at (-.5,-2) {$R$};

  \path[red] (rr1) edge[bend left=45] node[left] {=} (t1);
  
  \path[red,very thick] (r1) edge[bend left=45] node[left] {$\bm{\neq}$} (rr1);
  \path[red,very thick] (r2) edge[bend left=45] node[left] {$\bm{=}$} (rr2);
  \path[red,very thick] (r3) edge[bend left=45] node[left] {$\bm{=}$} (rr3);
\end{tikzpicture}
  \caption{Example~\ref{exa:hfdviol}: an $\fd$ violation}
  \label{fig:hfds}
\end{subfigure}\hfill\begin{subfigure}[b]{.6\linewidth}
  \centering
  \begin{tikzpicture}[
  text height=1.3ex,text depth=0ex,xscale=1.2,yscale=1.5,
  mispos/.append style={red, font=\scriptsize},
  newfact/.append style={red, dashed, -{>[scale=1.5]}},
  sccp/.append style={pattern=north west lines, pattern
  color=darkgreen,opacity=.2},
  scc/.append style={darkgreen}
]
  \node (s1) at (0, 0) {$a$};
  \draw (-.2, -.2) rectangle (0.2, .2);
  \node (S) at (-.5,0) {$S$};

  \node (r1) at (0, -1) {$a$};
  \node (r2) at (1, -1) {$b$};
  \node (r3) at (2, -1) {$c$};
  \draw (-.2, -1.2) rectangle (2.2, -0.8);
  \node (R) at (-.5,-1) {$R$};

  \path (s1) edge[bend left=45] node[left] {=} (r1);

  \node (t1) at (0, -3) {$z$};
  \draw (-.2, -3.2) rectangle (0.2, -2.8);
  \node (T) at (-.5,-3) {$T$};

  \node[red] (rr1) at (0, -2) {$z$};
  \node[red] (rr2) at (1, -2) {$b'$};
  \node[red] (rr3) at (2, -2) {$c$};
  \draw[red,dashed] (-.2, -2.2) rectangle (2.2, -1.8);
  \node[red] (RR) at (-.5,-2) {$R$};

  \path[red] (rr1) edge[bend left=45] node[left] {=} (t1);
  
  \node (bs1) at (4, 0) {$a'$};
  \draw (3.8, -.2) rectangle (4.2, .2);
  \node (bS) at (3.5,0) {$S$};

  \node (br1) at (4, -1) {$a'$};
  \node (br2) at (5, -1) {$b'$};
  \node (br3) at (6, -1) {$c'$};
  \draw (3.8, -1.2) rectangle (6.2, -0.8);
  \node (bR) at (3.5,-1) {$R$};

  \path (bs1) edge[bend left=45] node[left] {=} (br1);

  \node (bt1) at (4, -3) {$z'$};
  \draw (3.8, -3.2) rectangle (4.2, -2.8);
  \node (bT) at (3.5,-3) {$T$};

  \node[red] (brr1) at (4, -2) {$z'$};
  \node[red] (brr2) at (5, -2) {$b$};
  \node[red] (brr3) at (6, -2) {$c'$};
  \draw[red,dashed] (3.8, -2.2) rectangle (6.2, -1.8);
  \node[red] (bRR) at (3.5,-2) {$R$};

  \path[red] (brr1) edge[bend left=45] node[left] {=} (bt1);
  
  \path[red] (r3) edge[bend left=45] node[left] {$=$} (rr3);
  \path[red] (br3) edge[bend left=45] node[left] {$=$} (brr3);
  
  \path[red] (r2) edge[] node[above] {$=$} (brr2);
  \path[red] (br2) edge[] node[below] {$=$} (rr2);
\end{tikzpicture}
  \caption{Example~\ref{exa:hfdsol}: a way to avoid the $\fd$ violation}
  \label{fig:hfds2}
\end{subfigure}
\caption{Example of a higher-arity FD violation in our process, and the proposed solution}
\label{fig:hfdsall}
\end{figure}

We fix the finitely-closed constraints $\con \defeq \fds
\fragcup \idsb$, consisting of arbitrary-arity $\fd$s $\fds$ and $\uid$s $\idsb$. We
denote by $\ufds$ the \emph{unary} $\fd$s among $\fds$, and write $\uconb \defeq
\ufds \fragcup \idsb$.
From the
definition of the finite closure (Section~\ref{sec:implication}), it is clear
that $\uconb$ is finitely closed as well, so the construction of the previous
\secnames applies to $\uconb$.

The problem to address in this \secname is that our completion process to satisfy $\idsb$ was defined with fact-thrifty
chase steps. These chase steps may reuse elements from the same facts at the same positions
multiple times. This may violate $\fds$, and it is in fact the only point
where we do so in the construction.

\begin{mainexample}
  \label{exa:hfdviol}
  For simplicity, we work with instances rather than aligned superinstances.
  Consider $I_0 \defeq \{S(a), T(z)\}$, the $\uid$s $\tau: \ui{S^1}{R^1}$ and $\tau':
  \ui{T^1}{R^1}$ for a $3$-ary relation $R$,
  and the $\fd$ $\phi: R^{2,3} \rightarrow R^1$.
  Consider $I \defeq I_0 \sqcup \{R(a, b, c)\}$ obtained by one chase step of~$\tau$
  on $S(a)$. Figure~\ref{fig:hfds} represents $I$ in solid black, using edges to
  highlight equalities between elements.

  We can perform a fact-thrifty chase step
  of $\tau'$ on $z$ to create $R(z, b, c)$, reusing $(b, c)$ at $\nondanger(R^1)
  = \{R^2, R^3\}$; this is illustrated in dashed red in Figure~\ref{fig:hfds}.
  However, the two $R$-facts would then be a violation of~$\phi$, as shown by the patterns
  of equalities and inequalities illustrated as thick red edges.
\end{mainexample}

The goal of this \secname is to define a new version of thrifty chase steps that
preserves $\fds$ rather than just $\ufds$; we call them \defo{envelope-thrifty
chase steps}. We first describe the new saturation process designed for
them, which is much more complex because we need to saturate \emph{sufficiently}
with respect to the completion process that we do next. To saturate, we use a separate combinatorial result, of possible independent
interest: Theorem~\ref{thm:combinatorial}, proved in Section~\ref{sec:combinatorial}. Second, we redefine the completion
process of the previous \secname for this new notion of chase step, and use this new completion process
to prove Theorem~\ref{thm:hfds}.

\subsection{Envelopes and Envelope-Saturation}

We start by defining our new notion of saturated instances.
Recall the notions of fact classes (Definition~\ref{def:factclass}) and thrifty 
chase steps (Definition~\ref{def:thrifty}).
When a fact-thrifty chase step creates a fact $F_{\fnew}$ whose chase witness
$F_\w$ has
fact class $(R^p, \mybf{C})$, we need elements to reuse in $F_{\fnew}$ at positions of
$\nondanger(R^p)$, which
need to already occur at the positions where we reuse them. Further,
the reused elements must have $\cov$-images of the right class.

Fact-thrifty chase steps reuse a tuple of elements from one fact $F_\rr$,
and thus apply to \emph{fact-saturated instances}, where each fact
class $D$ which is achieved in the chase is also achieved by some fact (recall
Definitions~\ref{def:factclass} and~\ref{def:factsat}).
Our new notion
of envelope-thrifty chase steps will consider \emph{multiple} tuples that
achieve each class $D$, that we call an \defo{envelope} for~$D$; with the
difference, however, that not all \emph{tuples} need to actually occur in an achiever
fact in the instance, though each \emph{individual} element needs to occur in some
achiever fact. Formally:

\begin{definition}
  \label{def:envelope}
  Consider $D = (R^p, \mybf{C})$ in $\afactcl$, and write $O \defeq
  \nondanger(R^p)$. Remember that $O$ is then non-empty.
  An \deft{envelope}~$E$ for~$D$ and for an aligned superinstance $J = (I, \cov)$ of~$I_0$
  is a non-empty set of $\card{O}$-tuples indexed by $O$, with domain
  $\dom(I)$, such that:
  \begin{enumerate}
  \item for every $\fd$ $\phi : R^L \rightarrow R^r$ of $\fds$ with $R^L \subseteq O$
    and $R^r \in O$,
    for any $\mybf{t}, \mybf{t}' \in E$, $\pi_{R^L}(\mybf{t}) =
    \pi_{R^L}(\mybf{t}')$ implies $t_r = t'_r$;
  \item for every $\fd$ $\phi : R^L \rightarrow R^r$ of $\fds$ with $R^L
    \subseteq O$ and $R^r \notin O$, for all $\mybf{t}, \mybf{t}' \in E$,
    $\pi_{R^L}(\mybf{t}) = \pi_{R^L}(\mybf{t}')$ implies $\mybf{t} = \mybf{t}'$;
    \item for every $a \in \dom(E)$, there is exactly one position $R^q \in
      O$ such
      that $a \in \pi_{R^q}(E)$, and then we also have $a \in \pi_{R^q}(J)$;
    \item for any fact $F = R(\mybf{a})$ of~$J$ and $R^q \in
      O$, if
      $a_q \in \pi_{R^q}(E)$, then $F$ achieves $D$ in~$J$
      and $\pi_{O}(\mybf{a}) \in E$. \qedef
  \end{enumerate}
\end{definition}

Intuitively, the tuples in the envelope~$E$ satisfy
the $\fd$s of~$\fds$ within~$\nondanger(R^p)$ (condition~1), and never overlap on positions
that determine a position out of $\nondanger(R^p)$ (condition~2). Further, their elements
already occur at the position where they will be reused, and we require for
simplicity that there is exactly one such position (condition~3).
Last, the elements have the
right $\cov$-image for the fact class~$D$, and for simplicity, whenever a fact reuses an
envelope element, we require that it reuses a whole envelope tuple (condition~4).

We then extend this definition across all achieved fact 
classes in the natural way:

\begin{definition}
  A \deft{global envelope} $\calE$ for an aligned superinstance $J = (I, \cov)$
  of~$I_0$ is a mapping from each $D \in \afactcl$
  to an envelope $\calE(D)$ for~$D$ and~$J$, such that the envelopes have
  pairwise disjoint domains.
\end{definition}

It is not difficult to see  that an aligned superinstance with a global envelope must
be fact-saturated, as for each $D \in \afactcl$, the envelope $\calE(D)$ is a non-empty
set of non-empty tuples, and any element of this tuple must occur in a fact that
achieves~$D$, by conditions~3 and~4.
However, the point of envelopes is that they can
contain more than a single tuple, so we have multiple choices of elements to
reuse.

For some fact classes $(R^p, \mybf{C})$ it is not useful for envelopes to contain more
than one tuple. This is the case if the position $R^p$
is \defo{safe}, meaning that no $\fd$ from positions in $O \defeq \nondanger(R^p)$
determines a position outside of~$O$. (Notice that by definition of
$\nondanger(R^p)$, such an $\fd$ could never be a $\ufd$.) Formally:

\begin{definition}  
  We call $R^p \in \pos(\sigma)$ \deft{safe} 
  if there is no $\fd$ $R^L
  \rightarrow R^r$ in~$\fds$ with $R^L \subseteq \nondanger(R^p)$ and $R^r
  \notin \nondanger(R^p)$. Otherwise, $R^p$ is \deft{unsafe}.
  
  We accordingly
  call a fact class $(R^p, \mybf{C}) \in \afactcl$ \deft{safe} or \deft{unsafe}
  depending on~$R^p$.
  Observe that the second condition of Definition~\ref{def:envelope} is trivial
  for envelopes on safe fact classes.
\end{definition}

It is not hard to see that when we apply a fact-thrifty (or even
relation-thrifty) chase step, and the exported position of the new fact is safe, 
then the problem illustrated by
Example~\ref{exa:hfdviol} cannot arise. In fact, one could show that
fact-thrifty or relation-thrifty chase steps cannot introduce $\fd$
violations in this case. Because of this, in envelopes
for \emph{safe} fact classes, we do not need more than one tuple, which we can reuse as
we did with fact-thrifty chase steps.

For unsafe fact classes, however, it will
be important to have more tuples, and to \emph{never reuse the same tuple
twice}. This motivates our definition of the \defo{remaining tuples} of an
envelope, depending on whether the fact class is safe or not; and the definition
of \defo{envelope-saturation}, which depends on the number of remaining tuples:

\begin{definition}
  Letting $E$ be an envelope for~$(R^p, \mybf{C}) \in \afactcl$ and $J$ be an aligned
  superinstance, the \deft{remaining tuples} of~$E$ are $E \backslash \pi_{\nondanger(R^p)}(J)$ if
  $(R^p, \mybf{C})$ is unsafe, and just $E$ if it is safe.

  We call $J$ \deft{$\bm{n}$-envelope-saturated} if it has a global envelope~$\calE$
  such that $\calE(D)$ has $\geq n$ remaining tuples for all unsafe $D \in
  \afactcl$. $J$ is
  \deft{envelope-saturated} if it is $n$-envelope-saturated for some $n > 0$.
\end{definition}

In the rest of the \subsecname, inspired 
by the Fact-Saturated Solutions Lemma (Lemma~\ref{lem:nondetexhaust}), we will
show that we can construct envelope-saturated solutions. However, there are some
complications when doing so. First, we must show that we can construct
\emph{sufficiently} envelope-saturated solutions, i.e., instances with
sufficiently many remaining tuples. To do this, we will need multiple
copies of the chase, which explains the technical switch from $I_0$ to $I_0'$ in
the statement of the next result. Second, 
for reasons that will become clear later in this \secname, we need to ensure that
the envelopes are large \emph{relative to the resulting instance size}. This
makes the result substantially harder to show.

\begin{proposition}[Sufficiently envelope-saturated solutions]
  \label{prp:preproc}
  For any $K \in \mathbb{N}$ and instance~$I_0$, we can construct an instance $I_0'$
  formed of disjoint copies of~$I_0$, and an aligned superinstance $J$ of
  $I_0'$ that satisfies~$\fds$ and is $(K\cdot \card{J})$-envelope-saturated.
\end{proposition}

We prove the proposition in the rest of the \subsecname. It is not hard to
see that $I_0'$ and $J$ can be constructed separately for each fact class in
$\afactcl$, and that this is difficult only for unsafe classes. In other words,
the crux of the matter is to prove the following:

\begin{lemma}[Single envelope]
  \label{lem:oneenv}
  For any unsafe class $D$ in $\afactcl$, instance $I_0$ and constant factor $K \in \mathbb{N}$, 
  there exists $N_0 \in \NN$ such that,
  for any $N \geq N_0$,
  we can construct an instance $I_0'$ formed of disjoint copies of~$I_0$,
  and an aligned superinstance
  $J = (I, \cov)$ of~$I_0'$ that satisfies $\fds$, with an envelope $E$ for $D$ of
  size $\geq K \cdot N$, such that $\card{J} \leq N$.
\end{lemma}

Indeed, let us prove Proposition~\ref{prp:preproc} with this lemma, and we will
prove the lemma afterwards:

\begin{proof}[\proofof Proposition~\ref{prp:preproc}]
  Fix the constant $K \in \mathbb{N}$ and the initial instance
$I_0$, and let us build $I_0'$ and the aligned superinstance $J = (I,
\cov)$ of~$I_0'$ that has a global envelope $\calE$.
As $\afactcl$ is finite, we build one $J_D$ per $D \in \afactcl$ with an
envelope~$E_D$ for the class $D$,
and we will define $J \defeq \bigsqcup_{D \in \afactcl} J_D$
and define $\calE$ by $\calE(D) \defeq E_D$ for all $D \in \afactcl$.
When $D = (R^p, \mybf{C})$ is safe, we proceed as in the
proof of the Fact-Saturated Solutions Lemma (Lemma~\ref{lem:nondetexhaust}): take a
single copy $J_D$ of the truncated chase to achieve the class $D$, and take as the
only fact of the envelope $E_D$ the projection to $\nondanger(R^p)$ of an
achiever of~$D$ in~$J_D$.
  When $D$ is unsafe, we use the Single Envelope Lemma (Lemma~\ref{lem:oneenv}) to obtain $J_D$ and the
envelope $E_D$. As $\afactcl$ is finite and its size does not depend on~$I_0$,
we can ensure
that that $\card{E_D} \geq (K+1) \cdot \card{J}$ for all unsafe $D
\in \afactcl$
by using the Single Envelope Lemma with $K' \defeq (K+1) \cdot \card{\afactcl}$, and
taking $N \in \NN$ which is larger than the largest $N_0$ of that lemma across
all $D \in \afactcl$. Indeed, the resulting model 
$J$ then ensures that $\card{J} \leq \card{\afactcl} \cdot N$ and 
$\card{E_D} \geq (K+1) \cdot \card{\afactcl} \cdot N$.

We now check that the resulting $J$ and $E$ satisfy the conditions.
Each $J_D$ is an aligned superinstance of an instance $(I_0')_D$ which is formed of
disjoint copies of~$I_0$ (for unsafe classes) or which is exactly $I_0$ (for
safe classes), so 
$J$ is an aligned superinstance of $I_0' \defeq \bigsqcup_{D\in \afactcl}
(I_0')_D$,
so $I_0'$ is also a union of disjoint copies of~$I_0$.
There are no violations of~$\fds$ in $J$ because there are none in
any of the~$J_D$. The disjointness of domains of
envelopes in the global envelope $\calE$ is because the $J_D$ are disjoint. It
is easy to see that $J$ is $(K\cdot
\card{I})$-envelope-saturated, because $\card{\calE(D)} \geq (K+1)\cdot\card{I}$ for
all unsafe $D \in \afactcl$, so the
number of remaining facts of each envelope for an unsafe class is $\geq K\cdot
\card{I}$ because every fact of $I$ eliminates at most one fact in each envelope.
Hence, the proposition is proven.
\end{proof}

So the only thing left to do is to prove 
the Single Envelope Lemma (Lemma~\ref{lem:oneenv}).
Let us accordingly fix the unsafe class $D = (R^p, \mybf{C})$ in $\afactcl$.
We will need to study more precisely the $\fd$s implied by the definition of an
envelope for~$D$
(Definition~\ref{def:envelope}). We first introduce notation for them:

\begin{definition}
  \label{def:fdproj}
  Given a set $\fds$ of $\fd$s on a relation~$R$ and $O \subseteq \pos(R)$,
  the \deft{$\fd$ projection} $\fdrestr{\fds}{O}$ of 
  $\fds$ to~$O$ consists of the following $\fd$s, which we close under
  implication:
  \begin{enumerate}
    \item the $\fd$s $R^L \rightarrow R^r$ of~$\fds$ such
  that $R^L \subseteq O$ and $R^r \in O$;
    \item for every $\fd$ $R^L \rightarrow
  R^r$ of~$\fds$ where $R^L \subseteq O$ and $R^r \notin O$, the key dependency
  $R^L \rightarrow O$.\qedef
  \end{enumerate}
\end{definition}

We will need to show that, as $R^p$ is unsafe,
$\fdrestr{\fds}{O}$ cannot have a \emph{unary key} in $O$, namely, there cannot
be $R^q \in O$
such that, for every $R^r \in O$, either $R^q = R^r$ or the $\ufd$ $R^q
\rightarrow R^r$ is in $\fdrestr{\fds}{O}$. We show the contrapositive of
this statement:

\begin{lemma}
  \label{lem:connectndg}
  For any $R^p \in \pos(\sigma)$, letting $O \defeq \nondanger(R^p)$, if
  $O$ has a unary key in $\fdrestr{\fds}{O}$, then $R^p$ is safe.
\end{lemma}

\begin{proof}
  Fix $R^p \in \pos(\sigma)$ and let $O \defeq \nondanger(R^p)$.
  We first show that if $O$ has a unary key $R^s \in O$ in
  the original $\fd$s $\fds$, then
  $R^p$ is safe. Indeed, assume the existence of such an $R^s \in O$.
  Assume by way of contradiction that $R^p$ is not safe, so there is an
  $\fd$ $R^L \rightarrow R^r$ in $\fds$
  with $R^L \subseteq O$ and $R^r \notin O$. Then, as $\fds$ is closed under the
  transitivity rule, the $\ufd$
  $\phi: R^s \rightarrow R^r$ is in $\ufds$. Now, as $R^r \notin O$, either
  $R^r = R^p$ or $R^r \in \danger(R^p)$; in both cases, $\phi$ witnesses that
  $R^s \in \danger(R^p)$, but we had $R^s \in O$, a contradiction.

  We must now show that if $O$ has a unary key in $O$ according to
  $\fdrestr{\fds}{O}$,
  then $O$ has a unary key in $O$ according to $\fds$. It suffices to show that
  for any two positions $R^q, R^s \in O$, if the $\ufd$ $\phi': R^q \rightarrow
  R^s$ is
  in~$\fdrestr{\fds}{O}$ then it also does in~$\fds$. Hence, fix $R^q$ in
  $O$, and consider the set $S$ of positions in $O$ that $R^q$ determines according to~$\fdrestr{\fds}{O}$.
  Let $\Phi$ be the $\fd$s in the list given in Definition~\ref{def:fdproj},
  so that $\fdrestr{\fds}{O}$ is the result of closing
  $\Phi$  under $\fd$ implication.
  We can compute $S$ using the well-known
  ``fd closure algorithm'' \cite[Algorithm~8.2.7]{abiteboul1995foundations}, which
  starts at $S = \{R^q\}$ and iterates the following operation: whenever
  there is $\phi: R^L \rightarrow R^r$ such that $R^L \subseteq S$, add $R^r$ to
  $S$.

  Assume now that there is a position $R^s$ in $S$ such that
  $\phi': R^q \rightarrow R^s$ is not in $\fds$. This implies that, when
  computing $S$, we must have used some $\fd$ $R^L \rightarrow R^t$ from a key
  dependency $\kappa$ in $\Phi$,
  as they are the only $\fd$s of $\Phi$ which are not in $\fds$. The first time
  we did this, we had derived, using only $\fd$s from $\fds$, that $R^L
  \subseteq S$, so that the key dependency $R^q \rightarrow R^L$ is in~$\fds$.
  Now, $\kappa$ witnesses that there is an $\fd$ $R^L \rightarrow R^r$ in $\fds$
  with $R^r \notin O$, so that, as $\fds$ is closed under implication, we deduce
  that $R^q \rightarrow R^r$ is in $\fds$ with $R^q$ in $O$ and $R^r \notin
  O$. As before, this contradicts the definition of $O \defeq \nondanger(R^p)$.
  So indeed, there is no such $R^s$ in
  $S$.

  Hence, if $O$ has a unary key in $O$ according to $\fdrestr{\fds}{O}$, then it
  also does according to $\fds$, and then, by the reasoning of the first
  paragraph, $R^p$ is safe, which is the desired claim.
\end{proof}

We now know that $O$ has no unary key in $\fdrestr{\fds}{O}$.
This allows us to introduce the crucial tool needed to prove the
Sufficiently Envelope-Saturated Solutions Proposition (Proposition~\ref{prp:preproc}).
It is the following independent result, which is proved separately in
Section~\ref{sec:combinatorial} using a combinatorial construction.

\mydefthm{combinatorial}{theorem}{Dense interpretations}{
  For any set $\fds$ of~$\fd$s over a relation~$R$ with no unary key, for all $K \in
  \mathbb{N}$, there exists $N_0 \in \NN$ such that
  for all $N \geq N_0$,
  we can construct a non-empty instance $I$ of~$R$ that satisfies $\fds$
  and such that $\card{\dom(I)} \leq N$ and $\card{I} \geq K \cdot N$.

  Further, we can impose a \textbf{disjointness condition}
  on the result~$I$: we can ensure that for all $a \in \dom(I)$, there exists exactly one
  $R^p \in \positions(R)$ such that $a \in \pi_{R^p}(I)$.
}

We can now prove the Single Envelope Lemma (Lemma~\ref{lem:oneenv}) and conclude the \subsecname.
Choose a fact $F_{\ach} = R(\mybf{b})$ of~$\chase{I_0}{\idsb} \backslash I_0$
that achieves the fact class $D$, and 
let $I_1$ be obtained from $I_0$ by
applying $\uid$ chase steps on $I_0$
to obtain a finite truncation of~$\chase{I_0}{\idsb}$ that
includes $F_{\ach}$ but no child fact of~$F_{\ach}$. Consider the aligned superinstance
$J_1 = (I_1, \cov_1)$ of $I_0$, where $\cov_1$ is the identity.

Remember that we wrote $D = (R^p, \mybf{C})$, and $O = \nondanger(R^p)$, which
is non-empty.
Define a $\card{O}$-ary relation
$\relrestr{R}{O}$ (with positions indexed by $O$ for convenience), 
define $\fdrestr{\fds}{O}$ as in Definition~\ref{def:fdproj}, 
and consider $\fdrestr{\fds}{O}$ 
as $\fd$s on $\relrestr{R}{O}$.
Because $D$ is unsafe, by
Lemma~\ref{lem:connectndg}, $\relrestr{R}{O}$ has no unary key in
$\fdrestr{\fds}{O}$.
Letting $K \in \NN$ be our target constant for the Single Envelope Lemma,
apply the Dense Interpretations Theorem
(Theorem~\ref{thm:combinatorial}) to $\relrestr{R}{O}$ and $\fdrestr{\fds}{O}$,
taking $K' \defeq 2 \cdot K \cdot \card{J_1}$ as the constant.
Define $N_0 \in \NN$ for the Single Envelope Lemma as $2 \cdot \max(\card{J_1}, 1) \cdot
\max(N_0', 1)$ where
$N_0'$ is obtained from
the Dense Interpretations Theorem for~$K'$.
Letting $N' \in \NN$ be our target size for the Single Envelope Lemma,
using $N \defeq \lfloor N' / \card{J_1} \rfloor$
as the target size for the Dense Interpretations Theorem (which is $\geq N_0'$),
we can build an instance $I_{\dense}$ of~$\relrestr{R}{O}$ that
satisfies~$\fdrestr{\fds}{O}$ and such that
$\card{I_{\dense}} \geq N \cdot K'$
and $\card{\dom(I_{\dense})} \leq N$.

Let $I'_{\dense} \subseteq I_{\dense}$ be a
subinstance of size exactly $N$ of~$I_{\dense}$ such that we have $\dom(I'_{\dense}) =
\dom(I_{\dense})$, that is,
such that each element of~$\dom(I_{\dense})$ occurs in some fact of~$I'_{\dense}$:
we can clearly construct $I'_{\dense}$
by picking, for each element of $\dom(I_{\dense})$, one fact of~$I_{\dense}$ where it occurs, removing duplicate
facts, and completing with other arbitrary facts of~$I_{\dense}$ so the number of facts
is exactly $N$. Number the facts of $I'_{\dense}$ as $F'_1, \ldots, F'_N$.

Let us now create $N-1$ disjoint 
copies of~$J_1$, numbered $J_2$ to $J_N$.
Let $I_{\pre}$ be the disjoint union of the underlying instances of the~$J_i$,
let $I_0'$ be formed of the $N$ disjoint copies of $I_0$ in~$I_{\pre}$,
and define a mapping $\cov_{\pre}$ from $\dom(I_{\pre})$ to
$\chase{I_0'}{\idsb}$ following the $\cov_i$ in the expected way.
It is clear that $J_{\pre}$ is an aligned superinstance of~$I_0'$.
For $1 \leq i \leq N$, we call $F_i = R(\mybf{a}^{\mybf{i}})$ \emph{the fact of~$I_i$ that
corresponds to the achiever $F_{\ach}$ in $\chase{I_0}{\idsb}$}. In particular, for all
$1 \leq i \leq N$, we have
that $\cov(a^i_j) = b_j$ for all $j$, and $a^i_p$ is the only element of~$F_i$
that also occurs in other facts of~$J_i$, as $J_i$ does not contain any
descendent fact of~$F_i$.

Intuitively, we will now identify elements in $J_{\pre}$ so that the restriction of
the~$F_i$ to~$O$ is exactly the $F'_i$, and this will allow us to use the
instance
$I_{\dense}$ to define the envelope. Formally,
as the $a^i_j$ are
pairwise distinct,
we can define the function $f$ that maps each $a^i_j$, for $1 \leq i \leq N$ and
$R^j \in O$, to $\pi_{R^j}(F_i')$. In other words, $f$ is a surjective (but
generally not injective) mapping, the domain of~$f$ is the 
projection to~$O$ of the~$F_i$ in~$I_i$, the range of~$f$ is $\dom(I_{\dense}')$,
and $f$ maps each element of the
projection to the
corresponding element in~$F_i'$.
Extend $f$ to a mapping $f'$ with domain $\dom(I_\pre)$
by setting $f'(a) \defeq f(a)$ when $a$ is in the domain of~$f$, and
$f'(a) \defeq a$ otherwise. Now, let $I \defeq f'(I_{\pre})$. In other words, $I$ is
$I_{\pre}$ except that
elements in the projection to~$O$ of the facts~$F_i$ are renamed, and some are
identified, 
so that the projection to~$O$ of 
$\{f'(F_i) \mid 1 \leq i \leq N\}$,
seen as an instance of~$\restr{R}{O}$-facts,
is exactly 
$I'_\dense$. 
Because $a^i_j$ occurs only in $F_i$ for all $R^j \neq R^p$,
and $R^p \notin O$, this means that the elements identified by $f'$ only occurred in
the $F_i$ in~$I_{\pre}$.

We now build $J = (I, \cov)$ obtained by defining $\cov$ from $\cov_{\pre}$ as
follows: if $a$ is in the domain of~$f$, then $\cov(a) \defeq
\cov_{\pre}(a')$ for any preimage of $a'$ by $f'$ (as we will see, the
choice of preimage does not matter), and
if $a$ is not in the domain of~$f$, then $\cov(a) \defeq \cov_{\pre}(a)$ because
$a$ is then the only preimage of~$a$ by~$f'$.
We have now defined the instance $I_0'$ formed of disjoint copies of $I_0$ and
the final $J$, and we define $E \defeq I_{\dense}$. We must now show that $J$ is indeed an
aligned superinstance of~$I_0'$, and that $E$ is an envelope for $I$ and $D$, and that
they
satisfy the required conditions.

\medskip

We note that it is immediate that $J = (I, \cov)$ is a superinstance of~$I_0'$. Indeed,
we have $I \defeq f(I_{\pre})$, and $I_{\pre}$ was a superinstance of~$I_0'$, so
it suffices to note that $\dom(I_0')$ is not in the domain of $f$: this is
because the achiever $F_{\ach}$ is not a fact of~$I_0$, so the domain of $f$,
namely, the projection of the $F_i$ on
$O$, does not intersect $\dom(I_0')$.
Further, it is clear that $J$ is finite and has $N \cdot \card{J_1}$ facts, because this is
the case of $J_{\pre}$ by definition, and $f'$ cannot have caused any facts of
$J_{\pre}$ to be identified in $J$, because we have $R^p \notin
O$, so the projection of each $F_i$ to $R^p$ is a different element which is
mapped to itself by~$f'$. Hence, we have $\card{J} = N \cdot \card{J_1} \leq
N'$.
Further, we have $\card{E} = \card{I_{\dense}} \geq N \cdot K' \geq \lfloor N' /
\card{J_1} \rfloor \cdot 2 \cdot K \cdot \card{J_1}$, and as $N' \geq N_0 \geq
2 \cdot \card{J_1}$ we
have $\lfloor N' / \card{J_1} \rfloor \geq (1/2) \cdot (N' / \card{J_1})$.
Hence, $\card{E} \geq K \cdot N'$,
so we have achieved the required size bound.

We now show that $J$ is indeed an aligned superinstance of~$I_0'$. The technical
conditions on $\cov$ are clearly respected, because they were respected on
$J_{\pre}$, because $f'$ only identifies elements in $I_{\pre} \backslash I_0$,
and because the
identified elements occur at the same positions as their preimages so the
directionality condition is respected.

We show that $\cov$ is a $k$-bounded simulation from $J$ to
$\chase{I_0'}{\idsb}$ by showing the stronger claim that it is actually a
$k'$-bounded simulation for all $k' \in \NN$, which we show by induction
on~$k'$.
The case of~$k'=0$ is trivial. The
induction case is trivial for all facts except for the $f'(F_i)$,
because the $a^i_j$ only occurred in~$I_{\pre}$ in the facts $F_i$, by our assumption
that the $F_i$ have no children in the $I_i$, and because the exported
position of~$F_{\ach}$ is $R^p \notin O$. Hence, consider a fact $F' =
R(\mybf{c})$ of~$I$ which is the image by $f'$ of some fact~$F_i$.
Choose $1 \leq p \leq \arity{R}$. We wish to show that there exists a fact $F'' =
R(\mybf{d})$ of
$\chase{I_0'}{\idsb}$ such that $\cov(c_p) = d_p$ and for all $1 \leq
q \leq \card{R}$ we have $(I, c_q) \bsim_{k'-1} (\chase{I_0'}{\idsb}, d_q)$.
Let $a^{i_0}_{j_0}$ be the preimage of~$c_p$ used to define $\cov(c_p)$; by the
disjointness condition of 
the Dense Interpretations Theorem (Theorem~\ref{thm:combinatorial}),
we must have $j_0 = p$.
Observe that $\chase{I_0'}{\idsb}$ is formed of disjoint copies of
$\chase{I_0}{\idsb}$, so, recalling the definition of $J'_{i_0}$, 
consider the fact
$F'' = R(\mybf{d})$ of~$\chase{I_0'}{\idsb}$ corresponding to $F_{i_0}$ in
$I$. By definition, $\cov(c_p) = \cov(a^{i_0}_{j_0}) = d_p$.

We now show that for all $1 \leq
q \leq \card{R}$ we have $(I, c_q) \bsim_{k'-1} (\chase{I_0'}{\idsb}, d_q)$.
Fix  $1 \leq q
\leq \arity{R}$.
It suffices to show that $\cov(c_q) \bbsim_{k'} d_q$, as we can 
then use the induction hypothesis to know that $(I, c_q)
\bsim_{k'-1} (\chase{I_0'}{\idsb}, \cov(c_q))$, so that by transitivity $(I,
c_q) \bsim_{k'-1} (\chase{I_0'}{\idsb}, d_q)$.
Hence, we show that $\cov(c_q) \bbsim_{k'} d_q$.
Let $a^{i_0'}_{j_0'}$ be
the preimage of $c_q$ used
to define $\cov(c_q)$.
Again we must have $j_0'
= q$ by the disjointness condition, and, considering the fact $F''' = R(\mybf{e})$ of
$\chase{I_0'}{\idsb}$ corresponding to $F_{i_0'}$ in $I$, we have 
$\cov(c_q) = e_q$. But as both $F'''$
and $F''$ are copies in $\chase{I_0'}{\idsb}$ of the same fact $F_{\ach}$ of
$\chase{I_0}{\idsb}$, it is indeed the case that
$d_q \bbsim_{k'} e_q$. Hence, $\cov(c_q) \bbsim_{k'} d_q$, from which we
conclude that $F''$ is a suitable witness fact for~$F'$.
By induction, we have shown that $\cov$ is indeed a $k'$-bounded simulation from~$J$
to $\chase{I_0'}{\idsb}$ for any $k' \in \NN$, so that it is in particular a
$k$-bounded simulation.

\medskip

We now show that $J$ satisfies $\fds$. For this, it will be convenient to define
the \defo{overlap} of two facts:

\begin{definition}
  \label{def:overlap}
  The \deft{overlap} $\ovl(F, F')$ between two facts $F = R(\mybf{a})$ and $F'
  = R(\mybf{b})$ of the same relation~$R$ in an instance~$I$ is the subset $O$ of~$\pos(R)$ such that $a_s = b_s$ iff $R^s
  \in O$. If $\card{O} > 0$, we say that $F$ and $F'$ \deft{overlap}.
\end{definition}

As $I_{\pre}$ satisfies $\fds$ by the Unique Witness Property of the $\uid$
chase,
any new violation of~$\fds$ in $I$ relative to $I_\pre$
must include some fact $F = f'(F'_{i_0})$, and some fact $F' \neq F$ that
overlaps with~$F$,
so necessarily $F' = f'(F'_{i_1})$ for some $i_1$ by construction of~$I$,
and $\ovl(F, F') \subseteq O$. If $\ovl(F, F') = O$, then, by our
definition of $f$ and of the $F'_i$, this implies that $F'_{i_0} = F'_{i_1}$, a
contradiction because $F \neq F'$. So the only case to consider is when $\ovl(F,
F') \subsetneq O$, but we can also exclude this case:

\begin{lemma}
  \label{lem:liftovl}
  Let $I$ be an instance, $\fds$ be a conjunction of $\fd$s, and $F \neq F'$ be two 
  facts of~$I$. Assume there is a position $R^p \in \pos(\sigma)$ such that,
  writing $O \defeq \nondanger(R^p)$, we have $\ovl(F, F')
  \subsetneq O$, and
  that $\{\pi_O(F), \pi_O(F')\}$  is not a violation of~$\fdrestr{\fds}{O}$.
  Then $\{F, F'\}$ is not a violation of~$\fds$.
\end{lemma}

\begin{proof}
  Assume by way of contradiction that $F$ and $F'$ violate an $\fd$ $\phi: R^L
  \rightarrow R^r$ of~$\fds$, which implies that
  $R^L \subseteq \ovl(F, F') \subseteq O$ and $R^r \notin \ovl(F, F')$. Now, if $R^r
  \in O$, then $\phi$ is in $\fdrestr{\fds}{O}$, so that $\pi_O(F)$ and
  $\pi_O(F')$ violate~$\fdrestr{\fds}{O}$, a contradiction. Hence, $R^r \in \pos(R) \backslash O$, and
  the key dependency $\kappa: R^L \rightarrow O$ is in $\fdrestr{\fds}{O}$, so that
  $\pi_O(F)$ and $\pi_O(F')$ must satisfy~$\kappa$.
  Thus, because $R^L \subseteq \ovl(F, F')$, we must have $\ovl(F, F') = O$,
  which is a contradiction because we assumed $\ovl(F, F') \subsetneq O$.
\end{proof}

Now,  by definition of~$I'_{\dense}$, we know that $I'_{\dense}$ satisfies
$\fdrestr{\fds}{O}$, so that $\{\pi_O(F), \pi_O(F')\}$ is not a violation of
$\fdrestr{\fds}{O}$. Thus, we can conclude with Lemma~\ref{lem:liftovl}
that $\{F, F'\}$ is not a violation of~$\fds$, so that $J$ satisfies $\fds$.
We have thus shown that $J$ is an aligned superinstance of~$I_0'$.

\medskip

Last, we check that $E$ is indeed an envelope for~$D$ and for~$J$.
Indeed, $E$ satisfies~$\fdrestr{\fds}{O}$ by construction, so conditions~1 and~2 are
respected. The first part of condition~3 is ensured by the disjointness condition,
and its second part follows from our definition of $I'_{\dense}$ that ensures
that any element in $\dom(E)$ occurs in a fact $F'_i$ of $I'_{\dense}$, hence occurs in
$f'(F_i)$ in~$J$.
Last, condition~4 is true because the elements of $\dom(E)$ are only used in the
$f'(F_i)$,
and the $\cov$-images of the $f'(F_i)$
are copies in $\chase{I_0'}{\idsb}$ of the same fact $F_{\ach}$ in
$\chase{I_0}{\idsb}$ that achieves~$D$, so the $F_i$ are all achievers of~$D$;
further, by definition, their projection to~$O$ is a tuple of $E$ because it is a
fact of~$I'_{\dense}$.

Hence, $J$ is indeed an aligned superinstance of a disjoint union $I_0'$ of
copies of~$I_0$, $J$ satisfies $\fds$, $\card{J} \leq N'$,
and $J$ has an envelope $E$ of size $K \cdot
N'$ for~$D$. This concludes the proof of the Single Envelope Lemma (Lemma~\ref{lem:oneenv}), and hence of
the Sufficiently Envelope-Thrifty Solutions Proposition (Proposition~\ref{prp:preproc}).

\subsection{Envelope-Thrifty Chase Steps}

We have shown that we can construct sufficiently envelope-saturated
superinstances of the input instance. The point of this notion is to
introduce \defo{envelope-thrifty chase steps}, namely, thrifty chase steps that
use remaining tuples from the envelope to fill the non-dangerous positions:

\begin{definition}
  \deft{Envelope-thrifty chase steps} are thrifty chase steps
  (Definition~\ref{def:thrifty}) which apply to
  envelope-saturated aligned superinstances.
  Following Definitions~\ref{def:thrifty} and~\ref{def:factthrifty},
  we write $S^q$ for the exported position of the new fact~$F_{\fnew}$,
  we write $F_\w = S(\mybf{b}')$ for the chase witness,
  and we let $D = (S^q, \mybf{C}) \in \afactcl$ be the fact class of~$F_\w$.
  Analogously to Definition~\ref{def:thrifty}, we define an
  \deft{envelope-thrifty} chase step as follows: if $\nondanger(S^q)$ is
  non-empty, choose one remaining tuple $\mybf{t}$ of $\calE(D)$, and
  set $b_r \defeq t_r$ for all $S^r \in \nondanger(S^q)$.

  We define a \deft{fresh envelope-thrifty step} in the same way as a fresh
  fact-thrifty step: all elements at dangerous positions are fresh elements
  only occurring at that position.
\end{definition}

\begin{example}
  \label{exa:hfdsol}
  Recall $I_0$, $\tau$, $\tau'$ and $\phi$ from Example~\ref{exa:hfdviol}.
  Now, consider $I_0' \defeq \{S(a), T(z), S(a'), S(z')\}$ formed of two copies
  of $I_0$, and $I' \defeq I_0' \sqcup
  \{R(a, b, c), \allowbreak R(a', b', c')\}$ obtained by two chase steps: this
  is illustrated in solid black in Figure~\ref{fig:hfds2} on page~\pageref{fig:hfds2}.
  The two facts $R(a, b, c)$ and $R(a', b', c')$ would achieve
  the same fact class~$D$, so we can define
  $E(D) \defeq \{(b, c), \allowbreak (b', c'), \allowbreak (b', c), \allowbreak (b, c')\}$.
  
  We can now satisfy $\idsb$
  on~$I'$ without violating~$\phi$,
  with two envelope-thrifty chase steps that reuse the remaining tuples $(b', c)$ and
  $(b, c')$ of~$E(D)$: the new facts and the pattern of equalities between them is
  illustrated in red in Figure~\ref{fig:hfds2}.
\end{example}

Recall 
that fact-thrifty chase steps apply to
fact-saturated aligned superinstances (Lemma~\ref{lem:ftappl}).
Similarly, envelope-thrifty chase
steps apply to envelope-saturated aligned superinstances:

\begin{lemma}[Envelope-thrifty applicability]
  \label{lem:etappl}
  For any envelope-saturated superinstance $I$ of an instance $I_0$, $\uid$
  $\tau: \ui{R^p}{S^q}$ and element $a \in \appelem{I}{\tau}$, we can apply an
  envelope-thrifty chase step on $a$ with $\tau$ to satisfy this violation.

  Further, for any new fact $S(\mybf{e})$
  that we can create by chasing on~$a$ with $\tau$ with a fact-thrifty chase
  step, we can instead apply an envelope-thrifty chase step on $a$ with~$\tau$
  to create a fact
  $S(\mybf{b})$ with $b_r = e_r$ for all $S^r \in \pos(S) \backslash
  \nondanger(S^r)$.
\end{lemma}

\begin{proof}
  For the first part of the claim, as in the proof of 
  the Fact-Thrifty Applicability Lemma (Lemma~\ref{lem:ftappl}), 
  there is nothing to show unless $\nondanger(S^q)$ is non-empty,
  and 
  the fact
  class $D = (S^q, \mybf{C})$ is then in~$\afactcl$, where $\mybf{C}$ is the tuple of the
  $\bbsim_k$-equivalence classes of the elements of the chase witness $F_\w$.
  Hence, as $J$ is envelope-saturated, it has some
  remaining tuple for the class $D$ that we can use to define the non-dangerous
  positions of the new fact.

  For the second part, again as in the proof of the Fact-Thrifty Applicability Lemma,
  observe
  that the definition of envelope-thrifty chase steps only poses additional
  conditions (relative to thrifty chase steps) on $\nondanger(S^q)$,
  so that, for any fact that we would create with a fact-thrifty chase step, we
  can change the elements at $\nondanger(S^q)$ to perform an envelope-thrifty chase
  step, using the fact that $I$ is envelope-saturated.
\end{proof}

Further, recall that we showed that relation-thrifty chase steps never violate~$\ufds$
(Lemma~\ref{lem:rthriftyp}). We now show that envelope-thrifty chase steps never
violate $\fds$, which is their intended purpose:

\begin{lemma}[Envelope-thrifty $\fd$ preservation]
  \label{lem:envfds}
  For an $n$-envelope-saturated aligned superinstance $J$  satisfying
  $\fds$,
  the result of an envelope-thrifty chase step on~$J$ satisfies $\fds$.
\end{lemma}

\begin{proof}
  Fix $J$ and its global envelope $\calE$. Let $F_{\fnew} = S(\mybf{b})$ be the fact created by the envelope-thrifty
  step, let $\tau: \ui{R^p}{S^q}$ be the $\uid$, let $J' = (I', \cov')$ be the
  result of the chase step,
  let $F_\w$ be the chase witness,
  and let $D$ be the fact class of~$F_\w$.
   Write $O \defeq \nondanger(S^q)$.
  Assume by contradiction that $I' \not\models \fds$; as $I
  \models \fds$, any violation of~$\fds$ in~$I'$ must be
  between the new fact $F_{\fnew}$ and an existing fact $F = S(\mybf{c})$ of $I$.
  Recalling the definition of overlaps (Definition~\ref{def:overlap}),
  note that we only have $b_r \in \pi_{S^r}(I)$ for $S^r \in O$ by definition of thrifty
  chase steps, so we must have $\ovl(F_{\fnew}, F) \subseteq O$.
  Now, as $\pi_O(F_{\fnew})$ was defined using elements of
  $\dom(\calE(D))$, taking any $S^r \in \ovl(F_{\fnew}, F) \subseteq O$ (which is
  non-empty by definition of an $\fd$ violation), we have $c_r = b_r
  \in \pi_{S^r}(\calE(D))$, so that, by condition~4 of the definition of the
  envelope~$\calE(D)$,  we know that
  $\pi_O(\mybf{c})$ is a tuple $\mybf{t}'$ of~$\calE(D)$.
  Now, either
  $\ovl(F_{\fnew}, F) \subsetneq O$ or $\ovl(F_{\fnew}, F) = O$.

In the first case, we observe that, by conditions~1
and~2 of the definition of the envelope~$\calE(D)$,
we know that $\{\pi_O(\mybf{c}),
\pi_O(\mybf{b})\}$ is not a violation of $\fdrestr{\fds}{O}$. 
  Together with the fact that $\ovl(F_{\fnew}, F) \subsetneq O$,
  this allows us to apply Lemma~\ref{lem:liftovl} and deduce that $\{F,
  F_{\fnew}\}$ actually does not violate $\fds$, a contradiction.

  In the second case, where $\ovl(F_{\fnew}, F) = O$, we have $\mybf{t} =
  \mybf{t}'$. Now, either $D$ is safe or $D$ is unsafe.
If $D$ is unsafe, we have a contradiction because 
$F$ witnesses that $\mybf{t}$ was not a remaining tuple of $\calE(D)$,
  so we cannot have used $\mybf{t}$ to define $F_{\fnew}$. If $D$ is safe,
then by definition there is no $\fd$ $R^L \rightarrow R^r$ of~$\fds$ with $R^L \subseteq O$ and
  $R^r \notin O$. Now, as $\ovl(F_{\fnew}, F) = O$, it is clear that $F$ and
  $F_{\fnew}$ cannot violate any $\fd$
of~$\fds$, a contradiction again.
\end{proof}

Last, recall that we showed that fresh fact-thrifty steps preserve
the property of being
aligned (Lemma~\ref{lem:fftp}) and that non-fresh fact-thrifty steps also do when we
additionally assume $k$-essentiality, which they also preserve (Lemma~\ref{lem:usekrev}). We now
prove the analogous claims for envelope-thrifty steps assuming
envelope-saturation. The only difference is that envelope-thrifty chase steps
make envelope-saturation \emph{decrease}, unlike fact-thrifty steps which always
preserved fact-saturation:

\begin{lemma}[Envelope-thrifty preservation]
  \label{lem:envpres}
  For any $n \in \NN$,
  for any $n$-envelope-saturated aligned superinstance $J$ of~$I_0$, the result $J'$ of
  a fresh envelope-thrifty chase step on $J$ is an $(n-1)$-envelope-saturated aligned
  superinstance of~$I_0$. Further, if $J$ is $k$-essential, the claim holds even
  for non-fresh envelope-thrifty chase steps, and the result $J'$ is
  additionally $k$-essential.
\end{lemma}

\begin{proof}
We reuse notation from Lemma~\ref{lem:envfds}: considering an application of an envelope-thrifty chase step:
let $J = (I, \cov)$ be the aligned superinstance of~$I_0$, 
let $\tau : \ui{R^p}{S^q}$ be the $\uid$, 
write $O \defeq \nondanger(S^q)$,
let $F_\w = S(\mybf{b}')$ the chase witness,
let $D = (S^q, \mybf{C})$ be the fact class,
  let $F_{\fnew} = S(\mybf{b})$ be the new fact to be created, and
  let $\mybf{t}$ be the remaining tuple of~$\calE(D)$ used to define $F_{\fnew}$,
and let $J' = (I', \cov')$ be the result.

\medskip

We now prove that~$J'$ is still an aligned superinstance. 
We first adapt the Fresh Fact-Thrifty Preservation Lemma (Lemma~\ref{lem:fftp})
to work with envelope-thrifty chase steps. We can no longer use
Lemma~\ref{lem:rthriftyp} to prove that $J' \models \ufds$, but we have shown
already that $J' \models \fds$ in Lemma~\ref{lem:envfds}, so this point is
already covered. The only other point specific to
fact-thriftiness is proving that $\cov'$ is still a $k$-bounded simulation, but
it actually only relies on the fact that $\cov'(b_r) \bbsim_k b'_r$ in
$\chase{I_0}{\idsb}$ for all $S^r \in \nondanger(S^q)$, which is still ensured
by envelope-thrifty chase steps: by conditions~3 and~4 of the definition of
envelopes, we know that, for any $S^r \in \nondanger(S^q)$, the element $t_r$ already occurs
at position $S^r$ in a fact of~$I$ that achieves $D$, so that $\cov(t_r)
\bbsim_k b'_r$.

Second, we adapt the Fact-Thrifty Preservation Lemma (Lemma~\ref{lem:usekrev})
to envelope-thrifty chase steps. Again, the only condition of fact-thrifty chase steps
used when proving that lemma is that $\cov'(b_r) \bbsim_k b'_r$ in
$\chase{I_0}{\idsb}$ for all $S^r \in \nondanger(S^q)$, which is still true.
Hence, having adapted these two lemmas, we conclude that $J'$ has the required
properties.

\medskip

We now prove that $\calE$ is still a global envelope of~$J'$ after performing an
envelope-thrifty chase step. The condition on the disjointness of the envelope
domains only concerns $\calE$ itself, which is unchanged. Hence, we need only show
that, for any $D' \in \afactcl$,
$\calE(D')$ is still an envelope.
All conditions of the definition of envelopes except condition~4 are clearly
true, because they were true in~$J$, and they only depend on $\calE(D')$ or
they are preserved
when creating more facts.
  We now check condition~4, which only needs to be verified on the new fact
  $F_{\fnew}$.

Consider $S^u \in \pos(S)$ and $S^t \in \nondanger(S^u)$, and assume that $b_t
\in \pi_{S^t}(\calE(D'))$. As $\calE(D)$ is an envelope for~$J$, by condition~3
of the definition, we have
$b_t \in \pi_{S^t}(I)$ as well, so that, by definition of thrifty chase steps,
we must have $S^t \in O$.
Now, as the envelopes of~$\calE$
are pairwise disjoint, and as the~$b_r$ for $S^r \in O$ are all in
$\dom(\calE(D))$, we must have $D = D'$,
and $\mybf{t}$ witnesses that $\pi_O(\mybf{b}) \in \calE(D)$.
Hence $\calE$ is still a global envelope of~$J'$.

\medskip

Last, to see that the resulting $J'$ is $(n-1)$-envelope-saturated, it suffices to
observe that, for each unsafe class $D \in
\afactcl$, the remaining tuples of $\calE(D)$ for $J'$ are those of $\calE(D)$ for
$J$ minus at most one tuple (namely, some projection of $F_{\fnew}$). This concludes the proof.
\end{proof}

Hence, we know that envelope-thrifty chase steps preserve being aligned and also preserve
$\fds$ (rather than $\ufds$ for fact-thrifty chase steps). Our goal is then to
modify the Fact-Thrifty Completion Proposition of the previous \secname
(Proposition~\ref{prp:ftcomp}) to use envelope-thrifty rather than
fact-thrifty chase steps, relying on the previous lemmas to preserve all
invariants. The problem is that unlike fact-saturation, envelope-saturation
``runs out''; whenever we use a remaining tuple $\mybf{t}$ in a chase step to
create~$F_{\fnew}$ and obtain a new aligned superinstance~$J'$, then we can no longer
use the same~$\mybf{t}$ in~$J'$. This is why the result of an envelope-thrifty chase step
is less saturated than its input, and it is why we made sure in the Sufficiently
Envelope-Saturated Solutions Proposition (Proposition~\ref{prp:preproc})
that we could construct arbitrarily saturated superinstances.

For this reason, before we modify the Fact-Thrifty Completion Proposition, we
need to account for the number of chase steps that the proposition performs. We show that it is
linear in the size of the input instance.

\begin{lemma}[Accounting]
  \label{lem:eblowup}
  There exists $B \in \mathbb{N}$ depending only on $\sigma$, $k$, and $\uconb$, such that,
  for any aligned superinstance $J = (I, \cov)$ of~$I_0$,
  letting $L$ be the preserving fact-thrifty sequence constructed in the Fact-Thrifty Completion
  Proposition (Proposition~\ref{prp:ftcomp}), we have $\card{L} < B \cdot \card{I}$.
\end{lemma}

\begin{proof}
  It suffices to show that $\card{L(J)} < B \cdot \card{I}$, because, as each chase
  step creates one fact, we have $\card{L} \leq \card{L(J)}$.

Remember that the fact-thrifty completion process starts by constructing an ordered
partition $\mybf{P} = (P_1, \ldots, P_{\neqidsc})$
of $\idsb$ (Definition~\ref{def:opartition}).
This $\mybf{P}$ does not depend on~$I$.
Hence, as we satisfy the
$\uid$s of each $P_i$ in turn, if we can show that the instance size only
increases by a multiplicative constant for each class, then the blow-up
for the entire process is by a multiplicative constant (obtained as the product
of the constants for each $P_i$).

For trivial classes, we apply one chase round by fresh fact-thrifty chase
steps (Lemma~\ref{lem:trivscc}),
It is easy to see that applying a chase round 
by any form of thrifty chase step
on an aligned superinstance $J_1 = (I_1, \cov_1)$
yield a result whose size has only increased relative to~$J_1$ by a
multiplicative constant.
This
is because $\card{\dom(I_1)} \leq \arity{\sigma} \cdot \card{I_1}$, and the number of
facts created per element of~$I_1$ in a chase round is at most
$\card{\positions(\sigma)}$.
Hence, for trivial classes, we only incur a blowup by a constant multiplicative
factor.

For non-trivial classes, we apply the
Reversible Fact-Thrifty Completion Proposition (Proposition~\ref{prp:ftcompr}).
Remember that
this lemma first ensures $k$-essentiality by applying $k+1$
fact-thrifty chase rounds (Lemma~\ref{lem:achkrev}) and then makes the
result satisfy $\idsb$ using the sequence constructed by the 
Reversible Relation-Thrifty Completion Proposition.
Ensuring $k$-essentiality only implies a blow-up by a multiplicative constant,
because it is performed by applying $k+1$ fact-thrifty chase rounds, so we can
use the same reasoning as for trivial classes.
Hence, we focus
on the Reversible Relation-Thrifty Completion Proposition,
and show that it also causes only a blow-up by a
multiplicative constant.

When we apply the Reversible Relation-Thrifty Completion Proposition
to an instance~$I$, we start by constructing a balanced
pssinstance $P$ using the Balancing Lemma (Lemma~\ref{lem:hascompletion}), and a
$\uconb$-compliant piecewise realization $\pire$ of $P$
by the Realizations Lemma (Lemma~\ref{lem:balwsndc}),
and we then apply fact-thrifty chase steps to
satisfy $\idsb$ following $\pire$.
We know that, whenever we apply a fact-thrifty chase step to an element $a$,
the element $a$ occurs after the chase step at a new
position where it did not occur before. Hence, it suffices to show that
$\card{\dom(P)}$ is within a constant factor of $\card{I}$, because then we know
that the final number of facts 
created by the sequence of the Reversible Relation-Thrifty Completion Proposition
will be $\leq
\card{\dom(P)} \cdot \card{\positions(\sigma)}$.

To show this,
remember that $\dom(P) = \dom(I) \sqcup \calH$, where $\calH$ is the helper set.
Hence, we only need to show that $\card{\calH}$ is within a
multiplicative constant factor of~$\card{I}$.
From the proof of the Balancing
Lemma, we know that $\calH$ is a disjoint union of $\leq \card{\positions(\sigma)}$ sets
whose size is linear in $\card{\dom(I)}$ which is itself $\leq \arity{\sigma}
\cdot \card{I}$. Hence, the
Reversible Relation-Thrifty Completion Proposition only causes a blowup by a
constant factor.
As we justified, this implies the same about the entire
completion process, and concludes the proof.
\end{proof}

This allows us to deduce the minimal level of envelope-saturation required to
adapt the Fact-Thrifty Completion Proposition (Proposition~\ref{prp:ftcomp}):

\begin{proposition}[Envelope-thrifty completion]
  \label{prp:etcomp}
  Let $\con = \fds \fragcup \ids$ be finitely closed $\fd$s and $\uid$s,
  let $B \in \NN$ be as in the Accounting Lemma (Lemma~\ref{lem:eblowup}), and
  let $I_0$ be an instance that satisfies $\fds$.
  For any $(B \cdot \card{J})$-envelope-saturated aligned superinstance $J$ of~$I_0$ that satisfies~$\fds$, we
  can obtain by envelope-thrifty chase steps an aligned superinstance $J_\f$ of
  $I_0$
  that satisfies $\con$.
\end{proposition}

\begin{proof}
  We define envelope-thrifty sequences, and preserving
  envelope-thrifty sequences, analogously to (preserving) fact-thrifty sequences
  (Definition~\ref{def:tsec} and Definition~\ref{def:preserving}) in the expected manner,
  but further requiring that all intermediate aligned superinstances remain
  envelope-saturated. This definition makes sense thanks to 
  the Envelope-Thrifty Preservation Lemma (Lemma~\ref{lem:envpres}).
 
  By the Fact-Thrifty Completion Proposition (Proposition~\ref{prp:ftcomp}), there exists a preserving
  fact-thrifty sequence $L$ such that $L(J)$ satisfies $\idsb$, and 
  $\card{L} < B \cdot \card{I}$. 
  Construct from $L$ an envelope-thrifty sequence $L'$ that non-dangerously
  matches~$L$, by changing each
  fact-thrifty chase step to an envelope-thrifty chase step, which we can do at each
  individual step thanks to 
  the Envelope-Thrifty Applicability Lemma (Lemma~\ref{lem:etappl}).
  It is clear that this is a preserving envelope-thrifty sequence, thanks to
  the Envelope-Thrifty Preservation Lemma, and thanks to the fact that 
  the Ensuring Essentiality Lemma (Lemma~\ref{lem:achkrev}) clearly adapts from
  fact-thrifty chase steps to envelope-thrifty chase steps: again, it only
  relies on the fact that, letting $F_{\fnew} = S(\mybf{b})$ be the new fact,
  $F_\w = S(\mybf{b}')$ the chase witness, and $\tau: \ui{R^p}{S^q}$ the $\uid$,
  we have $\cov'(b_r) \bbsim_k b'_r$ in
  $\chase{I_0}{\idsb}$ for all $S^r \in \nondanger(S^q)$. This also uses the
  fact that, by the Accounting Lemma (Lemma~\ref{lem:eblowup}), we have $\card{L} \leq B \cdot \card{I}$, so by
  the Envelope-Thrifty Preservation Lemma, all intermediate aligned
  superinstances remain envelope-saturated.

  Hence, $J_\f \defeq L'(J)$ is an aligned superinstance of $I_0$.
  Further, by the Thrifty Sequence Rewriting Lemma (Lemma~\ref{lem:thsrw}),
  as $L(J) \models \idsb$, so does $J_\f$. Last, as $J \models \fds$, 
  by the Envelope-Thrifty $\fd$ Preservation Lemma (Lemma~\ref{lem:envfds}),
  so does~$J_\f$. This concludes the
  proof.
\end{proof}

We can now conclude the proof of Theorem~\ref{thm:hfds}.
Start by applying the saturation process
of the Sufficiently Envelope-Saturated Solutions Proposition (Proposition~\ref{prp:preproc})
to obtain an aligned superinstance $J = (I, \cov)$ of a disjoint union $I_0'$ of
copies of~$I_0$,
such that $J$ satisfies $\fds$ and is $(B \cdot \card{I})$-envelope-saturated.
Now, apply the Envelope-Thrifty Completion Proposition (Proposition~\ref{prp:etcomp})
to obtain an aligned
superinstance $J_\f = (I_\f, \cov_\f)$ of $I_0'$ that satisfies~$\con$. We know
that $I_\f$ satisfies $\con$ and is a $k$-sound superinstance of~$I_0'$ for
$\acq$, but clearly
it is also a $k$-sound superinstance of $I_0$, as is observed by the $k$-bounded
simulation from $I'$ to $\chase{I_0}{\idsb}$ obtained by composing $\cov'$ with
the obvious homomorphism from $\chase{I_0'}{\idsb}$ to $\chase{I_0}{\idsb}$.
This concludes the proof.

\subsection{Constructing Dense Interpretations}
\label{sec:combinatorial}
All that remains is to show the Dense Interpretations Theorem:

\combinatorial*

Fix the relation $R$, and let $\fds$ be an arbitrary set of~$\fd$s which we
assume is closed under $\fd$ implication.
Let $\ufds$ be the $\ufd$s implied by $\fds$; it is also closed under $\fd$
implication. Recall
the definition of~$\ovl$ (Definition~\ref{def:overlap}).
We introduce a notion of \defo{tame overlaps} for $\ufds$,
which depends only on $\ufds$ but is
a sufficient condition to satisfy~$\fds$, as we will show.

\begin{definition}
  We say a subset $O \subseteq \pos(R)$ is \deft{tame} for $\ufds$ if $O$ is
  empty or for every
  $R^p \in \pos(R) \backslash O$, there exists $R^q \in \pos(R)$ such that:
  \begin{itemize}
    \item for all $R^s \in O$, the $\ufd$ $R^q \rightarrow R^s$ is in~$\ufds$,
    \item the $\ufd$ $R^q \rightarrow R^p$ is not in $\ufds$.
  \end{itemize}
  We say that an instance $I$ has the \deft{tame
  overlaps} property (for $\ufds$) if for every $F \neq F'$ of~$I$, $\ovl(F, F')$ is
  tame.
\end{definition}

\begin{example}
  \label{exa:tame}
  Consider a $5$-ary relation $R$ and $\ufds$ containing $R^1 \rightarrow R^5$,
  $R^2 \rightarrow R^4$, $R^2 \rightarrow R^5$, $R^3 \rightarrow R^4$, and $R^3
  \rightarrow R^5$. The subset $O = \{R^1, R^5\}$ is tame, because it is determined
  by $R^1$ and all other positions are not determined by~$R^1$, so we can always
  take $R^q = R^1$. In fact, more generally, when there exists a position that
  determines exactly $O$, then $O$ is tame; we will show a refinement of this as
  Lemma~\ref{lem:tamekey}. 

  However, this is not a characterization, because the subset $\{R^4, R^5\}$ is also
  tame: for $R^p = R^2$ we take $R^q = R^3$, for $R^p = R^3$ we take $R^q =
  R^2$, and for $R^p = R^1$ we take $R^q$ to be one of $R^2$ or $R^3$.

  The subsets $\{R^4\}$ and $\{R^5\}$ are also tame (always taking $R^q = R^4$
  or $R^q = R^5$ respectively).

  The subset $O = \{R^1, R^4\}$ is not tame, because $\pos(R)
  \backslash O$ is non-empty but there is no single position determining all
  positions of~$O$. The
  subset $\{R^2, R^4\}$ is not tame because for $R^p = R^5$ there is no choice
  for~$R^q$.
\end{example}

We now claim the following lemma, and its immediate corollary:

\begin{lemma}
  \label{lem:tameoverlaps}
  If $O \subseteq \pos(R)$ is tame for~$\ufds$ then there is no $\fd$ $\phi: R^L
  \rightarrow R^r$ in $\fds$ such that $R^L \subseteq O$ but $R^r \notin O$.
\end{lemma}

\begin{proof}
  If $O$ is empty the claim is immediate. Otherwise,
  assume to the contrary the existence of such an $\fd$~$\phi$. As $R^r \notin O$ and $O$ is
  tame, there is $R^q \in \pos(R)$ such that $R^q \rightarrow R^s$
  is in $\ufds$ for all $R^s \in O$, but $\phi' : R^q \rightarrow R^r$ is not in $\ufds$.
  Now, as $R^L \subseteq O$, we know that $R^q \rightarrow R^s$
  is in~$\ufds$ for all $R^s \in R^L$, so that, by transitivity with $\phi$,
  as $\fds$ is closed by implication,
  $\phi'$ is in $\fds$. As $\phi'$ is a $\ufd$, by
  definition of~$\ufds$, $\phi'$ is in~$\ufds$, a contradiction.
\end{proof}

\begin{corollary}
  \label{cor:safefds}
  For any instance $I$,
  if $I$ has the tame overlaps property for $\ufds$, then $I$ satisfies $\fds$.
\end{corollary}

\begin{proof}
  Considering any two facts $F$ and $F'$ in~$I$, as $O \defeq \ovl(F, F')$ is tame, we know
  by Lemma~\ref{lem:tameoverlaps}
  that for any $\fd$ $\phi: R^L \rightarrow R^r$ in $\fds$, we cannot have $R^L
  \subseteq O$ but $R^r \notin O$. Hence, $F$ and $F'$ cannot be a violation
  of~$\phi$.
\end{proof}

We forget for now the disjointness condition in 
the Dense Interpretations Theorem (Theorem~\ref{thm:combinatorial}),
which we will prove at the very end of the \subsecname
(Corollary~\ref{cor:combinatorial2}), and focus only on the first part.
We claim the following generalization of the result:

\begin{theorem}
  \label{thm:fds}
  Let $R$ be a relation and $\ufds$ be a set of~$\ufd$s over $R$.
  Let $D$ be the smallest possible cardinality of a \deft{key} $K$ of~$R$
  (i.e., $K \subseteq \pos(R)$ and for all $R^q \in \pos(R)$, there is $R^p \in
  K$ such that $R^p \rightarrow R^q$ is in~$\ufds$).
  Let $x$ be ${D \over D-1}$ if $D > 1$ and $1$ otherwise.
  
  For every
  $N \in \NN$, there exists a finite instance~$I$ of~$R$ such that
  $\card{\dom(I)}$ is $O(N)$, $\card{I}$ is $\Omega(N^x)$,
  and $I$ has the tame overlaps property for~$\ufds$.
\end{theorem}

Observe that, thanks to the use of the tame overlaps, the result does not
mention higher-arity $\fd$s, only $\ufd$s; intuitively, tame overlaps ensures
that the construction works for any $\fd$s that have the same consequences as
$\ufd$s.

It is clear that this theorem implies the first part of 
the Dense Interpretations Theorem (Theorem~\ref{thm:combinatorial}), because
if $R$ has no unary key for $\fds$ then $D > 1$ and thus $x > 1$, which implies that, for any~$K$, by taking
a sufficiently large $N_0$, we can obtain for all $N \geq N_0$ an instance~$I$
for $R$ with $\leq N$
elements and $\geq K\cdot N$ facts that has the tame overlaps property for $\ufds$; now, by
Lemma~\ref{cor:safefds}, this implies that $I$ satisfies $\fds$.

\medskip

In the rest of this \subsecname, we prove Theorem~\ref{thm:fds}, until the very
end where we additionally show that we can enforce the disjointness condition
for the Dense Interpretations Theorem. Fix the relation $R$ and set of~$\ufd$s
$\ufds$. The case of~$D = 1$ is vacuous and can be eliminated directly (consider
the instance $\{R(a_i, \ldots, a_i) \mid 1 \leq i \leq N\}$). Hence,
assume that $D > 1$, and let $x \defeq {D \over D-1}$.

\medskip

We first show the claim on a specific relation $R_{\full}$ and set $\ufds^{\full}$ of~$\ufd$s. We
will then generalize the construction to arbitrary relations and $\ufd$s. Let
$T \defeq \{1, \ldots, D\}$, and consider a bijection $\nu: \{1, \ldots, 2^D-1\}
\to \parts(T) \backslash \{\emptyset\}$,
where $\parts(T)$ is the powerset of~$T$. 
Let $R_{\full}$ be a $(2^D-1)$-ary relation,
and take $\ufds^{\full} \defeq \{R^i \rightarrow R^j \mid \nu(i)
\subseteq \nu(j)\}$. Note that $\ufds^{\full}$ is clearly closed under implication of
$\ufd$s.
Fix $N \in \NN$, and let us build an instance $I_{\full}$ with $O(N)$
elements and $\Omega(N^x)$ facts.

Fix $n \defeq \lfloor N^{1/(D-1)} \rfloor$.
Let $\calF$ be the set of partial functions from $T$ to $\{1, \ldots, n\}$,
and write $\calF = \calF_\t \sqcup \calF_\p$, where $\calF_\t$ and $\calF_\p$
are respectively the total and the strictly partial functions.
We take $I_{\full}$ to consist of one fact $F_f$ for each $f \in \calF_\t$, where $F_f
= R_{\full}(\mybf{a}^{\mybf{f}})$
is defined as follows: for $1 \leq i \leq 2^D-1$, $a^f_i \defeq \restr{f}{T \backslash \nu(i)}$. In particular:
\begin{itemize}
\item $a^{f}_{\nu^{-1}(T)}$, the element of~$F_{f}$ at the position mapped
  by~$\nu$ to
  $T \in \parts(T) \backslash \{\emptyset\}$, is the strictly partial function
  that is nowhere defined; note that $R_{\full}^{\nu^{-1}(T)}$ is determined by
  \emph{all} positions in~$\ufds^{\full}$.
\item $a^{f}_{\nu^{-1}(\{i\})}$, the element of
  $F_{f}$ at the position mapped by~$\nu$ to $\{i\} \in \parts(T) \backslash
  \{\emptyset\}$,
  is the strictly partial function equal to $f$ except that it is undefined on
  $i$;
  note that $R_{\full}^{\nu^{-1}(\{i\})}$ is determined by no other position of
  $R_{\full}$ in $\ufds^{\full}$.
\end{itemize}
Hence, $\dom(I_{\full}) = \calF_\p$ (because $\emptyset$ is not in the image of
$\nu$), so that $\card{\dom(I_{\full})} = \sum_{0 \leq i < D} {D \choose i} n^i$.
Remembering that $D$ is a constant, this implies that $\card{\dom(I_{\full})}$ is
$O(n^{D-1})$, so it is $O(N)$ by definition of $n$.
Further, we claim that $\card{I_{\full}} = \card{\calF_\t} = n^D = N^x$.
To show this, consider
two facts $F_f$ and $F_g$. We show that $F_f = F_g$ implies $f = g$,
so there are indeed $\card{\calF_\t}$ different facts in~$I_{\full}$. As
$\pi_{\nu^{-1}(\{1\})}(F_f) =
\pi_{\nu^{-1}(\{1\})}(F_{g})$, we have $f(t) = g(t)$ for all $t \in T \backslash \{1\}$,
and as $D \geq 2$, we can look at $\pi_{\nu^{-1}(\{2\})}(F_f)$ and
$\pi_{\nu^{-1}(\{2\})}(F_{g})$ to conclude that $f(1) = g(1)$, hence $f = g$ as claimed. 
Hence, the cardinalities of~$I_{\full}$ and of its domain are suitable.

We must now show that $I_{\full}$ has the tame overlaps property
for~$\ufds^{\full}$. For this we first make the
following general observation:

\begin{lemma}
  \label{lem:tamekey}
  Let $\ufds$ be any conjunction of $\ufd$s and $I$ be an instance such that $I
  \models \ufds$. Assume that, for any pair of facts $F \neq F'$ of~$I$ that overlap, there
  exists $R^p \in \ovl(F, F')$ which is a unary key for $\ovl(F, F')$. Then $I$
  has the tame overlaps property for $\ufds$.
\end{lemma}

\begin{proof}
  Consider $F, F' \in I$ and $O \defeq \ovl(F, F')$. If
  $F = F'$, then $O = \pos(R)$, and $O$ is vacuously tame.
  Otherwise, if $F \neq F'$, let $R^p \in \pos(R) \backslash O$. We take
  $R^q \in O$ to be the unary key of~$O$. We know that $R^q \rightarrow R^s$ is in
  $\ufds$ for all~$R^s \in O$, so to show that $O$ is tame it suffices to show that $\phi: R^q \rightarrow
  R^p$ is not in $\ufds$. However, if it were, then as $R^q \in O$ and
  $R^p \notin O$, $F$ and $F'$ would witness a violation of~$\phi$,
  contradicting the fact that $I$ satisfies $\ufds$.
\end{proof}

So we show that $I_{\full}$ satisfies $\ufds^{\full}$ and that every non-empty overlap
between facts of~$I_{\full}$ has a unary key, so we can conclude by
Lemma~\ref{lem:tamekey} that $I_{\full}$ has tame overlaps.

First, to show that $I_{\full}$ satisfies $\ufds^{\full}$, observe that (*) whenever $\phi: R_{\full}^i
\rightarrow R_{\full}^j$ is in $\ufds^{\full}$,
then $\nu(i) \subseteq \nu(j)$, so that, for any fact~$F$ of~$I_{\full}$, for any $1
\leq t \leq D$,
whenever $(\pi_j(F))(t)$ is defined,
so is $(\pi_i(F))(t)$, and we have $(\pi_j(F))(t) = (\pi_i(F))(t)$.
Further, by our construction, we easily see that (**) for any fact~$F$
of~$I_{\full}$, for any $1 \leq i \leq 2^D-1$ and
$1 \leq t \leq D$, the fact that $(\pi_i(F))(t)$ is defined or not only depends on $i$
and $t$, not on $F$.
Hence,
consider a $\ufd$ $\phi: R_{\full}^i \rightarrow R_{\full}^j$ in~$\ufds^{\full}$,
let $F$ and $F'$ be two facts of~$I_{\full}$ such that $\pi_i(F) = \pi_i(F')$,
and show that $\pi_j(F) = \pi_j(F')$. Take $1 \leq t \leq D$ and show that
either $(\pi_j(F))(t)$ and $(\pi_j(F'))(t)$ are both undefined, or they are both
defined and equal. By (**), either both are undefined or both are defined, so it
suffices to show that if they are defined then they are equal. But then, if both
are defined, by (*), we have $(\pi_j(F'))(t) = (\pi_i(F'))(t) = (\pi_i(F))(t) =
(\pi_j(F))(t)$. So we conclude indeed that $\pi_j(F) = \pi_j(F')$, so that 
$F$ and $F'$ cannot witness a violation of $\phi$.
Hence, $I_{\full} \models \ufds^{\full}$.

Second, to show that non-empty overlaps in~$I_{\full}$
have unary keys,
consider two facts 
$F_{f} = R_{\full}(\mybf{a}^{\mybf{f}})$ and $F_{g} = R_{\full}(\mybf{a}^{\mybf{g}})$, with $f \neq g$ so
that $F_{f} \neq F_{g}$. Assume that $\ovl(F_f, F_g)$ is non-empty,
and let us show that it has a unary key.
Let $O \defeq \{t \in T \mid f(t) = g(t)\}$, and let $X = T
\backslash O$; we have $X \neq \emptyset$, because otherwise $f =
g$, so we can define $p \defeq \nu^{-1}(X)$.
We will show that 
\[
\ovl(F_f, F_g) = \{R^i \in \positions(R_{\full}) \mid X \subseteq
\nu(i)\}
\]
This implies that $R^p \in \ovl(F_f, F_g)$ and that $R^p$ is a unary
key of~$\ovl(F_f, F_g)$, because,
for all $R^q \in \overlap(F_{f}, F_{g})$, $X \subseteq \nu(q)$, so that
$R^p \rightarrow R^q$ is in $\ufds^{\full}$.

To show the equality above, consider $R^i$ such that $X \subseteq \nu(i)$. Then $T \backslash
\nu(i) \subseteq T \backslash X$. 
Because $a^f_i = \restr{f}{T \backslash \nu(I)}$ and $a^g_i = \restr{g}{T \backslash
\nu(I)}$, we have $a^f_i = a^g_i$ by definition of~$O = T \backslash X$. Thus
$R^i \in \overlap(F_{f}, F_{g})$. Conversely, if $R^i \in
\overlap(F_{f}, F_{g})$, then we have $a^f_i =
a^g_i$, so by definition of~$O$ we must have $T
\backslash \nu(i) \subseteq O = T \backslash X$, which implies $X \subseteq
\nu(i)$.

Hence, $I_{\full}$ is a finite instance of~$\ufds^{\full}$ which satisfies the tame overlaps
property and contains $O(N)$ elements and
$\Omega(N^{x})$ facts. This concludes the proof of Theorem~\ref{thm:fds}
for the specific case of~$R_{\full}$ and $\ufds^{\full}$.

\medskip

Let us now show Theorem~\ref{thm:fds} for an arbitrary
relation $R$ and set $\ufds$ of $\ufd$s.
Let $K$ be 
a key of~$R$ of minimal cardinality, so that $\card{K} = D$.
Let $\lambda$ be a bijection from~$K$ to~$T$.
Extend $\lambda$ to a function $\mu$
such that, for all $R^p \in \pos(R)$, we set
$\mu(R^p) \defeq \{\lambda(R^k) \mid R^k \in K\text{~such~that~}R^k = R^p\text{~or~}R^k \rightarrow
R^p\text{~is~in~}\ufds\}$; note that this set is never empty.

Consider the instance $I_{\full}$ for relation $R_{\full}$ that we defined
previously,
and create an instance $I$ of~$R$ that contains, for every fact
$R_{\full}(\mybf{a})$ of~$I_{\full}$, a fact $F = R(\mybf{b})$ in $I$, with $b_i =
a_{\nu^{-1}(\mu(R^i))}$ for all $1 \leq i \leq \card{R}$.

We first show that $\card{\dom(I)} = O(N)$ and $\card{I} = \Omega(N^x)$. Indeed, for the
first point, we have $\dom(I) \subseteq \dom(I_{\full})$, and as we had
$\card{\dom(I_{\full})} = O(N)$, we deduce the same of $\dom(I)$. For the second
point, it suffices to show that we never create the same fact twice in~$I$ for
two different facts of~$I_{\full}$. Assume that there are two facts
$F_f = R_{\full}(\mybf{a})$ and $F_g = R_{\full}(\mybf{a}')$ in~$I_{\full}$ for which we created
the same fact $F = R(\mybf{b})$ in~$I$, and let
us show that we then have $f = g$ so that $F_f = F_g$. As
$\card{K} \geq 2$, consider $R^{k_1} \neq R^{k_2}$ in~$K$. 
We have $\mu(R^{k_1})
= \{\lambda(R^{k_1})\}$ and $\mu(R^{k_2}) = \{\lambda(R^{k_2})\}$. Hence, let $i_j
\defeq \lambda(R^{k_j})$ for $j \in \{1, 2\}$; as $\lambda$ 
is bijective, we deduce from $R^{k_1} \neq R^{k_2}$ that $i_1 \neq i_2$.
From the definition of
$b_{k_1}$ we deduce that $a_{\nu^{-1}(\{i_1\})} =
a'_{\nu^{-1}(\{i_1\})}$, and likewise $a_{\nu^{-1}(\{i_2\})} =
a'_{\nu^{-1}(\{i_2\})}$.
Similarly to the proof of why $I_{\full}$ has no duplicate facts, this implies that
$f(t) = g(t)$ for all $t \in T \backslash \{i_1\}$ and for all $t \in T
\backslash \{i_2\}$. As $i_1 \neq i_2$, we conclude
that $f = g$, so that $F_f = F_g$. Hence, we have
$\card{I} = \card{I_{\full}} = \Omega(N^x)$.

Let us now show that $I$ has tame overlaps for~$\ufds$. Consider two facts $F, F'$ of
$I$ that overlap, and let $O \defeq \ovl(F, F')$.
We first claim that there exists $\emptyset \subsetneq K' \subseteq K$,
such that, letting $X' \defeq \{\lambda(R^k) \mid R^k \in K'\}$, we have
$\ovl(F, F') = \{R^i \in \positions(R) \mid X' \subseteq \mu(R^i)\}$. Indeed,
letting $F_f$ and $F_g$ be the facts of~$I_{\full}$ used to create $F$ and $F'$, we
previously showed the existence of~$\emptyset \subsetneq X \subseteq T$ such that 
$\ovl(F_f, F_g) = \{R^i \in \positions(R_{\full}) \mid X \subseteq \nu(i)\}$.
Our definition of~$F$ and $F'$ from $F_f$ and $F_g$ makes it clear that we can
satisfy the condition by taking $K' \defeq \lambda^{-1}(X)$, so that $X' = X$.

Consider now $R^p \in \pos(R) \backslash O$. We
cannot have $X' \subseteq \mu(R^p)$, otherwise $R^p \in O$. Hence, there exists
$R^k \in K'$ such that $\lambda(R^k) \notin \mu(R^p)$. This implies that $R^k
\rightarrow R^p$ is not in $\ufds$. However, as $R^k \in K'$, we have
$\lambda(R^k) \in \mu(R^q)$ for all $R^q \in O$, so that $R^k \rightarrow O$
is in $\ufds$. This proves that~$O = \ovl(F, F')$ is tame. Hence, $I$ has the
tame overlaps property, which concludes the proof of Theorem~\ref{thm:fds}.

\medskip

The only thing left is to show that we can enforce the disjointness condition in
the Dense Interpretations Theorem (Theorem~\ref{thm:combinatorial}), namely:

\begin{corollary}
  \label{cor:combinatorial2}
  We can assume in the Dense Interpretations Theorem
  (Theorem~\ref{thm:combinatorial}) the following \textbf{disjointness
  condition} on the resulting instance $I$:
  each element occurs at exactly one position of the relation $R$. Formally, for
  all $a \in \dom(I)$, there exists exactly one $R^p \in \positions(R)$ such
  that $a \in \pi_{R^p}(I)$.
\end{corollary}

\begin{proof}
  Let $I$ be the instance constructed in the proof of the Dense Interpretations
  Theorem, 
  and consider the instance $I'$ whose domain is $\{(a, R^p) \mid a \in
  \dom(I), R^p \in \positions(\sigma)\}$ and which contains for every fact $F =
  R(\mybf{a})$
  of~$I$ a fact $F' = R(\mybf{b})$ such that $b_p = (a_p, R^p)$ for every
  $R^p \in \positions(\sigma)$. Clearly this defines a bijection $\phi$ from the
  facts of~$I$ to the facts of~$I'$, and for any facts $F, F'$ of~$I'$,
  $\ovl(F, F') = \ovl(\phi^{-1}(F), \phi^{-1}(F'))$. Thus any violation
  of the $\fd$s $\fds$ in $I'$ would witness one in $I$. Of course,
  $\card{\dom(I')} = \arity{\sigma} \cdot \card{\dom(I)}$, so to achieve a
  constant factor of~$K$ between the domain size and instance size with the
  disjointness condition, we need to use the proof of the Dense Interpretations
  Theorem (Theorem~\ref{thm:combinatorial}) with a constant factor of $K' \defeq \arity{\sigma} \cdot K$.
\end{proof}

\section{Blowing up Cycles}
\label{sec:cycles}
We are now ready to prove the Universal Models Theorem, which concludes the
proof of our Main Theorem (Theorem~\ref{thm:mymaintheorem}):

\univexists*

To do this, we must ensure $k$-soundness 
for arbitrary Boolean $\cq$s rather than just acyclic $\cq$s.

Intuitively, the only cyclic $\cq$s that hold in $\chase{I_0}{\idsb}$ either
have an acyclic self-homomorphic match (so they are implied by an acyclic $\cq$ that
also holds) or have all cycles matched to elements of~$I_0$.
Hence, in a $k$-sound instance for $\cq$, no other cyclic queries should be true.
Our way to ensure this is by a cycle blowup process:
starting with the superinstance constructed by Theorem~\ref{thm:hfds}, which
satisfies~$\con$ and is $k$-sound for $\acq$,
we build its product with a group of high girth.
The standard way to do so, inspired by \cite{otto2002modal}, is presented in
Section~\ref{sec:simpleprod}.

The problem is that this blowup process
may create $\fd$ violations. We work around this problem using some
additional properties ensured by our construction. In
Section~\ref{sec:cautiousness}, we accordingly show the Cautious Models Theorem, 
a variant of Theorem~\ref{thm:hfds} with additional properties.
Section~\ref{sec:cautiousness} is the only part of this \secname that
depends on the details of the previous \secnames.

We then apply a slightly different blowup construction to that model,
as described in Section~\ref{sec:mixedprod}, which ensures that no $\fd$
violations are created.
This blowup no longer depends on the specifics of the construction,
and does not depend on the specific $\uid$s and $\fd$s that hold; in
particular, the blowup constructions do not even require that the $\uid$s and $\fd$s are finitely closed.

\subsection{Simple Product}
\label{sec:simpleprod}

We first define a simple notion of product, which we can use to extend $k$-soundness
from $\acq$ to $\cq$, but which may introduce $\fd$ violations.
Let us first introduce preliminary notions:

\begin{definition}
  \label{def:group}
  A \deft{group} $G = (S, \cdot)$ over a finite set~$S$ consists of:
  \begin{itemize}
    \item an associative \deft{product law} $\cdot : S \times S \rightarrow S$;
    \item a \deft{neutral element} $e \in S$ such that $e \cdot x = x \cdot e =
      x$ for all $x \in S$;
    \item an \deft{inverse law} $\cdot^{-1} : S \rightarrow S$
  such that $x \cdot x^{-1} = x^{-1} \cdot x = e$ for all $x \in S$.
  \end{itemize}
  We say that $G$ is
  \deft{generated} by $X \subseteq S$ if all elements of~$S$ can be written as a
  product of elements of~$X$ and $X^{-1} \defeq \{x^{-1} \mid x \in X\}$.

  Given a group $G = (S, \cdot)$ generated by~$X$, assuming $\card{S} > 2$, the
  \deft{girth} of~$G$ under $X$ is the length of the shortest non-empty word
  $\mybf{w}$ of elements of~$X$ and $X^{-1}$
  such that $w_1 \cdots w_n = e$ and $w_i \neq w_{i+1}^{-1}$ for all $1 \leq i <
  n$.
\end{definition}

The following result, originally from \cite{margulis1982explicit}, is proven
for $\card{X} > 1$
in, e.g., \cite{otto2012highly} (Section~2.1), and is straightforward for
$\card{X} = 1$
(take $\ZZ/n\ZZ$):

\begin{lemma}
  \label{lem:groups}
  For all $n \in \mathbb{N}$ and finite non-empty set~$X$, there is a finite group $G = (S,
  \cdot)$ generated by $X$
  with girth $\geq n$ under $X$.
  We call $G$ an \deft{$\bm{n}$-acyclic group generated by~$X$}.
\end{lemma}

In other words, in an $n$-acyclic group generated by $X$, there is no short
product of elements of~$X$ and their inverses which evaluates to~$e$, except
those that include a factor $x \cdot x^{-1}$.

We now explain how to take the product of a superinstance~$I$ of~$I_0$ with such a finite group~$G$. This 
ensures that any cycles in the product instance are large, because they project to
cycles in~$G$, i.e., words evaluating to~$e$ as in Definition~\ref{def:group}. We use a specific generator:

\begin{definition}
  The \deft{fact labels} of a superinstance~$I$ of~$I_0$ are $\lab{I} \defeq
  \{\ll^F_i \mid F \in I \backslash I_0, 1 \leq i \leq \card{F}\}$, where
  $\card{F}$ is the arity of the relation for fact~$F$.
\end{definition}

Now, we define the product of a superinstance $I$ of~$I_0$
with a group generated by~$\lab{I}$.
We make sure not to blow up cycles in~$I_0$, so the result remains a
superinstance of~$I_0$:

\begin{definition}
  \label{def:prod}
  Let $I$ be a finite superinstance of~$I_0$ and 
  $G$ be a finite group generated by $\lab{I}$. The \deft{product of~$I$
  by~$G$ preserving $I_0$},
  written
  $(I, I_0) \sprod G$,
  is the finite instance
  with domain $\dom(I) \times G$
  consisting of the following facts, for all $g \in G$:
  \begin{itemize}
  \item For every fact $R(\mybf{a})$ of~$I_0$, the fact $R((a_1, g), \ldots,
    (a_{\card{R}}, g))$.
  \item For every fact $F = R(\mybf{a})$ of~$I \backslash I_0$, the
    fact
    $R((a_1, g \cdot \ll^F_1), \ldots, (a_{\card{R}}, g \cdot \ll^F_{\card{R}}))$.
  \end{itemize}
  We identify $(a, e)$ to $a$ for $a \in \dom(I_0)$,
  so $(I, I_0) \sprod G$ is still a superinstance of~$I_0$.
\end{definition}

It will be simpler to reason about initial instances $I_0$ where each element
has been \defo{individualized} by the addition of a fresh fact that is unique to
that element. We give a name to this notion:

\begin{definition}
  \label{def:individualizing}
  An \deft{individualizing} instance $I_0$ is such that, for each $a \in
  \dom(I_0)$, $I_0$ contains a fact $P_a(a)$ where $P_a$ is a fresh unary
  predicate which does not occur in queries, in $\uid$s or in $\fd$s.

  An \deft{individualizing superinstance} of an instance $I_0$ is a
  superinstance $I_1$ of~$I_0$ that adds precisely one
  unary fact $P_a(a)$, for a fresh unary relation $P_a$, to each 
  $a \in \dom(I_0)$, so that $I_1$ is individualizing. In particular, we have
  $\dom(I_0) = \dom(I_1)$, and $I_0$ and $I_1$ match for all relations
  of~$\sigma$ that occur in the query $q$ and the constraints~$\con$.
\end{definition}

We can now state the following property, which we will prove in the rest of
this \subsecname:

\begin{lemma}[Simple product]
  \label{lem:prodppty}
  Let $\con$ be finitely closed $\fd$s and $\uid$s,
  let $I$ be a finite superinstance of an individualizing~$I_0$ and
  let $G$ be a finite $(2k+1)$-acyclic group generated by~$\lab{I}$.
  If $I$ is $(k \cdot (\arity{\sigma}+1))$-sound for~$\acq$, $I_0$, and~$\con$,
  then $(I, I_0) \sprod G$ is \mbox{$k$-sound} for~$\cq$, $I_0$, and~$\con$.
\end{lemma}

The following example illustrates the idea of taking the simple product of an instance
with a group of high girth:

\begin{figure}
  \centering
  \begin{tikzpicture}[
  text height=1.3ex,text depth=0ex,xscale=1.4,yscale=2,
  mispos/.append style={red, font=\scriptsize},
  newfact/.append style={red, dashed, -{>[scale=1.5]}}
]
  \node (a) at (-4.5, -1) {$a$};
  \node (b) at (-3.5, -1) {$b$};

  \path (a) edge[->,bend left=15] node[above,midway] {$R$} (b);
  \path (b) edge[->,bend left=15,newfact] node[below,midway] {$S$} (a);

  \node (a0) at (-2, -1) {$a_0$};
  \node (b0) at (-1, -1) {$b_0$};

  \node (a1) at (0, -1) {$a_1$};
  \node (b1) at (1, -1) {$b_1$};

  \node (a2) at (2, -1) {$a_2$};
  \node (b2) at (3, -1) {$b_2$};

  \path (a0) edge[->,bend left=15] node[above,midway] {$R$} (b0);
  \path (a1) edge[->,bend left=15] node[above,midway] {$R$} (b1);
  \path (a2) edge[->,bend left=15] node[above,midway] {$R$} (b2);

  \path (b0) edge[->,bend right=15,newfact] node[below,midway] {$S$} (a1);
  \path (b1) edge[->,bend right=15,newfact] node[below,midway] {$S$} (a2);
  \path (b2) edge[->,bend left=18,newfact] node[above,midway] {$S$} (a0);
\end{tikzpicture}
  \caption{Product with a group of large girth (see Example~\ref{exa:product})}
  \label{fig:product}
\end{figure}
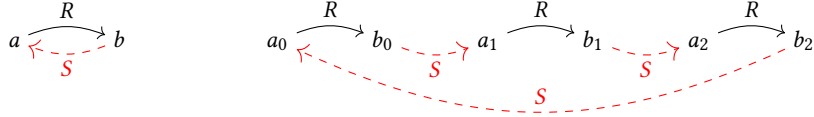

\begin{example}
  \label{exa:product}
  Consider $F_0 \defeq R(a, b)$ and $I_0 \defeq \{F_0\}$, illustrated in solid black
  in the left part of Figure~\ref{fig:product}. Consider $\idsb$ consisting of $\tau: \ui{R^2}{S^1}$, $\tau':
  \ui{S^2}{R^1}$, $\tau^{-1}$, and $(\tau')^{-1}$. Let $F \defeq S(b, a)$, and
  $I \defeq I_0 \sqcup \{F\}$, where $F$ is a red dashed edge in the drawing. $I$ satisfies~$\idsb$ and is sound for $\acq$, but not for
  $\cq$: take for instance $q : \exists x y ~ R(x, y) \wedge S(y, x)$, which is
  cyclic and holds in $I$ while $(I_0, \idsb) \not\entunr q$.

  We have $\lab{I} = \{\ll^F_1, \ll^F_2\}$. Identify $\ll^F_1$ and $\ll^F_2$
  to~$1$ and $2$ and
  consider the group $G \defeq (\{0, 1, 2\}, +)$ where $+$ is
  addition modulo~3. The group $G$ has girth 2 under $\lab{I}$.

  The product $I_\p \defeq (I, I_0) \sprod G$, writing pairs as subscripts for brevity, is
  $\{R(a_0, b_0), \allowbreak R(a_1, b_1),\allowbreak  R(a_2, b_2),\allowbreak
    S(b_1, a_2),\allowbreak  S(b_2, a_0),\allowbreak  S(b_0,
  a_1)\}$. The right part of Figure~\ref{fig:product} represents $I_\p$. Here, $I_\p$ happens to be \mbox{$5$-sound} for $\cq$.
\end{example}

We cannot directly use the simple product for our purposes, however, because
$I_\p \defeq (I_\f, I_0) \sprod G$ may violate $\ufds$ even though our instance
$I_\f$ satisfies $\fds$.
Indeed, there may be a relation~$R$,
a $\ufd$ $\phi: R^p \rightarrow R^q$ in~$\ufds$,
and two $R$-facts $F$ and $F'$ in~$I_\f \backslash I_0$
with $\pi_{R^p, R^q}(F) = \pi_{R^p, R^q}(F')$.
In $I_\p$ there will be images of~$F$ and $F'$ that overlap only on~$R^p$,
so they will violate~$\phi$.

Nevertheless, in the remainder of this \subsecname we prove 
the Simple Product Lemma (Lemma~\ref{lem:prodppty}),
as it will be useful for our purposes later.
Remember that a \emph{match} of a $\cq$ in an instance is witnessed by a
homomorphism~$h$, and that we also call the \emph{match} the image of~$h$.
We start by proving an easy lemma:

\begin{lemma}
  \label{lem:diffact}
  For any $\cq$ $q$ and instance $I$, if
  $I \models q$ with a witnessing homomorphism $h$ that maps two different atoms
  of~$q$ to the same fact, then there is a $\cq$ $q'$ such that:
  \begin{itemize}
    \item $\card{q'} < \card{q}$
    \item $q'$ \deft{entails} $q$, meaning that for any instance $I$, if $I \models q'$
      then $I \models q$
   \item $I \models q'$
  \end{itemize}
\end{lemma}

\begin{proof}
  Fix $q$, $I$, $h$, and let $A = R(\mybf{x})$ and $A' = R(\mybf{y})$ be the two
  atoms of~$q$ mapped to the same fact $F$ by $h$.
  Necessarily $A$ and $A'$ are atoms for the same relation $R$ of the fact $F$,
  and $h(A) = h(A')$ means that $h(x_i) = h(y_i)$ for all $R^i \in \pos(R)$.

  Let $\dom(q)$ be the set of variables occurring in $q$. Consider the
  map $f$ from $\dom(q)$ to $\dom(q)$ defined by $f(y_i) = x_i$ for all
  $i$, and $f(z) = z$ if $z$ does not occur in~$A'$. Observe that this ensures
  that $h(z) = h(f(z))$ for all $z \in \dom(q)$. Let $q' = f(q)$ be the query
  obtained by replacing every variable $z$ in $q$ by $f(z)$, and, as $f(A') =
  f(A)$, removing one of those duplicate atoms so that $\card{q'} < \card{q}$.
  We claim that $h' \defeq \restr{h}{\dom(q')}$ is a match of~$q'$ in $I$.
  Indeed, observe that any atom $f(A'')$
  of~$q'$ is homomorphically mapped by $h'$ to $h(A'')$ because $h'(f(z)) =
  h(z)$ for all $z$ so $h'(f(A'')) = h(A'')$.

  To see why $q'$ entails $q$, observe that $f$ defines a homomorphism from
  $q$ to $q'$, so that, for any instance $I'$, if $q'$ has a match $h''$ in~$I'$, then $h''
  \circ f$ is a match of~$q$ in~$I'$.
\end{proof}

Let us now prove the Simple Product Lemma (Lemma~\ref{lem:prodppty}).
Fix the constraints $\con$ and
the superinstance $I$ of the individualizing~$I_0$ such that $I$ is
$((\arity{\sigma}+1)\cdot k)$-sound for $\acq$, $I_0$, and $\con$.
Fix the $(2k+1)$-acyclic group~$G$ generated by~$\lab{I}$.
Consider $I_\p \defeq (I, I_0) \sprod G$, which is a
superinstance of~$I_0$, up to our identification of~$(a, e)$ to $a$ for $a \in
\dom(I_0)$, where $e$ is the neutral element of~$G$. We must show that $I_\p$ is
$k$-sound for~$\cq$, $I_0$, and~$\con$.

We call a match $h$ of a $\cq$ $q$ in $I_\p$ \deft{pure-instance-cyclic} if
every atom containing two occurrences of the same variable is mapped by~$h$ to a
fact of $I_0 \times G$, and
every Berge cycle of~$q$ contains an atom mapped by $h$ to a fact of~$I_0 \times G$.
In particular, if $q$ is in $\acq$ then any match $h$ of~$q$ in $I_\p$ is
vacuously pure-instance-cyclic. Our proof
consists of two claims:
\begin{enumerate}
  \item If a $\cq$ $q$ with $\card{q} \leq k$ has a pure-instance-cyclic match
    $h$ in $I_\p$, then $\chase{I_0}{\idsb} \models q$.
  \item If a $\cq$ $q$ with $\card{q} \leq k$ has a match $h$ in $I_\p$ which is
    not pure-instance-cyclic, then there is a $\cq$ $q'$ with $\card{q'} <
    \card{q}$ such that $q'$ entails $q$ and $q'$ has a match in $I_\p$.
\end{enumerate}
The fact that $I_\p$ is $k$-sound for~$\cq$ clearly follows from the two claims:
if a $\cq$~$q$ with $\card{q} \leq k$ has a match in $I_\p$, then we can apply the
second claim repeatedly until we obtain a $\cq$ $q'$ with $\card{q'} < \card{q}
\leq k$, $q'$ entails $q$, and $q'$ has a pure-instance-cyclic match in $I_\p$:
this must eventually occur because the empty query is in $\acq$.
Then use the first claim to deduce that $\chase{I_0}{\idsb} \models q'$, where
it follows that $\chase{I_0}{\idsb} \models q$. So it suffices to prove 
these two claims.

\medskip

We start by proving the first claim.
Let $q$ be a $\cq$ with $\card{q} \leq k$ that has a pure-instance-cyclic match
$h$ in $I_\p$.

We partition the atoms of~$q$ between the atoms $\calA$ matched by $h$ to $I_0
\times G$ and the atoms $\calA'$ which are not: we can then write $q$ as $\exists
\mybf{x} ~ \calA(\mybf{x}) \wedge \calA'(\mybf{x})$. Let $\calA''$ consist of the atom
$P_a(z)$ for each variable $z$ occurring in $\calA'$ which is mapped by $h$ to
an element $a \in \dom(I_0 \times G)$, and let $q'$ be the query $\exists
\mybf{x} ~ \calA'(\mybf{x}) \wedge \calA''(\mybf{x})$. 
As $I_0$ is individualizing,
it is immediate that $h$ is a match
of~$q'$ in~$I_\p$.

We first claim that $q'$ is in $\acq$. Indeed, no Berge cycle in $q'$
can use the atoms of~$\calA''$ as they are unary, and 
for the same reason no atom in~$\calA''$ contains two occurrences of the same
variable. Further,
$\calA'$
does not contain any Berge cycle or atom with two occurrences of the same
variable,
by definition of~$h$ being pure-instance-cyclic. Hence, $q'$ is indeed
in~$\acq$.
Further, we have
$\card{q'} \leq k \cdot (\arity{\sigma}+1)$, as $\card{\calA''} \leq
\arity{\sigma} \cdot \card{\calA'}$ and we have $\card{\calA'} \leq \card{q}
\leq k$, so that $\card{q'} \leq k \cdot (\arity{\sigma}+1)$.
Now, we know that $I
\models q'$, as evidenced by the homomorphism $\pr$ from $I_{\p}$ to $I$ 
defined by $\pr: (a, g) \mapsto a$ for
all $a \in \dom(I)$ and $g \in G$.
As $I$ is $(k \cdot (\arity{\sigma}+1))$-sound for $\acq$, and $q'$ is an $\acq$
query that holds in $I$ with $\card{q'} \leq k \cdot (\arity{\sigma}+1)$, we
know that $\chase{I_0}{\idsb} \models q'$.

Now, as $\calA''$ covers all
variables of~$q'$, by definition of~$I_0$ being individualizing, the only
possible match of~$q'$ in the chase is the one that maps each variable $z$ to
the $a \in \dom(I_0)$ such that the atom $P_a(z)$ is in $\calA''$. Further, as
$h$ matched $\calA$ to facts of~$I_0$ such that $h(z) = a$ where $P_a(z)$ occurs
in $\calA''$, we can clearly extend the match of~$q'$ in $\chase{I_0}{\idsb}$ to
a match of~$q$ in $\chase{I_0}{\idsb}$. This concludes the proof of the first
claim.

\medskip

We now prove the second claim.
Let $q$ be a $\cq$ with $\card{q} \leq k$ that has a match $h$ 
in $I_\p$ which is not pure-instance-cyclic.
Consider a Berge cycle $C$ of~$q$, of the form $A_1, x_1, A_2, x_2, \ldots,
A_n, x_n$, where the $A_i$ are pairwise distinct atoms and the $x_i$ pairwise
distinct variables,
where the $A_i$ are mapped by $h$ to facts not in $I_0 \times G$,
and where for all $1 \leq i \leq n$, variable $x_i$ occurs at position $q_i$
of atom $A_i$ and position $p_{i+1}$ of~$A_{i+1}$, with addition modulo $n
\defeq \card{C}$. We assume without loss of generality that $p_i \neq q_i$ for
all $i$. However, we do not assume that $n \geq 2$: either $n \geq 2$ and $C$ is
really a Berge cycle according to our previous definition, or $n = 1$ and
variable $x_1$ occurs in atom $A_1$ at positions
$p_1 \neq q_1$, which corresponds to the case where there are
multiple occurrences of the same variable in an atom.

For $1 \leq i \leq n$,
we write $F_i = R_i(\mybf{a}^{\mybf{i}})$ the image of~$A_i$ by $h$ in $I_\p$;
by definition of~$I_\p$, as $F_i$ is not a fact of $I_0 \times G$, there
is a fact $F'_i = R_i(\mybf{b}^{\mybf{i}})$ of~$I$ and $g_i \in G$ such that $a^i_j = (b^i_j,
g_i \cdot \ll^{F'_i}_j)$ for $R_i^j \in \pos(R_i)$.
Now, for all $1 \leq i \leq n$, 
as $h(x_i) = a^i_{q_i} = a^{i+1}_{p_i+1}$ for all $1 \leq i \leq n$, we deduce
by projecting on the second component that $g_i \cdot \ll^{F'_i}_{q_i} = g_{i+1}
\cdot \ll^{F'_{i+1}}_{p_{i+1}}$, so that, by collapsing the equations of the cycle
together, $\ll^{F'_1}_{q_1} \cdot
(\ll^{F'_{2}}_{p_{2}})^{-1} \cdot \cdots \cdot \ll^{F'_{n-1}}_{q_{n-1}} \cdot
(\ll^{F'_{n}}_{p_{n}})^{-1} \cdot \ll^{F'_{n}}_{q_{n}} \cdot
(\ll^{F'_{1}}_{p_{1}})^{-1} = e$.

As the girth of~$G$ under $\lab{I}$ is $\geq 2k+1$, and this product contains
$2n \leq 2k$ elements, we must have either $\ll^{F'_i}_{q_i} =
\ll^{F'_{i+1}}_{p_{i+1}}$ for some $i$, or $\ll^{F'_i}_{p_i} = \ll^{F'_i}_{q_i}$ for
some $i$. The second case is impossible because we assumed that $p_i \neq q_i$
for all $1 \leq i \leq n$. Hence, necessarily $\ll^{F'_i}_{q_i} = \ll^{F'_{i+1}}_{p_{i+1}}$, so in
particular we must have $n > 1$ and $F'_i = F'_{i+1}$. Hence the atoms $A_i \neq A_{i+1}$ of~$q$ are
mapped by $h$ to the same fact $F'_i = F'_{i+1}$. We conclude by
Lemma~\ref{lem:diffact} that there is a strictly smaller $q'$ which entails $q$
and has a match in $I_\p$, which is what we wanted to show. This concludes the
proof of the second claim, and of the Simple Product Lemma (Lemma~\ref{lem:prodppty}).

\subsection{Cautiousness}
\label{sec:cautiousness}

As the simple product may cause $\fd$ violations, we
will define a more refined notion of product, which intuitively does
not attempt to blow up cycles within fact overlaps.
In order to clarify this, however, we will first need to study in more detail the
instance $I_\f$ to which we will apply the process, namely,
the one that we constructed to prove Theorem~\ref{thm:hfds}.
We will consider a \defo{quotient} of~$I_\f$:

\begin{definition}
  The \deft{quotient} $\quot{I}{\sim}$ of an instance $I$ by an equivalence
  relation $\sim$ on $\dom(I)$ is defined as follows:
  \begin{itemize}
  \item $\dom(\quot{I}{\sim})$ is the equivalence classes of~$\sim$ on~$\dom(I)$,
  \item $\quot{I}{\sim}$ contains one fact $R(\mybf{A})$
    for every fact $R(\mybf{a})$ of~$I$,
    where $A_i$ is the $\sim$-class of~$a_i$ for all $R^i \in \pos(R)$.
  \end{itemize}
  The \deft{quotient homomorphism} $\chi_{\sim}$ is the homomorphism from~$I$ to
  $\quot{I}{\sim}$ defined by mapping each element of~$\dom(I)$ to its
  $\sim$-class.
\end{definition}

We quotient $I_{\f}$ by the equivalence relation $\bbsim_k$ (recall
Definition~\ref{def:bbsim}). The result may no longer satisfy $\con$. However, it is still $k$-sound for $\acq$, for the following reason:

\begin{lemma}
  \label{lem:ksimhomom}
  Any $k$-bounded simulation from an instance~$I$ to an instance $I'$ defines a
  $k$-bounded simulation from $\quot{I}{\bbsim_k}$ to~$I'$.
\end{lemma}

\begin{proof}
Fix the instance $I$ and the $k$-bounded simulation $\cov$ to an instance $I'$, and
consider $I'' \defeq \quot{I}{\bbsim_k}$. We show that there is a
$k$-bounded simulation $\cov'$ from $I''$ to $I$, because $\cov
\circ \cov'$ would then be a $k$-bounded simulation from $I''$ to $I'$, the
desired claim. We define $\cov'(A)$ for all $A \in I''$ to be $a$ for any member
$a \in A$ of the equivalence class~$A$ in~$I$, and show that $\cov'$ thus defined is
indeed a $k$-bounded simulation.

We will show the stronger result that $(I'', A) \bsim_{k} (I, a)$ for all $A
\in \dom(I'')$ and for any $a \in
A$. We do it by proving, by induction on~$0 \leq k' \leq k$,
that $(I'', A) \bsim_{k'} (I, a)$ for all $A \in \dom(I'')$ and $a \in
A$.
The case $k' = 0$ is trivial.
Hence, fix $0 < k' \leq k$, assume that $(I'', A) \bsim_{k'-1} (I, a)$ for all $A
\in \dom(I'')$ and $a \in A$, and show that this is also true for $k'$. Choose
$A \in \dom(I'')$, $a \in A$, we must show that $(I'', A) \bsim_{k'} (I, a)$. To do so,
consider any fact $F = R(\mybf{A})$ of~$I''$ such that $A_p = A$ for
some $R^p \in \pos(R)$. Let $F' = R(\mybf{a}')$ be a fact of~$I$ that is a preimage
of~$F$ by $\chi_{\bbsim_k}$, so that $a'_q \in A_q$ for all $R^q \in \pos(R)$.
We have $a'_p \in A$ and $a \in A$, so that $a'_p \bbsim_k a$
holds in $I$. Hence, in particular we have $(I, a'_p) \bsim_{k'} (I, a)$ because
$k' \leq k$, so there
exists a fact $F'' = R(\mybf{a}'')$ of~$I$ such that $a''_p = a$ and $(I, a'_q)
\bsim_{k'-1} (I, a''_q)$ for all $R^q \in \pos(R)$. We show that $F''$ is a
witness fact for $F$. Indeed, we have $a''_p = a$. Let us now choose $R^q \in
\pos(R)$ and show that $(I'', A_q) \bsim_{k'-1} (I, a''_q)$. By induction
hypothesis, as $a'_q \in A_q$, we have $(I'', A_q) \bsim_{k'-1} (I, a'_q)$, and as
$(I, a'_q) \bsim_{k'-1} (I, a''_q)$, by transitivity we have indeed $(I'', A_q)
\bsim_{k'-1} (I, a''_q)$. Hence, we have shown that $(I'', A) \bsim_{k'} (I, a)$.

By induction, we conclude that $(I'', A) \bsim_{k} (I, a)$ for all $A
\in \dom(I'')$ and $a \in A$, so that there is indeed a $k$-bounded simulation
from $I''$ to $I$, which, as we have explained, implies the desired claim.
\end{proof}

Let us thus consider $I_{\f}' \defeq
\quot{I_{\f}}{\bbsim_k}$ which is still $k$-sound for $\acq$
by the previous lemma, and
consider the homomorphism $\chi_{\bbsim_k}$ from $I_\f$ to~$I'_\f$. Our idea is
to blow up cycles in $I_\f$ by a \defo{mixed product} that only distinguishes
facts that have a different image in $I'_\f$ by $\chi_{\bbsim_k}$.
This is sufficient to lift $k$-soundness from $\acq$ to $\cq$, and it will not
create $\fd$ violations on facts that have the same image by~$\chi_{\bbsim k}$.
Crucially, however, we can show from our construction that all overlapping facts
of~$I_\f$ have the same image by~$\chi_{\bbsim k}$.
Let us formalize this condition:

\begin{definition}
  \label{def:cautious}
  Let $I$ be an instance, let $I_1 \subseteq I$, and let $f$ be any mapping with
  domain~$I$.
  We say $I$ is \deft{cautious} for~$f$ and~$I_1$
  if for any two \deft{overlapping} facts, namely, two facts
  $F = R(\mybf{a})$ and $F' = R(\mybf{b})$ 
  of the same relation with $a_p =
  b_p$ for some 
  $R^p \in \pos(R)$, one of the following holds: $F, F' \in I_1$, or $f(a_p) =
  f(b_p)$ for all $R^p \in \pos(R)$.
\end{definition}

We conclude the \subsecname by presenting a strengthening of
Theorem~\ref{thm:hfds}. This is the only point in this \secname where we rely on
the details of the process of the previous \secnames:

\begin{theorem}[Cautious models]
  \label{thm:cautiousmodels}
  For any finitely closed $\con$ formed of $\uid$s $\idsb$ and $\fd$s $\fds$, instance $I_0$, and
  $k \in \NN$,
  we can build 
  a finite superinstance $I_{\f}$ of an instance $I_1$ such that:
  \begin{itemize}
    \item $I_1$ is an
      individualizing superinstance of a disjoint union of copies of~$I_0$;
    \item $I_\f$ satisfies $\con$;
    \item $I_\f$ is $k$-sound for $\con$, $\acq$, and $I_1$;
    \item $I_\f$ is cautious for $\chi_{\bbsim_k}$ and $I_1$.
  \end{itemize}
\end{theorem}

We will use the Cautious Models Theorem in the next \subsecname. For now, let us
show how to prove it.
Fix $\con$, $I_0$, and $k \in \NN$.
Let $I_{0,\ii}$ be an individualizing superinstance of $I_0$, and 
apply $k$ $\uid$ chase rounds with the $\uid$s of $\idsb$ to $I_{0,\ii}$ to obtain $I_{0,\ii}'$.
Apply the Sufficiently Envelope-Saturated Proposition (Proposition~\ref{prp:preproc})
to~$I_{0,\ii}'$ to obtain an
aligned superinstance $J$ of a disjoint union $I_{0,\ii}''$ of copies of
$I_{0,\ii}'$.
Now, modify $J$ to $J'$ and $I_{0,\ii}''$ to $I_1'$ by replacing the copies of the facts of $I_{0,\ii} \backslash I_0$
by new individualizing facts (i.e., make the individualizing facts unique across
copies of $I_{0,\ii}'$).
This ensures by definition that $I_1'$ is the result of applying $k$ $\uid$
chase rounds to an individualizing superinstance of a
disjoint union of copies of~$I_0$.
Further, the modification to~$J'$ can be done so as to ensure that $J'$ is 
an aligned superinstance of~$I_1'$;
the $k$ chase rounds applied when defining~$I_{0,\ii}'$ ensure that
the $\cov$ mapping can still be defined notwithstanding the change in the
individualizing facts.
Further, we have $\card{J'} = \card{J}$, so
$J'$ is still sufficiently envelope-saturated.

We now apply the Envelope-Thrifty Completion Proposition (Proposition~\ref{prp:etcomp})
to the aligned
superinstance $J'$ of~$I_1'$ to obtain
a superinstance $J_\f$ of~$I_1'$ which is $k$-sound for $\con$, $\acq$, and
$I_1'$, and that satisfies
$\con$. Now, define $I_1$ from $I_1'$ by removing the facts created in the $k$
$\uid$ chase rounds, so it is by definition an individualizing superinstance of
a disjoint union of copies of~$I_0$. As $I_1'$ is the result of applying chase rounds to~$I_1$,
$I_\f$ is also $k$-sound for $\con$, $\acq$ and $I_1$. Hence, $I_\f$ satisfies the first three conditions that we have to show
in the Cautious Models Theorem (Theorem~\ref{thm:cautiousmodels}). The only thing left is to show the last one,
namely:

\begin{lemma}[Cautiousness]
  \label{lem:cautious}
  $I_\f$ is cautious for $\chi_{\bbsim_k}$ and~$I_1$.
\end{lemma}

We show the Cautiousness Lemma in the rest of the \subsecname, which concludes
the proof of the Cautious Models Theorem.

We
first show that overlapping facts in $J_{\f} = (I_\f, \cov_\f)$ are cautious
for the $\cov$ mapping that we construct, in terms of $\bbsim_k$-classes.
Formally, let
$I_{\cc} \defeq \chase{I_1}{\idsb}$, and let $\chi_{\bbsim_k}'$ be the
homomorphism from $I_\cc$ to $\quot{I_\cc}{\bbsim_k}$. We claim:

\begin{lemma}
  \label{lem:covcautious}
  $I_\f$ is cautious for $\chi_{\bbsim_k}' \circ \cov$ and~$I_1$.
\end{lemma}

In other words, whenever two facts $F = R(\mybf{a})$ and $F' = R(\mybf{b})$
have non-empty overlap in $I_\f$ and are not both in $I_1$, then, for any position $R^p \in
\pos(R)$, we have $\cov(a_p) \bbsim_k \cov(b_p)$ in $I_\cc$.

\begin{proof}
  We first check that this claim holds on the result~$J'$ of the
  Sufficiently Envelope-Saturated Proposition (Proposition~\ref{prp:preproc}),
  with our modifications to the individualizing facts.
  $J'$ is a disjoint union of instances $J_D$ for each fact class $D \in
  \afactcl$. If $D$ is safe, no facts overlap in $J_D$ except possibly fact
  pairs in the copy
  of~$I_0$, hence, in~$I_1$. For unsafe $D$, in Lemma~\ref{lem:oneenv},
  the only facts with non-empty overlap in $J_D$
  are fact pairs in some copy of $I_0$, hence in~$I_1$, or they are the facts
  $f'(F_i)$,
  which all map to $\bbsim_k$-equivalent $\cov$-images by construction. So the
  claim holds on $J'$.

  Second, it suffices to show that the claim is preserved by envelope-thrifty
  chase steps.
  By their definition, whenever we create a new fact $F_{\fnew}$ for a fact
class $D$, the only
  elements of $F_{\fnew}$ that can be part of an overlap between $F_{\fnew}$ and an
existing fact are envelope elements, appearing at the one position at which
they appear in $\calE(D)$.
Then, by condition~4 of the definition of envelopes
(Definition~\ref{def:envelope}), we deduce that the two overlapping facts
achieve the same fact class.
\end{proof}

Returning to the proof of the Cautiousness Lemma (Lemma~\ref{lem:cautious}),
we now show that two elements in $J_\f$ having $\bbsim_k$-equivalent
$\cov$ images in $I_\cc$
must themselves be $\bbsim_k$-equivalent in~$J_\f$. We do it by showing that, in
fact, for any $a \in \dom(I_\f)$, not only do we have
$(I_\f, a) \bsim_k (I_\cc, \cov(a))$, as required by the $k$-bounded simulation
$\cov$, but we also have the reverse:
$(I_\cc, \cov(a)) \bsim_k (I_\f, a)$; in fact, we even have a
\emph{homomorphism} from
$I_\cc$ to $I_\f$ that maps $\cov(a)$ to~$a$.
The existence of this homomorphism is thanks to our
specific definition of~$\cov$, and on the directionality condition of aligned
superinstances; further, it only holds for the final result $I_\f$,
which satisfies $\ids$; it is not respected at intermediate steps of the process.

To prove this, and conclude the proof of the Cautiousness Lemma,
remember the forest structure on the $\uid$ chase 
(Definition~\ref{def:chaserev}).
We define the \deft{ancestry} $\calA_F$ of a fact $F$ in $I_\cc$ as
$I_1$ plus the facts of the path in the chase forest that leads to $F$;
if $F \in I_1$ then $\calA_F$ is just $I_1$.
The \deft{ancestry} $\calA_a$ of~$a \in
\dom(I_\cc)$ is that of the fact where $a$ was introduced. 

We now claim the following lemma about $J_\f$, which relies on the
directionality condition:

\begin{lemma}
  \label{lem:ancesthom}
  For any $a \in \dom(I_\f)$, there is a homomorphism $h_a$ from $\calA_{\cov(a)}$
  to~$I_\f$ such that $h_a(\cov(a)) = a$.
\end{lemma}

\begin{proof}
  We prove that this property holds on $I_\f$, by first showing that it is true
  of~$J'$ constructed by our modification of 
  the Sufficiently Envelope-Saturated Solutions Proposition (Proposition~\ref{prp:preproc}).
  This is clearly the case because the instances created by
  Lemma~\ref{lem:oneenv} are just truncations of the chase where some elements
  are identified at the last level.

  Second, we show that the property is maintained by envelope-thrifty steps;
  in fact, by any thrifty chase steps (Definition~\ref{def:ft2})
  Consider a thrifty chase step where,
  in a state $J_1 = (I_1, \cov_1)$
  of the construction of our aligned superinstance,
  we apply a $\uid$ $\tau: \ui{R^p}{S^q}$ to a fact $F_{\factive} = R(\mybf{a})$
  to create a fact $F_{\fnew} = S(\mybf{b})$
  and obtain the aligned superinstance $J_2 = (I_2, \cov_2)$.
  Consider the chase witness $F_\w = S(\mybf{b}')$. By
  Lemma~\ref{lem:wadef}, $b'_q$ is the exported element
  between $F_\w$ and its parent in $\chase{I_0}{\idsb}$.
  So we know that for any $S^r \neq S^q$,
  we have $\calA_{b'_r} = \calA_{b'_q} \sqcup \{F_\w\}$.

  We must build the desired homomorphism $h_a$ for all $a
  \in \dom(I_2) \backslash \dom(I_1)$.
  Indeed, for $a \in \dom(I_1)$, by hypothesis on~$I_1$, there is a homomorphism $h_a$ from
  $\calA_{\cov_1(a)}$ to $I_1$ with $h_a(\cov_1(a)) = a$, and as $\cov_2(a) =
  \cov_1(a)$, we can use $h_a$ as the desired homomorphism from
  $\calA_{\cov_2(a)}$ to~$I_2$.
  So let us pick $b \in \dom(I_2) \backslash \dom(I_1)$ and construct $h_b$. By
  construction of~$I_2$, $b$ must occur in the new fact $F_{\fnew}$; further, by
  definition of thrifty chase steps, we have defined $\cov_2(b) \defeq b'_r$ for
  some $S^r$ where $b_r = b$.
  Now, as $a_p = b_q$ is in $\dom(I_1)$,
  we know that there is a homomorphism $h_{b_q}$
  from~$\calA_{\cov(b_q)} = \calA_{b'_q}$ to~$I_1$
  such that we have $h_{b_q}(b'_q) = b_q$.
  We extend $h_{b_q}$ to the homomorphism $h_{b}$
  from~$\calA_{b'_r} = \calA_{b'_q} \sqcup \{F_\w\}$ to~$I_2$
  such that $h_{b}(b'_r) = b$,
  by setting $h_{b}(F_\w) \defeq F_{\fnew}$
  and $h_{b}(F) \defeq h(F)$ for any other~$F$ of~$\calA_{b'_r}$;
  we can do this because, by definition of the $\uid$ chase,
  $F_\w$ shares no element with the other facts of~$\calA_{b'_r}$
  (that is, with $\calA_{b'_q}$),
  except $b'_q$ for which our definition coincides
  with the existing image of~$b'_q$ by~$h_{b_q}$.
  This proves the claim.
\end{proof}

This allows us to deduce the following, which is specific to~$J_\f$, and relates
to the universality of the chase~$I_\cc$:

\begin{corollary}
  \label{cor:covrevb}
  For any $a \in \dom(I_\f)$, there is a homomorphism $h_a$ from
  $I_\cc$ to $I_\f$ such that $h_a(\cov(a)) = a$. 
\end{corollary}

\begin{proof}
  Choose $a \in \dom(I_\f)$ and let us construct $h_a$.
Let $h'_a$ be the homomorphism
from $\calA_{\cov(a)}$ to $I_\f$ with $h'_a(\cov(a)) = a$ whose existence was
proved in Lemma~\ref{lem:ancesthom}. Now start by setting $h_a \defeq h'_a$, and
extend $h'_a$ to be the desired homomorphism,  fact by fact,
using the property that $I_\f \models \idsb$: for any $b \in
\dom(I_\cc)$ not in the domain of~$h'_a$ but which was
introduced in a fact $F$ whose exported element $c$ is in
the current domain of~$h'_a$, let us extend $h'_a$ to the elements of~$F$ in the
following way: consider the parent fact $F'$ of~$F$ in~$I_\cc$ and its match
by $h'_a$ in~$I_\f$, let $\tau$ be the $\uid$ used to create $F'$ from $F$, and
$c' \in \dom(I_\cc)$ be the exported element between~$F$ and~$F'$ (so $h'_a(c')$
is defined). We know that $c \defeq h'_a(c')$ occurs in~$I_\f$ at all positions
where $c'$ occurs in~$I_\cc$. Hence, because
$I_\f \models \tau$, there must be a suitable fact $F''$ in~$I_\f$ to extend $h'_a$ to all
elements of~$F$ by setting $h'_a(F) \defeq F''$,
which is consistent with the image of~$c$ previously defined in~$h'_a$.
The (generally infinite) result of this process is the desired homomorphism~$h_a$.
\end{proof}

We are now ready to show our desired claim:

\begin{lemma}
  \label{lem:covbbsim}
  For any $a, b \in \dom(I_\f)$, if $\cov(a)
\bbsim_k \cov(b)$ in $I_\cc$, then $a \bbsim_k b$ in~$I_\f$.
\end{lemma}

\begin{proof}
  Fix $a, b \in \dom(I_\f)$.
  We have $(I_\f, a) \bsim_k (I_\cc, \cov(a))$
  because $\cov$ is a $k$-bounded simulation;
  we have $(I_\cc, \cov(a)) \bsim_k (I_\cc, \cov(b))$
  because $\cov(a) \bbsim_k \cov(b)$;
  and we have $(I_\cc, \cov(b)) \bsim_k (I_\f, b)$
  by Corollary~\ref{cor:covrevb} as witnessed by~$h_b$.
  By transitivity, we have $(I_\f, a) \bsim_k (I_\f, b)$.
  The other direction is symmetric,
  so the desired claim follows.
\end{proof}

The Cautiousness Lemma (Lemma~\ref{lem:cautious}) follows immediately from
Lemma~\ref{lem:covcautious} and Lemma~\ref{lem:covbbsim}. This concludes the
proof of the Cautious Models Theorem (Theorem~\ref{thm:cautiousmodels}).

\subsection{Mixed Product}
\label{sec:mixedprod}

Using the Cautious Models Theorem, we now define the notion of mixed product,
which uses the same fact label for facts with the same image by~$h \defeq
\chi_{\bbsim_k}$:

\begin{definition}
  \label{def:prodm}
  Let $I$ be a finite superinstance of~$I_1$
  with a homomorphism $h$ to another finite superinstance $I'$ of~$I_1$
  such that $\restr{h}{\dom(I_1)}$ is the identity and $\restr{h}{\dom(I\backslash I_1)}$
  maps to~$I' \backslash I_1$.
  Let $G$ be a finite group generated by $\lab{I'}$.

  The \deft{mixed product} of~$I$ by $G$ via $h$ preserving $I_1$,
  written $(I, I_1) \mprod^h G$,
  is the finite superinstance of~$I_1$ with domain $\dom(I) \times G$
  consisting of the following facts, for every $g \in G$:
  \begin{itemize}
  \item For every fact $R(\mybf{a})$ of~$I_1$, the fact $R((a_1, g), \ldots,
    (a_{\card{R}}, g))$.
  \item For every fact $F = R(\mybf{a})$ of~$I \backslash I_1$, the 
    fact
    $R((a_1, g \cdot \ll^{h(F)}_1), \ldots, (a_{\card{R}}, g \cdot
    \ll^{h(F)}_{\card{R}}))$. \qedef
  \end{itemize}
\end{definition}

We now show that the mixed product preserves $\uid$s and $\fd$s when
cautiousness is assumed.

\begin{lemma}[Mixed product preservation]
  \label{lem:mixedprod}
  For any $\uid$ or $\fd$ $\tau$, if $I \models \tau$
  and $I$ is cautious for~$h$, then $(I, I_1) \mprod^h G \models \tau$.
\end{lemma}

\begin{proof}
Write $I_\m \defeq (I, I_1) \mprod^h G$ and write $I'$ for the range of~$h$ as
before.

If $\tau$ is a $\uid$, the claim is immediate even without the cautiousness
hypothesis. (In fact, the analogous claim could even be proven for the simple
product.) Indeed, for any $a \in
\dom(I)$ and $R^p \in \positions(\sigma)$, if $a \in \pi_{R^p}(I)$ then
$(a, g) \in \pi_{R^p}(I_\m)$ for all $g \in G$; conversely, if
$a \notin \pi_{R^p}(I)$ then
$(a, g) \notin \pi_{R^p}(I_\m)$ for all $g \in G$. Hence,
letting $\tau : \ui{R^p}{S^q}$ be a $\uid$ of~$\idsb$, if there is $(a, g) \in
\dom(I_\m)$ such that $(a, g) \in \pi_{R^p}(I_\m)$ but $(a, g) \notin
\pi_{S^q}(I_\m)$ then $a \in \pi_{R^p}(I)$ but $a \notin \pi_{S^q}(I)$. Hence any
violation of~$\tau$ in $I_\m$ implies the existence of a violation of~$\tau$
in $I$, so we conclude because $I \models \tau$.

Assume now that $\tau$ is a $\fd$ $\phi : R^L \rightarrow R^r$.
Assume by contradiction that there are two facts
$F_1 = R(\mybf{a})$ and $F_2 = R(\mybf{b})$ in~$I_\m$
that violate~$\phi$, i.e., we have $a_l = b_l$ for all $l \in L$,
but $a_r \neq b_r$.
Write $a_i = (v_i, f_i)$ and $b_i = (w_i, g_i)$ for all $R^i \in \pos(R)$.
Consider $F_1' \defeq R(\mybf{v})$ and $F_2' \defeq R(\mybf{w})$ 
the facts of~$I$ that are the images of~$F_1$ and~$F_2$
by the homomorphism from~$I_\m$ to~$I$
that projects on the first component.
As $I \models \tau$, $F_1'$ and $F_2'$ cannot violate~$\phi$,
so as $v_l = w_l$ for all $l \in L$, we must have $v_r = w_r$.
Now, as $I$ is cautious for~$h$ and $F_1'$ and $F_2'$ overlap
(take any $R^{l_0} \in R^L$),
either $F_1', F_2' \in I_1$ or $h(F_1') = h(F_2')$.

In the first case,
by definition of the mixed product,
there are $f, g \in G$
such that $f_i = f$ and $g_i = g$ for all $R^i \in \pos(R)$.
Thus, taking any $l_0 \in L$, as we have $a_{l_0} = b_{l_0}$,
we have $f_{l_0} = g_{l_0}$, so $f = g$, which implies that $f_r = g_r$.
Hence, as $v_r = w_r$,
we have $(v_r, f_r) = (w_r, g_r)$,
contradicting the fact that $a_r \neq b_r$.

In the second case,
as $h$ is the identity on $I_1$
and maps $I \backslash I_1$ to $I' \backslash I_1$,
$h(F_1') = h(F_2')$ implies that
either $F_1'$ and $F_2'$ are both facts of~$I_1$
or they are both facts of $I \backslash I_1$;
but we have already excluded the former possibility in the first case,
so we assume the latter.
By definition of the mixed product,
there are $f, g \in G$ such that
$f_i = f \cdot \ll^{h(F'_1)}_i$ and $g_i = g \cdot \ll^{h(F'_2)}_i$
for all $R^i \in \pos(R)$.
Picking any $l_0 \in L$, from $a_{l_0} = b_{l_0}$,
we deduce that $f \cdot \ll^{h(F'_1)}_{l_0} = g \cdot \ll^{h(F'_2)}_{l_0}$;
as $h(F'_1) = h(F'_2)$, 
this simplifies to $f = g$. Hence, $f_r = g_r$
and we conclude like in the first case.
\end{proof}

Second, we show that $h: I \to I'$ lifts to a homomorphism from the mixed product to the
simple product, so we can rely on the result 
of the Simple Product Lemma (Lemma~\ref{lem:prodppty}).

\begin{lemma}[Mixed product homomorphism]
  \label{lem:mixedhom}
  There is a homomorphism from $(I, I_1) \mprod^h G$ to~$(I, I_1) \sprod G$.
\end{lemma}

\begin{proof}
We use the homomorphism $h: I \to I_1$
to define the homomorphism $h'$ from $I_\m \defeq (I, I_1) \mprod^h G$ to $I_\p
\defeq (I, I_1) \sprod G$
by $h'((a, g)) \defeq (h(a), g)$
for every $(a, g) \in \dom(I) \times G$.

Consider a fact $F = R(\mybf{a})$ of~$I_\m$, with $a_i
= (v_i, g_i)$ for all $R^i \in \pos(R)$. Consider its image $F' = R(\mybf{v})$
by the homomorphism from $I_\m$ to
$I$ obtained by projecting to the first component,
and the image $h(F')$ of~$F'$ by the homomorphism~$h$.
As $\restr{h}{\dom(I_1)}$ is the identity and $\restr{h}{\dom(I \backslash I_1)}$ maps
to $I_1 \backslash I_1$, $h(F')$ is a fact of~$I_1$ iff $F'$ is.
Now by definition of the simple product it is clear that $I_\p$ contains the
fact~$h'(F)$: it was created in~$I_\p$ from~$h(F')$
for the same choice of $g \in G$. This shows that $h'$ is indeed a homomorphism,
which concludes the proof.
\end{proof}

We can now conclude our proof of the Universal Models Theorem
(Theorem~\ref{thm:univexists}).
Let $I_1$ be the individualizing union of disjoint copies of~$I_0$ and
$I_\f$ be the superinstance of~$I_1$ 
given by the Cautious Models Theorem (Theorem~\ref{thm:cautiousmodels}) applied to
$k' \defeq k \cdot (\arity{\sigma} + 1)$.
As $I_1$ is individualizing, we know that each element of~$I_1$ is alone in its
$\bbsim_{k'}$-class in $I_\f$, so the restriction of $\quot{I_\f}{\bbsim_{k'}}$ to
$\chi_{\bbsim_{k'}}(I_1)$ is actually $I_1$ up to isomorphism; so we define 
$I_\f'$ to be $\quot{I_\f}{\bbsim_{k'}}$ modified by identifying
$\chi_{\bbsim_{k'}}(I_1)$ to $I_1$; it is a finite superinstance of~$I_1$.
Let $h$ be the homomorphism
from $I_\f$ to $I_\f'$ obtained by modifying $\chi_{\bbsim_{k'}}$ accordingly,
which ensures that $\restr{h}{\dom(I_1)}$ is the identity and $\restr{h}{\dom(I_\f
\backslash I_1)}$ maps to~$I_\f' \backslash I_1$.

Let $G$ be a $(2k+1)$-acyclic group generated by $\lab{I'_\f}$,
and consider $I_\p \defeq (I'_\f, I_1) \sprod G$.
As $I_\f$ was $k'$-sound for $\acq$, $I_1$ and~$\con$,
so is~$I'_\f$ by Lemma~\ref{lem:ksimhomom},
so, as $I_1$ is individualizing,
$I_\p$ is $k$-sound for~$\cq$, $I_1$ and~$\con$ 
by the Simple Product Lemma (Lemma~\ref{lem:prodppty}).
However, as we explained, it may be the case that 
 $I_\p \not\models \con$.
We therefore construct $I_\m \defeq (I_\f, I_1) \mprod^h G$.
By the Mixed Product Homomorphism Lemma (Lemma~\ref{lem:mixedhom}),
$I_\m$ has a homomorphism to $I_\p$,
so it is also $k$-sound for $\cq$, $I_1$ and~$\con$.
Now, as $I_1$ is an individualizing superinstance 
of a disjoint union of copies of~$I_0$,
and as the 
fresh relations in the individualizing superinstance~$I_1$ 
do not occur in queries or in constraints,
it is clear that $I_\m$ is also $k$-sound for $\cq$, $I_0$ and~$\con$.
Further, by the conditions ensured 
by the Cautious Models Theorem (Theorem~\ref{thm:cautiousmodels}),
$I_\f$ is cautious for $h$ and $I_1$.
So, by the Mixed Product Preservation Lemma (Lemma~\ref{lem:mixedprod}),
we have $I_\m \models \con$ because $I_\f \models \con$. 

Hence, the mixed product $I_\m$ is a finite $k$-universal instance for~$\cq$,
$I_0$ and~$\con$.
This concludes the proof of the Universal Models Theorem, and hence of our main
theorem (Theorem~\ref{thm:mymaintheorem}).

\section{Conclusion}
\label{sec:conclusion}

In this work we have developed the first techniques 
to build finite models on arbitrary arity schemas that satisfy both referential constraints and number
restrictions, while controlling which $\cq$s are satisfied. We
have used these techniques to prove that finite open-world query answering for $\cq$s,
$\uid$s and $\fd$s is finitely controllable up to finite closure of the
dependencies. This allowed us to isolate the complexity of FQA for $\uid$s and $\fd$s.

As presented the constructions are quite specific to dependencies, so we leave
as future work the question of 
constraint languages
featuring conjunction, disjunction, constants, and other such features. For
instance, one goal could be to generalize
to higher arity
the rich arity-2 constraint languages of, e.g.,
\cite{ibanezgarcia2014finite,pratt2009data}, while maintaining the decidability
of FQA.

\myparagraph{Acknowledgements}
This work was supported in part by the
Engineering and Physical Sciences Research Council, UK (EP/G004021/1) and the
French ANR NormAtis project. We are very grateful to Balder ten
Cate, Thomas Gogacz, Andreas Pieris, and Pierre Senellart for comments on earlier
drafts, and to the anonymous reviewers for their valuable feedback.

\bibliographystyle{ACM-Reference-Format-Journals}
\bibliography{lics-full}

\appendix

\section{Details about the $\uid$ Chase and Unique Witness Property}
\label{apx:chase}

In this appendix, we give details about the $\uid$ chase and the \emph{Unique
Witness Property}. Recall its statement:

\medskip 

For any element $a \in \dom(\chase{I}{\ids})$ and position $R^p$ of~$\sigma$,
if two facts of~$\chase{I}{\ids}$ contain~$a$ at position~$R^p$,
then they are both facts of~$I$.

\medskip

We first give an example showing why this may not be guaranteed by the first round of the
$\uid$ chase. Consider the instance $I = \{R(a), S(a)\}$ and the $\uid$s
$\tau_1: \ui{R^1}{T^1}$ and $\tau_2: \ui{S^1}{T^1}$, where $T$ is a binary
relation.
Applying a round of the $\uid$ chase creates the instance $\{R(a), S(a), T(a, b_1), T(a,
b_2)\}$, with $T(a, b_1)$ being created by applying $\tau_1$ to the active fact
$R(a)$, and $T(a, b_2)$ being created by applying $\tau_2$ to the active fact
$S(a)$.

By contrast, the core chase would create only one of these two facts, because it
would consider that two new facts are \defp{equivalent}: they have the same
exported element occurring at the same position. In general, the core chase
keeps only one fact within each class of facts that are equivalent in this sense.

However, after one chase round by the core chase, there is no longer any
distinction between the $\uid$ chase and the core chase, because the following property
holds on the result $I'$ of a chase round (be it by the core chase or by the $\uid$ chase)
on any instance $I''$: (*)~for any $\tau \in \ids$ and element $a \in \appelem{I'}{\tau}$, $a$ occurs in only one fact of~$I'$.
This is true because $\ids$ is transitively closed, so we know
that no $\uid$ of~$\ids$ is applicable to an element
of $\dom(I'')$ in $I'$; hence the only elements that witness violations occur
in the one fact where they were introduced in~$I'$.

We now claim that (*) implies that the Unique Witness Property holds when we
chase by the core chase for the first round and the $\uid$ chase for subsequent
rounds. Indeed, assume to
the contrary that $a \in \dom(\chase{I}{\ids})$ violates the Property, and that
two facts $F_1$ and $F_2$ contain~$a$ at some position~$R^p$.

If $a \in \dom(I)$, because $\ids$ is transitively closed, after the first chase
round on~$I$, we no longer create any fact that involves~$a$. Hence, each one of $F_1$ and
$F_2$ is either a fact of~$I$ or a fact created in the first round of the chase
(which is a chase round by the core chase). However, if one of $F_1$ and $F_2$
is in $I$, then it witnesses that we could not have $a \in \appelem{I}{R^p}$, so it
is not possible that the other fact was created in the first chase round. It
cannot be the case either that $F_1$ and $F_2$ were both created in the first
chase round, by definition of the core chase. Hence, $F_1$ and $F_2$ must 
both be facts of~$I$.

If $a \in \dom(\chase{I}{\ids}) \backslash \dom(I)$, assume that $a$ occurs at position $R^p$ in two
facts $F_1$, $F_2$. As $a \notin \dom(I)$, none of them is a fact of~$I$. We
then show a contradiction. It is not possible that one of those facts was
created in a chase round before the other, as otherwise the second created fact
could not have been created because of the first. Hence,
both facts were created in the same chase round.
So there was a chase round from $I''$ to $I'$ where we had $a \in
\appelem{I''}{R^p}$ and
both $F_1$ and $F_2$ were created respectively from active facts $F_1'$ and
$F_2'$ of $I''$ by $\uid$s
$\tau_1 : \ui{S^q}{R^p}$ and $\tau_2: \ui{T^r}{R^p}$.
But then, by property (*), $a$ occurs in only one fact, so as
it occurs in $F'_1$ and $F_2'$ we have $F_1' = F_2'$. Further, as $a \notin
\dom(I)$, $F'_1$ and $F'_2$ are not facts of~$I$ either, so by definition of the
$\uid$ chase and of the core chase, it is easy to see $a$ occurs at only one
position in $F'_1 = F'_2$.
This implies that $\tau_1 = \tau_2$. Hence, we must have $F_1 = F_2$, a
contradiction. This establishes the Unique Witness Property.

\end{document}